\documentclass[times,final]{elsarticle}
\usepackage{mathtools}

\usepackage{amssymb,amsmath,amsthm}
\usepackage{amsfonts}
\usepackage{mathrsfs}
\usepackage{graphicx,subfigure,epstopdf}
\graphicspath{{./figures_cv/}}
\usepackage{latexsym,float,tikz}
\usepackage[ruled,vlined,linesnumbered]{algorithm2e}
\usepackage{booktabs,tabularx}
\newcolumntype{L}{>{\raggedright\arraybackslash}X}
\usepackage{siunitx}
\usepackage{makecell}
\usepackage{threeparttable}
\RestyleAlgo{ruled}
\usepackage{adjustbox}
\usepackage{multirow}

\usepackage{verbatim}
\usepackage{enumerate}
\usepackage{pdflscape}
\usepackage{enumitem}
\usepackage{bbm}
\numberwithin{equation}{section}
\newtheorem{theorem}{Theorem}[section]

\newtheorem{proposition}{Proposition}[section]

\newtheorem{remark}[theorem]{Remark}

\usepackage{xcolor}
\usepackage[linktocpage]{hyperref}

\usepackage{jcomp}
\usepackage{framed,multirow}

\usepackage{amssymb}
\usepackage{latexsym}

\usepackage{url}
\usepackage{xcolor}
\definecolor{newcolor}{rgb}{.8,.349,.1}

\journal{Journal of Computational Physics}

\begin{document}

\verso{Yueqi Wang \textit{etal}}

\begin{frontmatter}

\title{A POD--DeepONet Framework for Forward and Inverse Design of 2D Photonic Crystals}%

\author[1]{Yueqi Wang}
\author[3]{Guanglian Li}
\author[1,2]{Guang Lin\corref{cor1}}
\cortext[cor1]{Corresponding author: Guang Lin, E-Mail: Guanglin@purdue.edu}

\address[1]{Department of Mathematics, Purdue University, 610 Purdue Mall, West Lafayette, 47907, IN, USA}
\address[2]{School of Mechanical Engineering, Purdue University, 610 Purdue Mall, West Lafayette, 47907, IN, USA}
\address[3]{Department of Mathematics, The University of Hong Kong, Pokfulam Road, Hong Kong SAR, P.R. China}


\begin{abstract}
We develop a reduced-order operator-learning framework for
forward and inverse band-structure design of two-dimensional photonic crystals
with binary, pixel-based $p4m$-symmetric unit cells.
We construct a POD--DeepONet surrogate for the discrete band map along the
standard high-symmetry path by coupling a POD trunk extracted from
high-fidelity finite-element band snapshots with a neural branch network that
predicts reduced coefficients. This architecture yields a compact and
differentiable forward model that is tailored to the underlying Bloch
eigenvalue discretization. We establish continuity
of the discrete band map on the relaxed design space and prove a uniform
approximation property of the POD--DeepONet surrogate,
leading to a natural decomposition of the total surrogate error into POD
truncation and network approximation contributions. Building on this forward
surrogate, we formulate two end-to-end neural inverse design procedures,
namely dispersion-to-structure and band-gap inverse design, with training
objectives that combine data misfit, binarity promotion, and supervised
regularization to address the intrinsic non-uniqueness of the inverse
mapping and to enable stable gradient-based optimization in the relaxed space.
Our numerical results show that the proposed framework achieves accurate
forward predictions and produces effective inverse designs on practical
high-contrast, pixel-based photonic layouts.
\end{abstract}


\end{frontmatter}


\section{Introduction}

Photonic crystals (PhCs) are periodically structured dielectric composites that manipulate the propagation of electromagnetic waves. Since the seminal contributions of Yablonovitch \cite{yablonovitch1987inhibited} on inhibited spontaneous emission and of John \cite{john1987strong} on light localization in photonic band-gap materials, PhCs have become a key platform for controlling light in integrated photonics, enabling waveguides, cavities, filters, and lasers with tailored dispersion properties \cite{johnson2001new,joannopoulos2008molding,notomi2010manipulating,painter1999two}. In the frequency-domain setting, their spectral behavior is characterized by the dispersion relation, that is, the family of band functions $\{\widetilde\omega_n(\mathbf k)\}_{n\geq 1}$ arising from the parameterized Helmholtz eigenvalue problem associated with Maxwell equations. The location and width of photonic band gaps in these dispersion diagrams determine essential features such as slow-light transport, confinement, and frequency selectivity \cite{joannopoulos2008molding,notomi2010manipulating}.

From a computational viewpoint, the dispersion relation is obtained by solving a parameterized self-adjoint complex-valued Helmholtz eigenvalue problem with periodic coefficients \cite{kuchment1993floquet}. Standard numerical approaches include the plane wave expansion method \cite{ho1990existence}, finite-difference time-domain (FDTD) schemes \cite{taflove2005computational,qiu2000numerical}, and finite element method (FEM) \cite{axmann1999efficient,andonegui2013finite}. Computing the band functions, therefore, requires solving this Helmholtz eigenvalue problem at a large number of wave vectors throughout the Brillouin zone. To reduce this cost, most band-structure calculations focus on canonical high-symmetry paths, i.e., piecewise linear segments connecting high-symmetry points on the boundary of the irreducible Brillouin zone, which already reveal band gaps and other salient spectral features in a wide range of PhC applications \cite{joannopoulos2008molding,setyawan2010high,degirmenci2013finite}. Even in this reduced setting, however, high-fidelity band structures demand dense sampling of $\mathbf k$ along these paths, so that for a given PhC structure, one must still solve a large number of generalized Helmholtz eigenvalue problems on fine meshes. This cost becomes particularly severe in many-query settings where band structures must be recomputed for thousands of candidate unit cells in high-throughput materials screening and database construction \cite{setyawan2010high,li2021computation,cersonsky2021diversity,wang2025hp}.

In practical PhC design, forward band-structure evaluation is only half the story, and a central objective is often inverse design. Depending on the application, one may either seek a periodic microstructure whose band structure matches, as closely as possible, a prescribed dispersion relation (dispersion-to-structure problems), or optimize the geometry to maximize a complete band gap within a target frequency range between selected bands (band-gap design). Classical approaches formulate these tasks as PDE-constrained optimization problems and solve them using topology optimization, typically via gradient-based methods \cite{sigmund2003systematic,men2014robust,dalklint2022tunable} or non-gradient schemes \cite{zhang2021realization,jia2024maximizing}, together with related level-set formulations for band-gap maximization \cite{kao2005maximizing,cheng2013maximizing}.
Starting from heuristic initial guesses, these algorithms repeatedly solve large-scale eigenproblems at many $\mathbf k$-points to evaluate objectives and sensitivities, so that each design update is computationally expensive. As a result, high-dimensional pixel or level-set design spaces incur substantial computational cost, which in turn motivates the development of reduced-order and data-driven, in particular deep-learning-based, surrogates that approximate the forward Bloch map with much lower online complexity while retaining sufficient accuracy for inverse design.

Deep learning has emerged as a powerful paradigm for modeling structure–property maps in metamaterials and PhCs in the past few years. Current approaches include using fully connected deep neural networks, ResNet-style architectures or convolutional neural networks (CNNs) to map geometric or pixelized unit cells to scattering spectra or transmission responses, and then utilizing either backpropagating through the forward network or training a separate network that maps target responses to designs for the purpose of inverse design \cite{peurifoy2018nanophotonic,tahersima2019deep,qiu2021nanophotonic,jiang2022dispersion}. 
An auto-encoder is trained to extract the topological features from sample images of unit cells, and then a multilayer perceptron (MLP) is trained to establish the inherent relation between band gaps and topological features \cite{li2020designing}. Later works employ variational or conditional autoencoders (VAE / cVAE) and tandem networks to generate unit cells from band-gap targets \cite{han2022deep,wang2024predicting,ma2023deep,wan2025deep,tran2025deep}. We refer to 
\cite{song2024artificial,deng2024inverse,chen2022see} for recent reviews on CNN-based, MLP-based and generative-model-based approaches for meta-structure and PhC design.
Most existing surrogates still treat the band structure as a high-dimensional black-box vector: the concatenated list of band frequencies are first evaluated at sampled $\mathbf k$-points on a fixed grid in the Brillouin zone, then the neural networks regress these sampled values directly. 
In parallel with these developments, operator learning has emerged as a powerful framework for data-driven surrogate modeling of parametric PDEs. Rather than approximating finite-dimensional input–output maps, Deep Operator Networks (DeepONets) approximate nonlinear operators that take function-valued inputs such as coefficients, source terms, or boundary conditions and return solution fields~\cite{lu2021learning,kovachki2023neural}. Under suitable assumptions, DeepONets satisfy universal approximation theorems for operators \cite{lu2021learning}. Once trained, they evaluate the learned operator on new inputs at negligible cost while remaining differentiable with respect to those inputs, which is a property particularly attractive for gradient-based inverse problems.
To improve efficiency and robustness in high-dimensional settings, several works have combined DeepONets with projection-based or multi-fidelity model reduction, including POD-augmented DeepONet frameworks \cite{eivazi2024nonlinear,cheng2025surrogate,wang2025reduced} and multi-fidelity DeepONet architectures for residual learning \cite{demo2023deeponet,lu2022multifidelity}.
These approaches exploit the empirical observation that, in many parametric PDEs, the solution manifold lies close to a low-dimensional subspace. A reduced basis, for instance obtained by proper orthogonal decomposition (POD), can then be used to represent the output field, while a neural network learns how the corresponding reduced coefficients depend on the input.

To the best of our knowledge, reduced-order operator-learning surrogates have not yet been applied to Bloch eigenvalue problems or band-structure computations in photonic crystals. Existing successes in other parametric PDE settings indicate that this approach can both accelerate band-structure evaluations and impose a transparent reduced-order structure on the models.
In this work, we study one forward problem and two inverse problems: 
(i) prediction of band functions for a given photonic structure, 
(ii) dispersion-to-structure design targeting prescribed band functions, and 
(iii) band-gap design targeting prescribed gap descriptors.
Our forward problem can be written as an operator
\begin{equation}\label{eq:operator forward}
    \mathcal{G}:\ \epsilon \ \longmapsto\ \bigl(\widetilde\omega_1(\cdot;\epsilon),\dots,\widetilde\omega_{N_b}(\cdot;\epsilon)\bigr),
\end{equation}
where $\epsilon$ denotes the periodic dielectric permittivity distribution of the photonic crystal and $\mathcal{G}$ returns the first $N_b$ band functions $\widetilde\omega_n(\mathbf k;\epsilon)$ along the high-symmetry path in the Brillouin zone. 
The first inverse problem, dispersion-to-structure design, seeks to approximately invert this operator,
\begin{equation}\label{eq:operator inverse1}
    \mathcal{I}_{\mathrm{disp}}:\ \bigl(\widetilde\omega_1^\ast(\cdot),\dots,\widetilde\omega_{N_b}^\ast(\cdot)\bigr)\ \longmapsto\ \epsilon,
\end{equation}
given target band functions $\widetilde\omega_n^\ast(\mathbf k)$. The second inverse problem considers a band-gap descriptor $\mathbf g=(a,b,p)$ and is written as
\begin{equation}\label{eq:operator inverse2}
    \mathcal{I}_{\mathrm{gap}}:\ \mathbf g^\ast \ \longmapsto\ \epsilon,
\end{equation}
where $\mathbf g^\ast=(a^\ast,b^\ast,p^\ast)$ specifies a target band gap, i.e., $a^\ast$ and $b^\ast$ are the lower and upper gap edges, and $p^\ast$ is the index of the band below the gap.

We introduce a pixelized representation of the unit cell, where the dielectric distribution is encoded as piecewise-constant values on a fixed mesh grid. On this discrete design space, we approximate the continuous operators in \eqref{eq:operator forward}–\eqref{eq:operator inverse2}. In particular, we approximate the forward operator $\mathcal{G}$ by a POD--DeepONet surrogate. Each band function $\widetilde\omega_n(\cdot;\epsilon)$ is expanded in a POD basis over the sampled wave vectors. The trunk network evaluates this fixed POD basis at the chosen $\mathbf k$-points, while the branch network maps the pixelized microstructure to the corresponding POD coefficients. Both inverse operators are then realized by neural networks trained on top of the learned POD--DeepONet surrogate. 
Our main contributions are summarized as follows:
\begin{itemize}
\item \textbf{Unified operator-learning formulation.}
We present a unified operator-learning framework for 2D photonic-crystal band-structure prediction and two inverse-design tasks.

\item \textbf{Error analysis.}
We derive an explicit decomposition of the total surrogate error for the learned band functions into POD truncation and network approximation components.

\item \textbf{Efficient end-to-end inverse design with minimal FEM calls.}
We enable fast dispersion matching and band-gap targeting by backpropagating through the differentiable surrogate,
substantially reducing the need for repeated FEM eigen-solves.
\end{itemize}



The remainder of the paper is organized as follows.
Section~\ref{sec:Problem formulation} formulates the band-structure problem for photonic crystals.
We set up the finite-element framework and introduce a pixel-based parametrization of unit cells, which leads to precise formulations of the forward and inverse problems studied in this work.
Section~\ref{sec:POD-DeepONet} develops the POD--DeepONet surrogate for the forward band map, namely the mapping from a pixel design vector to sampled band functions along the high-symmetry path.
We describe the snapshot POD construction of the trunk and the supervised training of the branch network. Then, we analyze approximation errors so as to distinguish POD truncation from neural-network contributions.
Section~\ref{sec:Inverse design based on the POD–DeepONet forward map} turns to inverse design by coupling the forward surrogate with inverse networks for dispersion-to-structure and band-gap problems.
Section~\ref{sec:numerics} presents numerical experiments that demonstrate the accuracy and efficiency of the proposed approach, and Section~\ref{sec:conclusion} concludes with a brief summary and outlook.

\section{Problem formulation}\label{sec:Problem formulation}
In this section, we briefly review the mathematical framework for computing dispersion relations and formulate the forward and inverse band-structure problems that constitute the main focus of this work.

\subsection{Dispersion relation calculation}\label{sec:Dispersion relation calculation}

To fix notation, we briefly recall the standard reduction from the time-harmonic Maxwell system to a parameterized Helmholtz eigenvalue problem governing Bloch modes in two-dimensional (2D) photonic crystals. The derivation follows our previous work~\cite{wang2023dispersion, wang2023hp}; see, e.g.,~\cite{bao2001mathematical,jackson1999classical} for background.

In the SI convention, the time-harmonic Maxwell equations for linear, non-dispersive, non-magnetic media with free charges and currents read
\begin{subequations}
  \begin{align}
    \nabla\times\mathbf{E}(\mathbf{x})-i\omega\mu_{0}\mathbf{H}(\mathbf{x}) &= 0, \label{harm1} \\
    \nabla\times\mathbf{H}(\mathbf{x})+i\omega\epsilon_{0}\epsilon(\mathbf{x})\mathbf{E}(\mathbf{x}) &= 0, \label{harm2} \\
    \nabla\cdot\bigl(\epsilon(\mathbf{x})\mathbf{E}(\mathbf{x})\bigr) &= 0, \label{harm3} \\
    \nabla\cdot\mathbf{H}(\mathbf{x}) &= 0, \label{harm4}
  \end{align}
\end{subequations}
where $\mathbf{x}\in\mathbb{R}^3$, $\mathbf{E}$ and $\mathbf{H}$ denote the electric and magnetic fields, $\omega\ge0$ is the angular frequency, $\mu_0$ and $\epsilon_0$ are the vacuum permeability and permittivity, and $\epsilon\in L^{\infty}(\mathbb{R}^3;\mathbb{R}^+)$ is the relative permittivity. Eliminating either $\mathbf{E}$ or $\mathbf{H}$ from~\eqref{harm1}–\eqref{harm2} leads to the familiar curl–curl formulations
\begin{equation}\label{E}
  \nabla\times\bigl(\nabla\times\mathbf{E}(\mathbf{x})\bigr)
  - (\omega c^{-1})^2\epsilon(\mathbf{x})\mathbf{E}(\mathbf{x}) = 0,
\end{equation}
and
\begin{equation}\label{H}
  \nabla\times\bigl(\epsilon(\mathbf{x})^{-1}\nabla\times\mathbf{H}(\mathbf{x})\bigr)
  - (\omega c^{-1})^2\mathbf{H}(\mathbf{x}) = 0,
\end{equation}
with $\epsilon_{0}\mu_{0}=c^{-2}$, where $c$ is the speed of light.

We restrict attention to 2D photonic crystals, which are invariant in the $z$-direction and periodic in the $x$–$y$ plane. Accordingly, $\epsilon(\mathbf{x})$ is taken to be independent of $z$. Under this assumption, the fields can be decomposed into a transverse electric (TE) polarization with $H_1 = H_2 = E_3 = 0$ and a transverse magnetic (TM) polarization with $E_1 = E_2 = H_3 = 0$. Inserting these ansatzes into~\eqref{E}–\eqref{H} yields scalar Helmholtz eigenvalue problems on $\mathbb{R}^2$,
\begin{align}
  -\nabla \cdot \bigl(\epsilon(\mathbf{x})^{-1}\nabla H(\mathbf{x})\bigr)
  - (\omega c^{-1})^2 H(\mathbf{x}) &= 0,
  \quad \mathbf{x}\in\mathbb{R}^2,
  &&\textbf{(TE mode)}, \label{TE} \\
  -\Delta E(\mathbf{x})
  - (\omega c^{-1})^2 \epsilon(\mathbf{x})E(\mathbf{x}) &= 0,
  \quad \mathbf{x}\in\mathbb{R}^2,
  &&\textbf{(TM mode)}. \label{TM}
\end{align}

The 2D photonic crystal exhibits discrete translational symmetry in the $x$–$y$ plane~\cite{joannopoulos2008molding}, so the relative permittivity satisfies
\[
  \epsilon(\mathbf{x}+c_1\mathbf{a}_1+c_2\mathbf{a}_2)
  = \epsilon(\mathbf{x}),
  \quad \forall\,\mathbf{x}\in\mathbb{R}^2,\; c_1,c_2\in\mathbb{Z},
\]
where the primitive lattice vectors $\mathbf{a}_1,\mathbf{a}_2$ span a fundamental periodicity domain $\Omega$ (the unit cell). The reciprocal lattice vectors $\mathbf{b}_1,\mathbf{b}_2$ are defined by
\begin{equation}
  \mathbf{b}_i\cdot\mathbf{a}_j = 2\pi\delta_{ij}, \quad i,j=1,2,
\end{equation}
and generate the reciprocal lattice. Its elementary cell is the (first) Brillouin zone $\mathcal{B}_F$, which can be characterized as the set of points in reciprocal space closer to the origin than to any other reciprocal lattice point. Throughout this work, we consider a square lattice with primitive vectors $\mathbf{a}_i = a \mathbf e_i$ for $i=1,2$, where $(\mathbf e_i)_{i=1,2}$ is the canonical basis of $\mathbb{R}^2$ and $a>0$ is the lattice constant. The corresponding reciprocal vectors are $\mathbf{b}_i = \frac{2\pi}{a}\mathbf e_i$.

By Bloch's theorem~\cite{kittel2018introduction}, solutions of~\eqref{TE}–\eqref{TM} can be written in the form
\[
  \Psi(\mathbf{x}) = e^{i\mathbf{k}\cdot\mathbf{x}}u(\mathbf{x}),
\]
where the wave vector $\mathbf{k}$ lies in $\mathcal{B}_F$ and the Bloch factor $u$ is periodic with respect to the lattice. In particular, the scalar fields in~\eqref{TE}–\eqref{TM} admit representations
\[
  H(\mathbf{x}) = e^{i\mathbf{k}\cdot\mathbf{x}}u_1(\mathbf{x}),
  \qquad
  E(\mathbf{x}) = e^{i\mathbf{k}\cdot\mathbf{x}}u_2(\mathbf{x}),
\]
with periodic functions $u_1,u_2$ defined on the unit cell $\Omega$. Substituting these into~\eqref{TE}–\eqref{TM} yields parameterized Helmholtz eigenvalue problems on $\Omega$,
\begin{subequations}
  \begin{align}
    -(\nabla+i\mathbf{k})\cdot\bigl(\epsilon(\mathbf{x})^{-1}(\nabla+i\mathbf{k})u_1(\mathbf{x})\bigr)
    - (\omega c^{-1})^{2} u_1(\mathbf{x}) &= 0,
    \quad \mathbf{x}\in\Omega,
    &&\textbf{(TE mode)}, \label{TE2} \\
    -(\nabla+i\mathbf{k})\cdot\bigl((\nabla+i\mathbf{k})u_2(\mathbf{x})\bigr)
    - (\omega c^{-1})^{2}\epsilon(\mathbf{x})u_2(\mathbf{x}) &= 0,
    \quad \mathbf{x}\in\Omega,
    &&\textbf{(TM mode)}, \label{TM2}
  \end{align}
\end{subequations}
where $\mathbf{k}$ varies in the Brillouin zone and the Bloch factors $u_i$ satisfy periodic boundary conditions $u_i(\mathbf{x}) = u_i(\mathbf{x}+\mathbf{a}_j)$ for $i,j=1,2$. If $\epsilon$ has additional point-group symmetries (e.g., mirror symmetry), the wave vector can be further restricted to the irreducible Brillouin zone (IBZ), denoted by $\mathcal{B}\subset\mathcal{B}_F$. An example of a square lattice and its Brillouin zone is shown in Figure~\ref{lattice}.

\begin{figure}[hbt!]
  \centering
  \subfigure{\label{lattice(1)}
    \includegraphics[width = .30\textwidth,trim={1cm 1.2cm 2cm 0.5cm},clip]{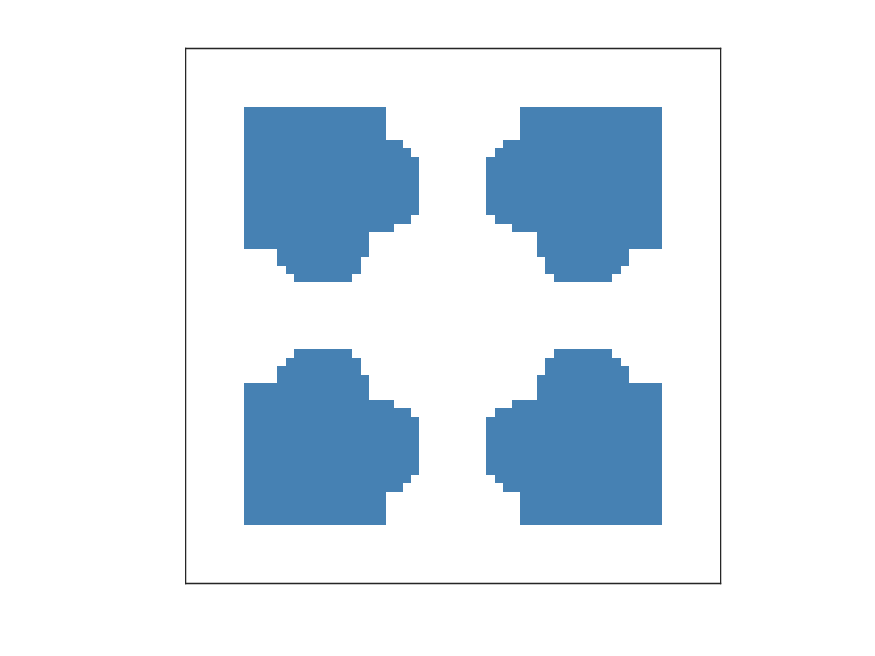}
  }%
  \subfigure{\label{lattice(2)}
    \includegraphics[width = .33\textwidth,trim={15cm 6cm 14cm 4.5cm},clip]{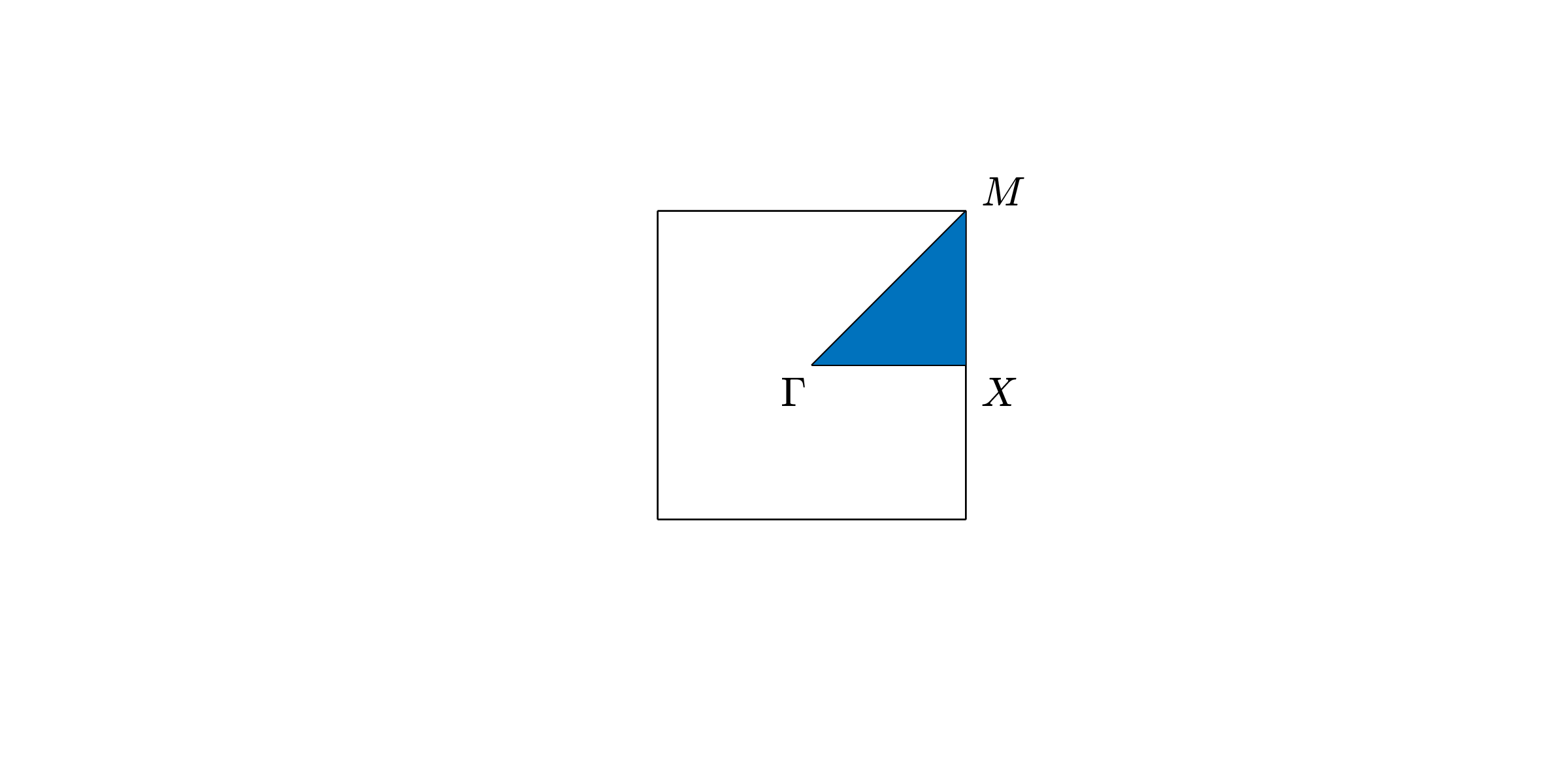}
  }%
  \caption{Illustration of a square-lattice unit cell $\Omega$ (left) and the corresponding first Brillouin zone $\mathcal{B}_F$ (right). In $\Omega$, blue denotes alumina with permittivity $8.9$ and white denotes air with permittivity $1$; in $\mathcal{B}_F$, the IBZ $\mathcal{B}$ is the shaded triangle with vertices $\Gamma=(0,0)$, $X = (\pi/a,0)$, and $M = (\pi/a,\pi/a)$.}
  \label{lattice}
\end{figure}

Both parameterized Helmholtz problems~\eqref{TE2}–\eqref{TM2} can be expressed as the unified form
\begin{equation}\label{both}
  -(\nabla+i\mathbf{k})\cdot \alpha(\mathbf{x})(\nabla+i\mathbf{k})u(\mathbf{x})
  - \lambda\beta(\mathbf{x})u(\mathbf{x}) = 0,
  \quad \mathbf{x}\in\Omega
\end{equation}
with $\Omega\subset\mathbb{R}^2$, $\mathbf{k} \in \mathcal{B}$, and $\lambda =(\omega c^{-1})^2$. In the TE mode, $u$ represents the magnetic field $H$ in $z$-direction and the coefficients are
$\alpha(\mathbf{x}) := \epsilon(\mathbf{x})^{-1}$ and $\beta(\mathbf{x}) := 1$; in the TM mode, $u$ represents the electric field $E$ in $z$-direction and
$\alpha(\mathbf{x}) := 1$ and $\beta(\mathbf{x}) := \epsilon(\mathbf{x})$.

The variational formulation of~\eqref{both} reads as follows: for each $\mathbf{k}\in\mathcal{B}$, find a non-trivial eigenpair $(\lambda,u)\in\mathbb{R}\times H^1_\pi(\Omega)$ such that
\begin{equation}\label{variational}
\left\{
\begin{aligned}
  \int_{\Omega}\alpha(\mathbf{x})(\nabla+i\mathbf{k})u\cdot(\nabla-i\mathbf{k})\overline{v}\,\mathrm{d}\mathbf{x}
  - \lambda \int_{\Omega}\beta(\mathbf{x}) u\overline{v}\,\mathrm{d}\mathbf{x} &= 0,
  &&\forall\, v\in H^1_\pi(\Omega),\\[2mm]
  \|u\|_{L^2_\beta(\Omega)} &= 1. &&
\end{aligned}
\right.
\end{equation}
Where, 
\[
  L^2_\beta(\Omega)
  := \biggl\{ f\in L^2(\Omega;\mathbb{C})\;:\;
  \|f\|_{L^2_\beta(\Omega)}^2
  := \int_\Omega \beta(\mathbf{x})|f(\mathbf{x})|^2\,\mathrm{d}\mathbf{x}
  <\infty \biggr\},
\]
and the periodic Sobolev space $H^1_\pi(\Omega)\subset H^1(\Omega;\mathbb{C})$ is defined by
\[
  H^1_\pi(\Omega)
  := \bigl\{ v\in H^1(\Omega;\mathbb{C})\;:\;
  v \text{ is periodic with respect to the lattice vectors of }\Omega \bigr\}.
\]
Here, $H^1(\Omega;\mathbb{C})$ denotes the Sobolev space of square integrable complex-valued functions with square integrable weak gradient, equipped with the standard $H^1$–norm.

The following result is a standard consequence of the spectral theory of second-order elliptic operators with compact resolvent~\cite{glazman1965direct}.

\begin{theorem}\label{thm:spectrum}
For every wave vector $\mathbf{k}\in \mathcal{B}$, the variational eigenproblem~\eqref{variational} defines a self-adjoint operator on $H^1_{\pi}(\Omega)$ with compact resolvent. Its spectrum is purely discrete and non-negative, and the eigenvalues can be arranged in a non-decreasing sequence (repeated according to finite multiplicities),
\[
  0 \le \lambda_1(\mathbf{k}) \le \lambda_2(\mathbf{k}) \le \cdots \le \lambda_n(\mathbf{k}) \le \cdots \to +\infty.
\]
Moreover, for each fixed $n\in\mathbb{N}$, the function
\(
  \mathbf{k}\mapsto \lambda_n(\mathbf{k})
\)
is continuous on $\mathcal{B}$, and $\lambda_n(\mathbf{k})\to+\infty$ as $n\to\infty$ uniformly in $\mathcal{B}$.
\end{theorem}

The Bloch frequencies are related to the eigenvalues by
\[
  \omega_{n}(\mathbf{k}) = c\,\sqrt{\lambda_{n}(\mathbf{k})},
\]
and we introduce the dimensionless (normalized) band functions
\[
  \widetilde\omega_{n}(\mathbf{k})
  := \frac{a}{2\pi c}\,\omega_{n}(\mathbf{k}).
\]
The family $\{\widetilde\omega_{n}(\cdot)\}_{n\ge1}$ is referred to as the dispersion relation of the photonic crystal. A photonic band gap is an open interval $(\omega_-,\omega_+)$ such that
\[
  \widetilde\omega_{n}(\mathbf{k}) \notin (\omega_-,\omega_+)
  \qquad\text{for all }\mathbf{k}\in\mathcal{B}\text{ and all }n\ge1,
\]
that is, there are no Bloch eigenmodes $(\mathbf{k},\widetilde\omega)$ in this frequency range. The dispersion relation therefore determines all admissible Bloch modes and, in particular, the location and width of photonic band gaps~\cite{joannopoulos1997photonic,johnson2001photonic}.

In most applications, one is primarily interested in the lowest few band functions, since photonic devices such as waveguides, cavities, and filters typically operate at relatively low frequencies~\cite{johnson1999guided,painter1999two,akahane2003high}, and the widest and most practically useful band gaps tend to occur between the first several bands~\cite{joannopoulos1997photonic,johnson2001photonic}. In addition, practical band-gap engineering often relies on high-contrast permittivity distributions, since a larger refractive-index contrast enhances Bragg scattering and generally leads to stronger field confinement and wider band gaps~\cite{joannopoulos1997photonic,johnson2001photonic}. Moreover, exploiting the point–group symmetries of the lattice, it is standard practice to restrict the wave vector $\mathbf{k}$ to the irreducible Brillouin zone and to plot $\widetilde\omega_n(\mathbf{k})$ only along a piecewise linear path connecting high–symmetry points on its boundary (for the square lattice in Figure~\ref{lattice}, the path $\Gamma\to X\to M\to\Gamma$), which greatly reduces the number of sampled $\mathbf{k}$–points while still capturing band edges, gaps, and other critical spectral features~\cite{kittel2018introduction,kuchment1993floquet,bouckaert1936theory,setyawan2010high}.

Let $\mathcal{K}_{\rm hs}\subset\mathcal{B}$ denote this standard high-symmetry path on the boundary of the irreducible Brillouin zone for the square lattice, defined by
\[
  \mathcal{K}_{\rm hs}
  := \overline{\Gamma X}\,\cup\,\overline{X M}\,\cup\,\overline{M\Gamma}, 
\]
then our focus is on the computation of the first $N_b$ band functions
$\{\widetilde\omega_n(\mathbf{k})\}_{n=1}^{N_b}$ along the high-symmetry path
$\mathcal{K}_{\rm hs}\subset\mathcal{B}$. An illustrative example is shown in Figure~\ref{fig:disp-example}, where
\ref{fig:cell} displays a unit cell and \ref{fig:band} shows the corresponding normalized band functions $\widetilde\omega_n(\mathbf{k})$ along the high–symmetry path $\mathcal{K}_{\rm hs}$.

\begin{figure}[hbt!]
	\centering
	\subfigure[Unit cell]{\label{fig:cell}
		\includegraphics[width = .35\textwidth,trim={1cm 1cm 1cm 0.5cm},clip]{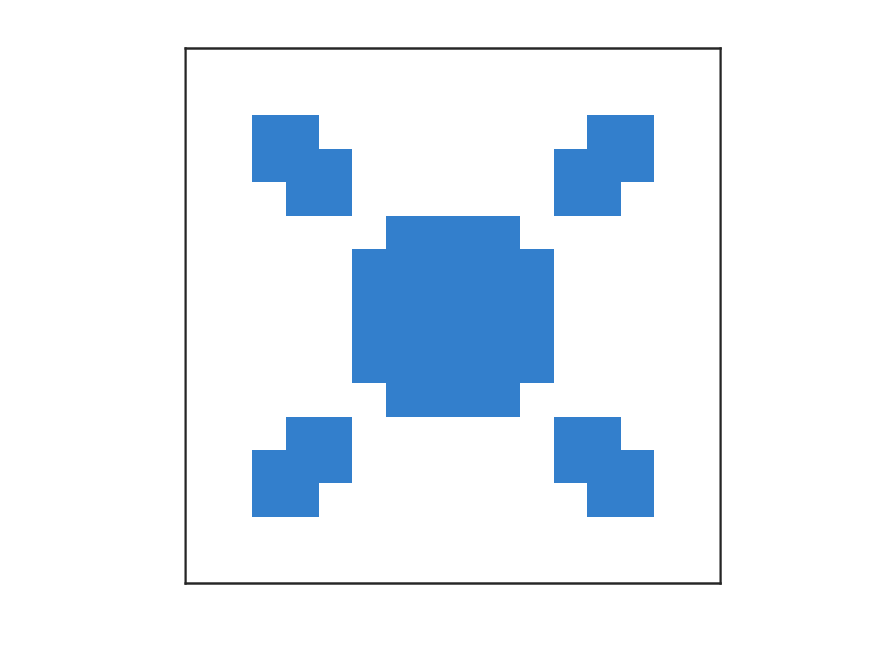}
	}%
    \subfigure[Band structure]{\label{fig:band}
		\includegraphics[width = .35\textwidth,trim={1cm 1cm 1cm 0.5cm},clip]{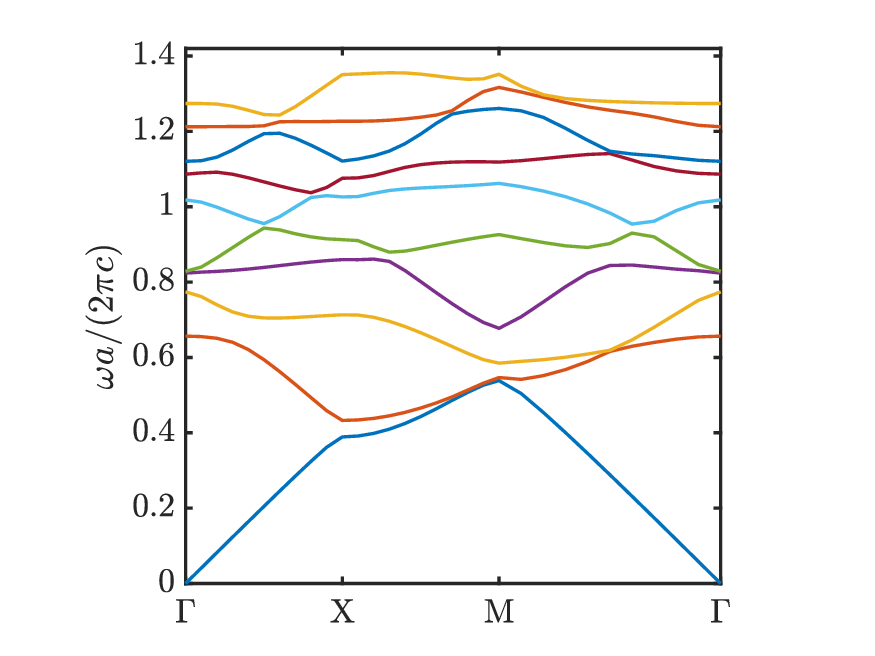}
	}%
	\centering
    \caption{Example of a unit cell and its band structure:
(a) the $16\times16$ unit cell; (b) the first $10$ TE band functions along the
high--symmetry path $\mathcal{K}_{\rm hs}$.}\label{fig:disp-example}
\end{figure}

\subsection{Finite element discretization}\label{sec:FEM}

By introducing the sesquilinear forms
\begin{equation*}\label{eq:sesqui-a}
     a(u,v)
  := \int_{\Omega}\alpha(\mathbf{x})(\nabla+i\mathbf{k})u\cdot(\nabla-i\mathbf{k})\overline{v}\,\mathrm{d}x,
  \qquad
  b(u,v)
  := \int_{\Omega}\beta(\mathbf{x})\, u\,\overline{v}\,\mathrm{d}x,
\end{equation*}
we can rewrite the problem \eqref{variational} in the compact form: for each $\mathbf{k}\in\mathcal{K}_{\rm hs}$, find a non-trivial eigenpair $(\lambda,u)\in\mathbb{R}\times H^1_\pi(\Omega)$ such that
\begin{equation}\label{eq:simply}
\left\{
\begin{aligned}
  a(u,v) &= \lambda\, b(u,v), &&\forall\, v\in H^1_\pi(\Omega),\\[1mm]
  \|u\|_{L^2_\beta(\Omega)} &= 1. &&
\end{aligned}
\right.
\end{equation}

Let $\mathcal T_h$ be a shape-regular, conforming triangulation of the unit cell $\Omega$ that is
periodic across opposite faces. We consider the conforming, periodic $P_1$ space
\[
V_h \;:=\;\bigl\{v_h\in H^1_\pi(\Omega):\; v_h|_T\in\mathbb P_1(T)\ \ \forall\,T\in\mathcal T_h \bigr\},
\]
and denote by $\{\phi_i\}_{i=1}^{N_h}\subset V_h$ its nodal basis.

The discrete eigenproblem
reads: for each $\mathbf{k}\in\mathcal{K}_{\rm hs}$, find non-trivial eigenpair $(\lambda_h,u_h)\in(\mathbb R,V_h)$ such that
\begin{equation}\label{eq:discEVP}
	\left\{
	\begin{aligned}
		a(u_h,v_h)&=\lambda_h\, b(u_h,v_h),\quad \forall\, v_h\in V_h \\
		b(u_h,u_h)&=1.
	\end{aligned}
	\right.
\end{equation}
Writing $u_h=\sum_{j=1}^{N_h} U_j \phi_j$ and testing \eqref{eq:discEVP} with $v_h=\phi_i$ yields
the generalized Hermitian matrix eigenproblem
\begin{equation}\label{eq:matrixEVP}
\mathbf A(\mathbf{k})\,\mathbf U \;=\; \lambda_h\,\mathbf B\,\mathbf U,
\qquad \mathbf U^* \mathbf B \mathbf U = 1,
\end{equation}
where $\mathbf U=(U_1,\dots,U_{N_h})^\top\in\mathbb C^{N_h}$, and
\begin{align*}
\mathbf A_{ij}(\mathbf{k})
&:= \int_\Omega \alpha(\mathbf x)\,(\nabla\phi_i + i \mathbf{k}\,\phi_i)\cdot(\nabla\phi_j - i \mathbf{k}\,\phi_j)\,{\rm d}\mathbf x,
\\
\mathbf B_{ij}
&:= \int_\Omega \beta(\mathbf x)\,\phi_i\,\phi_j\,{\rm d}\mathbf x.
\end{align*}
Solving \eqref{eq:matrixEVP} for a given wave vector $\mathbf k$ yields
eigenpairs $\{(\lambda_{h,n}(\mathbf k),\mathbf U_n(\mathbf k))\}_{n\ge1}$,
ordered nondecreasingly.  We define the normalized band functions by
\[
  \widetilde\omega_{h,n}(\mathbf k)
  := \frac{a}{2\pi c}\,\omega_{h,n}(\mathbf k)
  = \frac{a}{2\pi}\sqrt{\lambda_{h,n}(\mathbf k)}.
\]
The band structure is then approximated by evaluating
$\widetilde\omega_{h,n}$ along the high–symmetry path
$\mathcal K_{\rm hs}$.

\subsection{Pixel-based parametrization of unit-cell permittivity}
\label{sec:design-space}

Next, we introduce the binary, piecewise-constant unit-cell permittivities with $p4m$ plane symmetry \cite{schattschneider1978plane}. Because of this symmetry, it suffices to prescribe the material distribution in a symmetry-reduced subregion of the unit cell. We refer to the stair-shaped triangular region in the upper-right panel of Figure~\ref{generate_UnitCell} as the irreducible symmetry wedge, and we will simply call it the wedge in what follows.

Let $N_f$ denote the number of pixels in this wedge. We introduce a binary
design vector
\[
  \boldsymbol{\rho} = (\rho_1,\dots,\rho_{N_f})^\top \in \{0,1\}^{N_f},
\]
where $\rho_j = 1$ indicates that pixel $j$ is filled with the
high-permittivity material and $\rho_j = 0$ corresponds to the
low-permittivity background.
For later use, we also define a lifting operator
\begin{equation}\label{eq:lift-op}
  E^{\rm lift} : [0,1]^{N_f} \longrightarrow [0,1]^{N_{\rm pix}},\qquad
  \boldsymbol{\rho} \longmapsto
  E^{\rm lift}(\boldsymbol{\rho})
  := \bigl(\tilde\rho_1,\dots,\tilde\rho_{N_{\rm pix}}\bigr)^\top,
\end{equation}
which maps the wedge design vector $\boldsymbol{\rho}$ to a pixel vector on
the full unit cell by applying the $p4m$ symmetry operations. Here,
$N_{\rm pix}$ denotes the number of pixels in the full unit cell.

Let $\{P_j^{\rm w}\}_{j=1}^{N_f}$ denote the pixel subdomains in the wedge.
Each $P_j^{\rm w}$ is one of the small square pixels shown in
Figure~\ref{generate_UnitCell}. Applying the $p4m$ rotations and reflections
to these wedge pixels yields a collection of full-unit-cell pixels
$\{P_\ell\}_{\ell=1}^{N_{\rm pix}}$ that forms a partition of $\Omega$.
The associated pixel indicator field is then
\[
  \chi_{\boldsymbol{\rho}}(\mathbf x)
  :=
  \sum_{\ell=1}^{N_{\rm pix}} \tilde\rho_\ell\,\mathbf{1}_{P_\ell}(\mathbf x),
  \qquad \mathbf x\in\Omega,
\]
where $\mathbf{1}_{P_\ell}$ is the characteristic function of $P_\ell$.
For binary designs, we have
$\tilde\rho_\ell\in\{0,1\}$ and hence
$\chi_{\boldsymbol{\rho}}(\mathbf x)\in\{0,1\}$ a.e.\ in $\Omega$, with
$\chi_{\boldsymbol{\rho}}=1$ in the high-permittivity regions and
$\chi_{\boldsymbol{\rho}}=0$ in the low-permittivity background.

In all numerical examples, we consider binary composites of the form
\[
  \epsilon(\mathbf x)\in\{\epsilon_{\rm air},\epsilon_{\rm alum}\},
  \qquad \mathbf x\in\Omega,
\]
where $\epsilon_{\rm air}=1$ and
$\epsilon_{\rm alum}=8.9$; see Figure~\ref{lattice(1)}.
Given a design $\boldsymbol{\rho}$, the associated relative permittivity field
on the full unit cell is
\begin{equation}\label{eq:eps-rho}
  \epsilon(\mathbf x;\boldsymbol{\rho})
  :=
  \epsilon_{\rm air}
  + \bigl(\epsilon_{\rm alum}-\epsilon_{\rm air}\bigr)\,
    \chi_{\boldsymbol{\rho}}(\mathbf x),
  \qquad \mathbf x\in\Omega.
\end{equation}
This defines the pixel-based admissible set
\[
  \mathcal{E}_{\rm pix}
  :=
  \bigl\{
    \epsilon(\cdot;\boldsymbol{\rho}) : \boldsymbol{\rho}\in\{0,1\}^{N_f}
  \bigr\}
  \subset L^\infty(\Omega).
\]
All training and validation samples in our numerical experiments are drawn
from $\mathcal{E}_{\rm pix}$.

For the purposes of analysis and gradient-based optimization, we also consider
a continuous relaxation of the design variables. In this relaxed setting, we
allow
\[
  \boldsymbol{\rho}\in[0,1]^{N_f},
\]
while keeping the definitions of $E^{\rm lift}$, $\chi_{\boldsymbol{\rho}}$
and $\epsilon(\cdot;\boldsymbol{\rho})$ as above. The corresponding relaxed
family of permittivities is
\[
  \mathcal{E}_{\rm rel}
  :=
  \bigl\{
    \epsilon(\cdot;\boldsymbol{\rho}) : \boldsymbol{\rho}\in[0,1]^{N_f}
  \bigr\}.
\]
Clearly, we have $\mathcal{E}_{\rm pix} \subset \mathcal{E}_{\rm rel}$.
In the inverse design formulations, we
optimize over this continuous space to enable gradient-based methods and then
project the resulting relaxed designs back to binary configurations.

\begin{figure}[htbp]
  \centering
  \includegraphics[width = .45\textwidth,trim={3cm 3cm 3cm 1.6cm},clip]{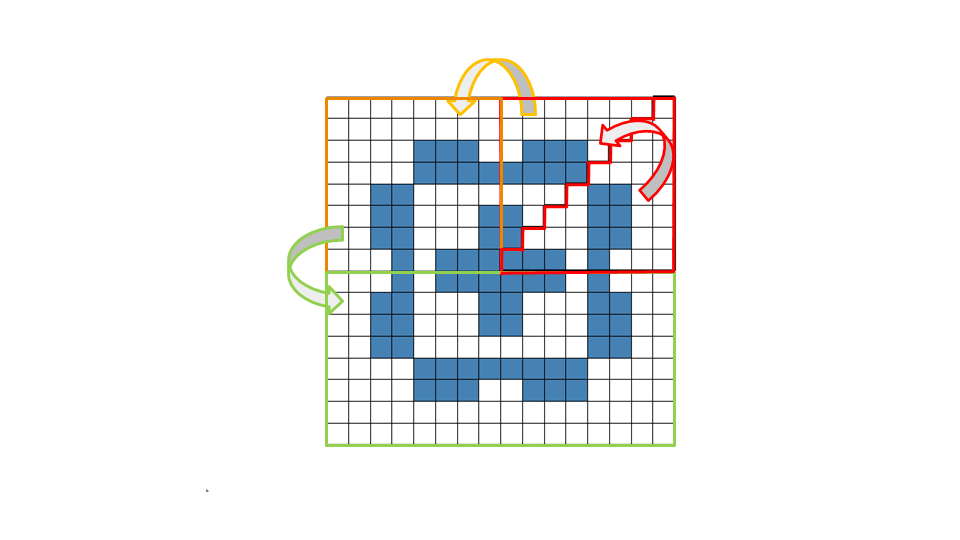}
  \caption{Pixel-based parametrization of a unit cell with $p4m$ plane symmetry.
  Blue pixels represent the high-permittivity material with dielectric constant $\epsilon_{\rm alum}$; white pixels represent the low-permittivity background $\epsilon_{\rm air}$.
  The red triangle marks the fundamental wedge.
  Its intersection with the pixel grid yields $N_f=36$ wedge pixels $\{P_j^{\rm w}\}_{j=1}^{N_f}$, each controlled by a binary design variable $\rho_j\in\{0,1\}$.
  Applying the $p4m$ rotations and reflections indicated by the coloured arrows maps these wedge pixels to the full set of pixels in the unit cell.}
  \label{generate_UnitCell}
\end{figure}

\begin{remark}[Generality of the pixel-based parametrization]
Any measurable two-phase permittivity layout in the unit cell with $p4m$ symmetry can be approximated by a sequence of binary, piecewise-constant pixel designs, so the design space $\mathcal{E}_{\rm pix}$ is rich enough for the inverse problems considered here and remains compatible with standard pixel-based topology-optimization techniques.
\end{remark}

\subsection{Forward and inverse band--structure problems}
\label{sec:forward-inverse}

We now formulate the forward and inverse band--structure problems in a
discrete setting and relate them to the operator mappings
\eqref{eq:operator forward}–\eqref{eq:operator inverse2}.

Through the pixel-based parametrization in
Section~\ref{sec:design-space}, each design vector
$\boldsymbol{\rho}\in[0,1]^{N_f}$ induces a permittivity field
$\epsilon(\cdot;\boldsymbol{\rho})\in\mathcal E_{\rm rel}$.
For any $\boldsymbol{\rho}\in[0,1]^{N_f}$ and any wave vector
$\mathbf k\in\mathcal K_{\rm hs}$, the discrete eigenproblem
\eqref{eq:discEVP} yields a nondecreasing sequence of eigenvalues
\[
  0 < \lambda_{h,1}(\mathbf k;\boldsymbol{\rho})
    \le \lambda_{h,2}(\mathbf k;\boldsymbol{\rho})
    \le \cdots,
\]
with associated eigenvectors $\mathbf U_{h,n}(\mathbf k;\boldsymbol{\rho})$.
The corresponding normalized band functions are
\[
  \widetilde\omega_{h,n}(\mathbf k;\boldsymbol{\rho})
  := \frac{a}{2\pi}\sqrt{\lambda_{h,n}(\mathbf k;\boldsymbol{\rho})},
  \qquad n\ge1.
\]
In practice, we sample the band structure on a finite set of wave vectors
\[
  \{\mathbf k_\ell\}_{\ell=1}^{N_k} \subset \mathcal K_{\rm hs},
\]
and retain only the first $N_b$ bands.  We then work with the truncated,
discretized band data
\[
  \mathbf W_h(\boldsymbol{\rho})
  :=
  \bigl(\widetilde\omega_{h,n}(\mathbf k_\ell;\boldsymbol{\rho})\bigr)
    _{\substack{1\le \ell\le N_k\\ 1\le n\le N_b}}
  \in \mathbb R^{N_k\times N_b}.
\]
This defines the discrete band map on the design space
\begin{equation}
  F_h^{\rm pix} : [0,1]^{N_f} \longrightarrow \mathbb R^{N_k\times N_b},
  \qquad
  F_h^{\rm pix}(\boldsymbol{\rho})
  := \mathbf W_h(\boldsymbol{\rho}),
  \label{eq:Fpix}
\end{equation}
which is a pixel-based realization
of the forward operator $\mathcal{G}$ in~\eqref{eq:operator forward}.

\subsubsection*{Forward problem}
The forward band--structure problem aims for computing the band functions along the high--symmetry path for a given design vector $\boldsymbol{\rho}\in[0,1]^{N_f}$, or equivalently a periodic permittivity field $\epsilon(\cdot;\boldsymbol{\rho})$.
Accurate and efficient evaluation of this map is essential for
characterizing dispersion relations of candidate unit cells and for
generating the band data used in the training and objective functions
of the inverse problems. In the discrete setting, the forward problem reduces to the evaluation of the map \eqref{eq:Fpix}.  A single call to $F_h^{\rm pix}$ requires
solving $N_k$ large-scale matrix eigenvalue problems of the
form~\eqref{eq:matrixEVP}.  

\subsubsection*{Inverse dispersion-to-structure problem}

In many photonic design tasks, one starts from a target dispersion
relation and wishes to ``invert'' it, that is, to recover a periodic
microstructure whose band diagram matches the target as closely as
possible.  This situation arises, for example, when a band diagram is
prescribed by physical considerations such as desired group
velocities or band-gap locations, or is computed at a higher level of
modeling, and one seeks a realizable unit cell that reproduces this
behavior \cite{sigmund2003systematic,notomi2010manipulating,castello2014inverse,wang2025hp}.
At the continuum level, it corresponds to the dispersion-to-structure
operator $\mathcal{I}_{\mathrm{disp}}$ in~\eqref{eq:operator inverse1},
which maps a set of target band functions to a suitable permittivity
distribution.

Let $\mathbf W_h^{\rm targ}\in\mathbb R^{N_k\times N_b}$ be a
given target band structure, representing discrete samples of the
desired band functions
$(\widetilde\omega_1^\ast(\cdot),\dots,\widetilde\omega_{N_b}^\ast(\cdot))$
in~\eqref{eq:operator inverse1}.  We seek a design
$\boldsymbol{\rho}\in[0,1]^{N_f}$ such that
$\mathbf W_h(\boldsymbol{\rho})$ is close to
$\mathbf W_h^{\rm targ}$.  At the discrete level, we consider
the optimization problem
\begin{equation}
  \min_{\boldsymbol{\rho}\in[0,1]^{N_f}}
  \; J_{\rm disp}(\boldsymbol{\rho})
  :=
  \frac{1}{N_k N_b}
  \bigl\|
    \mathbf W_h(\boldsymbol{\rho}) - \mathbf W_h^{\rm targ}
  \bigr\|_F^2,
  \label{eq:inv-disp}
\end{equation}
subject to additional constraints enforcing such as binary
material distributions.

\subsubsection*{Inverse band-gap design problem}

In many applications, such as in the design of waveguides and cavity
structures \cite{yablonovitch1987inhibited,ho1990existence,sigmund2003systematic,dobson1999maximizing,men2014robust}, the primary goal is not to reproduce an entire
dispersion relation, but to guarantee the existence of a complete
photonic band gap in a prescribed frequency range. 
In this setting, the design objective is the placement of a band gap.

Let $1\le p < N_b$ be a prescribed band index and let
$I_{\rm gap}=(a,b)\subset\mathbb R_+$ be a target frequency interval.
We seek designs $\boldsymbol{\rho}\in[0,1]^{N_f}$ such that the
$p$-th and $(p+1)$-st bands satisfy
\begin{equation*}
  \widetilde\omega_{h,p}(\mathbf k;\boldsymbol{\rho}) \le a,
  \qquad
  \widetilde\omega_{h,p+1}(\mathbf k;\boldsymbol{\rho}) \ge b,
  \qquad \forall\,\mathbf k\in\mathcal K_{\rm hs},
\end{equation*}
so that the target interval $I_{\rm gap}$ lies inside a complete band gap between bands $p$ and $p+1$ along the entire
high--symmetry path.  In the discrete setting, this requirement
becomes
\begin{equation}
  \widetilde\omega_{h,p}(\mathbf k_\ell;\boldsymbol{\rho}) \le a,
  \qquad
  \widetilde\omega_{h,p+1}(\mathbf k_\ell;\boldsymbol{\rho}) \ge b,
  \qquad \ell=1,\dots,N_k.
  \label{eq:gap-condition-disc}
\end{equation}
In terms of the band-gap descriptor $\mathbf g=(a,b,p)$ introduced
in~\eqref{eq:operator inverse2}, this inverse problem corresponds to
the operator $\mathcal{I}_{\mathrm{gap}}$ that maps the target descriptor
$\mathbf g$ to an admissible design.

We enforce conditions
\eqref{eq:gap-condition-disc} by minimizing a penalty functional
$J_{\rm gap}(\boldsymbol{\rho})$ that measures their violation.  An abstract form is
\begin{equation}\label{eq:Jgap-def}
  \min_{\boldsymbol{\rho}\in[0,1]^{N_f}}
  \; J_{\rm gap}(\boldsymbol{\rho})
  :=
  \frac{1}{N_k}\sum_{\ell=1}^{N_k}
  \Phi_{\rm gap}\!\bigl(
    \widetilde\omega_{h,p}(\mathbf k_\ell;\boldsymbol{\rho}),
    \widetilde\omega_{h,p+1}(\mathbf k_\ell;\boldsymbol{\rho}),a,b
  \bigr),
\end{equation}
where $\Phi_{\rm gap}:\mathbb R^4\to\mathbb R_+$ is a continuous
penalty function such that
\[
  \Phi_{\rm gap}(x,y,a,b)=0
  \quad\Longleftrightarrow\quad
  x\le a \ \text{and}\ y\ge b.
\]
We also further employ regularization terms to discourage gray-scale intermediate values.  

Directly solving the inverse optimization problems~\eqref{eq:inv-disp} and ~\eqref{eq:Jgap-def} with repeated
finite-element eigenvalue solves would be prohibitively expensive.
In the remainder of the paper, we therefore replace the discrete forward
map $F_h^{\rm pix}$ by a POD--DeepONet surrogate and couple it with
neural networks that parametrize admissible designs.  The resulting
end-to-end differentiable models allow us to solve the inverse
optimization problems by standard
gradient-based training.

\begin{remark}[Relaxed optimization and modeling perspective]
The inverse problems \eqref{eq:inv-disp} and \eqref{eq:Jgap-def} are posed on
the relaxed design domain $[0,1]^{N_f}$. This choice makes the objective
functionals $J_{\rm disp}$ and $J_{\rm gap}$ differentiable with respect to
$\boldsymbol{\rho}$ and permits the use of gradient-based algorithms.
However, physically realizable photonic crystals employ only two material
phases with permittivities $\epsilon_{\rm air}$ and $\epsilon_{\rm alum}$.
After solving the relaxed optimization problems, we therefore apply a
$0$--$1$ projection to obtain a binary design
$\boldsymbol{\rho}^{\rm bin}\in\{0,1\}^{N_f}$; see
Section~\ref{sec:Inverse design based on the POD–DeepONet forward map} for details.
\end{remark}

\section{POD–DeepONet for the forward band-structure map}
\label{sec:POD-DeepONet}

In this section, we construct a POD--DeepONet surrogate for the discrete band
map
\[
  F_h^{\rm pix} : [0,1]^{N_f} \to \mathbb R^{N_k\times N_b},
  \qquad
  \boldsymbol{\rho} \longmapsto \mathbf W_h(\boldsymbol{\rho}),
\]
defined in~\eqref{eq:Fpix}. The surrogate follows the
\emph{trunk--branch} structure of DeepONet. The dependence on the wave vector
$\mathbf k$ is represented in a fixed low-dimensional POD basis (the
\emph{trunk}), while the dependence on the design vector $\boldsymbol{\rho}$ is
learned by a neural network (the \emph{branch}). The trunk basis is computed
once from high-fidelity band data and then kept fixed during training.
This design exploits the spectral structure of the Bloch operator, yielding a compact low-rank representation in the wave-vector variable and reducing the number of effective trainable parameters compared with generic high-dimensional network parameterizations.

\subsection{Trunk construction via snapshot POD}
\label{sec:POD}

We first construct a fixed POD trunk basis by a snapshot
proper orthogonal decomposition (POD).

Let $\{\boldsymbol{\rho}^{(i)}\}_{i=1}^{N_{\mathrm{train}}}\subset\{0,1\}^{N_f}$
be a set of design vectors, and let
$\epsilon^{(i)} := \epsilon(\cdot;\boldsymbol{\rho}^{(i)}) \in \mathcal{E}_{\rm pix}$
be the corresponding piecewise-constant permittivities defined
in~\eqref{eq:eps-rho}.
For each $\boldsymbol{\rho}^{(i)}$, we solve the discrete eigenproblem
\eqref{eq:matrixEVP} at all wave vectors
$\{\mathbf k_\ell\}_{\ell=1}^{N_k}\subset\mathcal{K}_{\rm hs}$, then retain the
first $N_b$ normalized bands, and assemble the truncated band data
\[
  \mathbf W_h^{(i)}
  := \mathbf W_h(\boldsymbol{\rho}^{(i)})
  =
  \bigl( \widetilde\omega_{h,n}(\mathbf k_\ell;\boldsymbol{\rho}^{(i)}) \bigr)_{
    1\le\ell\le N_k,\;1\le n\le N_b}
  \in \mathbb{R}^{N_k\times N_b}.
\]
For fixed $i$ and $1\le n\le N_b$, we regard the column
\[
  \widetilde\omega_{h,n}(\cdot;\boldsymbol{\rho}^{(i)})
  :=
  \bigl( \widetilde\omega_{h,n}(\mathbf k_\ell;\boldsymbol{\rho}^{(i)}) \bigr)_{\ell=1}^{N_k}
  \in \mathbb{R}^{N_k}
\]
as one snapshot of the band functions.
Collecting all bands from all training designs yields
$N_s := N_{\mathrm{train}} N_b$ snapshots, which we index by
\[
  s = (i-1)N_b + n,
  \qquad
  1\le i\le N_{\mathrm{train}},\; 1\le n\le N_b,
\]
and assemble into the snapshot matrix
\begin{equation}\label{eq:snapshot-matrix}
  X := (x_{s\ell})_{1\le s\le N_s,\;1\le\ell\le N_k}
  \in\mathbb{R}^{N_s\times N_k},
  \qquad
  x_{s\ell} := \widetilde\omega_{h,n}(\mathbf k_\ell;\boldsymbol{\rho}^{(i)}).
\end{equation}
Thus, each row of $X$ is a band function sampled at the discrete points
$\{\mathbf k_\ell\}_{\ell=1}^{N_k}$.

We perform a singular value decomposition \cite{golub2013matrix} of the
transposed snapshot matrix
\begin{equation}\label{eq:svd}
  X^\top = U \Sigma V^\top,
\end{equation}
where $U\in\mathbb{R}^{N_k\times N_k}$ contains the left singular vectors and
$\Sigma = \mathrm{diag}(\sigma_1,\dots,\sigma_{N_k})$ with
$\sigma_1\ge\cdots\ge\sigma_{N_k}\ge 0$.
Given a POD tolerance $\tau_{\mathrm{POD}}>0$, we choose the POD rank
$N_{\mathrm{POD}}\le N_k$ as the smallest integer such that the normalized
tail energy satisfies
\begin{equation}\label{eq:POD-criterion}
  \displaystyle\sum_{j>N_{\mathrm{POD}}}\sigma_j^2\big/\displaystyle\sum_{j=1}^{N_k}\sigma_j^2
  < \tau_{\mathrm{POD}}.
\end{equation}
We then define the trunk matrix
\begin{equation}\label{eq:trunk}
  \Phi_{\mathrm{tr}}
  :=
  \bigl[ \varphi_1,\dots,\varphi_{N_{\mathrm{POD}}} \bigr]
  \in\mathbb{R}^{N_k\times N_{\mathrm{POD}}},
\end{equation}
where $\varphi_j\in\mathbb{R}^{N_k}$ denotes the $j$th column of $U$.
By construction, the columns $\{\varphi_j\}_{j=1}^{N_{\mathrm{POD}}}$ form an
orthonormal basis of the POD subspace
$\mathcal V_{\mathrm{POD}} := \operatorname{span}
\{\varphi_1,\dots,\varphi_{N_{\mathrm{POD}}}\}\subset\mathbb R^{N_k}$.

The classical Eckart--Young--Mirsky theorem for the singular value
decomposition~\cite{eckart1936approximation,golub2013matrix} yields an
explicit formula for the projection error of the snapshot set onto the POD
subspace.
For each snapshot
band matrix \(\mathbf W_h^{(i)}\in\mathbb R^{N_k\times N_b}\),
\(i=1,\dots,N_{\mathrm{train}}\), we introduce its orthogonal projection onto
\(\mathcal V_{\mathrm{POD}}\) by
\begin{equation*}\label{eq:POD-snapshot-proj}
  \mathbf W_{h,\mathrm{POD}}^{(i)}
  :=
  \Phi_{\mathrm{tr}} C^{(i)},
  \qquad
  C^{(i)} := \Phi_{\mathrm{tr}}^\top \mathbf W_h^{(i)}
  \in\mathbb R^{N_{\mathrm{POD}}\times N_b}.
\end{equation*}

\begin{proposition}[POD truncation error on the snapshot set]
\label{prop:POD-error}
Let \(X\in\mathbb R^{N_s\times N_{\mathrm{train}}}\) be the snapshot matrix
defined in~\eqref{eq:snapshot-matrix} and \(\{\sigma_j\}_{j=1}^{N_k}\) be the singular values of \(X^\top\) in
non-increasing order, and let \(\Phi_{\mathrm{tr}}\) and
\(\mathbf W_{h,\mathrm{POD}}^{(i)}\) be defined as above. Then the average
projection error of the dataset satisfies
\[
  \frac{1}{N_{\mathrm{train}}}
  \sum_{i=1}^{N_{\mathrm{train}}}
  \bigl\|
    \mathbf W_h^{(i)} - \mathbf W_{h,\mathrm{POD}}^{(i)}
  \bigr\|_F^2
  =
  \sum_{j>N_{\mathrm{POD}}} \sigma_j^2,
\]
where \(\|\cdot\|_F\) denotes the Frobenius norm.
\end{proposition}

For a general design \(\boldsymbol{\rho}\in[0,1]^{N_f}\), we approximate the
band data by orthogonal projection onto \(\mathcal V_{\mathrm{POD}}\) and
define the POD-projected band map
\begin{equation}\label{eq:F-POD}
  F_h^{\rm POD} : [0,1]^{N_f} \to \mathbb R^{N_k\times N_b},\qquad
  F_h^{\rm POD}(\boldsymbol{\rho}) := \mathbf W_h^{\rm POD}(\boldsymbol{\rho}),
\end{equation}
where the projected band matrix has the form
\begin{equation}\label{eq:POD-matrix-approx}
  \mathbf W_h^{\rm POD}(\boldsymbol{\rho})
  :=
  \Phi_{\mathrm{tr}}\,C(\boldsymbol{\rho}),
  \qquad
  C(\boldsymbol{\rho})
  :=
  \Phi_{\mathrm{tr}}^\top \mathbf W_h(\boldsymbol{\rho})
  \in\mathbb R^{N_{\mathrm{POD}}\times N_b}.
\end{equation}
By construction, the restriction of \(F_h^{\rm POD}\) to the snapshot designs
reproduces the optimal average projection error stated in
Proposition~\ref{prop:POD-error}.
If the snapshots
$\{\mathbf W_h^{(i)}\}_{i=1}^{N_{\mathrm{train}}}$ sample $\mathcal E_{\rm rel}$
adequately and the singular values $\{\sigma_j\}_{j=1}^{N_k}$ decay rapidly,
then the first $N_{\mathrm{POD}}$ modes capture most of the variability of the
band structures represented in the data. In this regime, one expects the POD
truncation error to remain small also for designs that lie in the same region
of the design space but were not included in the snapshot set.

In our POD--DeepONet construction, the trunk matrix $\Phi_{\mathrm{tr}}$ is
computed once in this offline stage from the high-fidelity band data and then
kept fixed as the trunk in all subsequent training and inverse-design
computations.

\subsection{POD--DeepONet framework}
\label{sec:architecture}

In the POD--DeepONet architecture, the \emph{trunk} is given by the precomputed matrix $\Phi_{\mathrm{tr}}$, while the
\emph{branch} approximates the nonlinear dependence of the POD coefficients
on the design vector. More precisely, we approximate the coefficient map
$C$ in~\eqref{eq:POD-matrix-approx} by a fully connected branch network
\[
  C_\theta :
  [0,1]^{N_f} \longrightarrow \mathbb{R}^{N_{\mathrm{POD}}\times N_b},
  \qquad
  \boldsymbol{\rho} \longmapsto
  C_\theta(\boldsymbol{\rho})
  :=
  \bigl(c^\theta_{j,n}(\boldsymbol{\rho})\bigr)_{j=1,\dots,N_{\mathrm{POD}}}^{n=1,\dots,N_b},
\]
with trainable parameters $\theta$. Given $\Phi_{\mathrm{tr}}$ in~\eqref{eq:trunk} and the branch output $C_\theta(\boldsymbol{\rho})$,
the POD--DeepONet prediction of the discrete band data at the sampled
$\mathbf k$-points is
\begin{equation}
  \mathbf W^{\mathrm{POD\text{-}DO}}_h(\boldsymbol{\rho};\theta)
  :=
  \Phi_{\mathrm{tr}}\,C_\theta(\boldsymbol{\rho})
  =
  \bigl(
 \omega^{\mathrm{POD\text{-}DO}}_{h,n}(\mathbf k_\ell;\boldsymbol{\rho},\theta)
  \bigr)_{\ell=1,\dots,N_k}^{n=1,\dots,N_b},
  \label{eq:POD-DO-pred}
\end{equation}
or, componentwise,
\begin{equation*}
  \omega^{\mathrm{POD\text{-}DO}}_{h,n}(\mathbf k_\ell;\boldsymbol{\rho},\theta)
  =
  \sum_{j=1}^{N_{\mathrm{POD}}}
    c^\theta_{j,n}(\boldsymbol{\rho})\,\varphi_j(\mathbf k_\ell),
  \qquad
  1\le\ell\le N_k,\;1\le n\le N_b,
  \label{eq:POD-DO-pointwise}
\end{equation*}
where
$c^\theta_{j,n}(\boldsymbol{\rho})$ denotes the $(j,n)$-entry of the
coefficient matrix
$C_\theta(\boldsymbol{\rho})$, and
$\varphi_j$ is the $j$th POD mode, i.e., the $j$th column of
$\Phi_{\mathrm{tr}}$.

For later reference, we summarize the resulting surrogate as the operator
\begin{equation}\label{eq:Fpod}
  F_h^{\mathrm{POD\text{-}DO}} :
  [0,1]^{N_f} \longrightarrow \mathbb{R}^{N_k\times N_b},\qquad
  F_h^{\mathrm{POD\text{-}DO}}(\boldsymbol{\rho};\theta)
  :=
  \mathbf W_h^{\mathrm{POD\text{-}DO}}(\boldsymbol{\rho};\theta)
  =
  \Phi_{\mathrm{tr}}\,C_\theta(\boldsymbol{\rho}),
\end{equation}
which serves as an efficient approximation of $F_h^{\rm pix}$ in the
subsequent forward and inverse computations.

Equation~\eqref{eq:Fpod} defines the
POD--DeepONet surrogate of the pixel-based band map $F_h^{\rm pix}$.
The precomputed POD modes $\{\varphi_j\}_{j=1}^{N_{\mathrm{POD}}}$ provide a
fixed, physics-informed trunk basis over the sampled wave vectors, while the
trainable branch $C_\theta$ learns how the POD coefficients depend on the
high-dimensional design vector $\boldsymbol{\rho}$. This separation of
variables combines the efficiency and interpretability of a reduced-order
model in $\mathbf k$ with the flexibility of a neural network in the design
space.

\subsection{Training and evaluation of the POD--DeepONet forward map}\label{sec:Inverse design based on the POD–DeepONet forward map and evaluation of the POD--DeepONet forward map}

In the numerical experiments, we train $F_h^{\mathrm{POD\text{-}DO}}$ in a
supervised fashion on a training set of $N_{\mathrm{train}}$ labelled samples $\{(\boldsymbol{\rho}^{(i)}, \mathbf W_h^{(i)})\}_{i=1}^{N_{\mathrm{train}}}$ drawn from the full data set $\{(\boldsymbol{\rho}^{(i)}, \mathbf W_h^{(i)})\}_{i=1}^{N_{\mathrm{data}}}$.

\subsubsection*{Band-data standardization}

To stabilize the training process, we apply an affine
standardization (zero mean and unit variance) to the band data, entrywise.
We define the empirical mean and variance
\[
  \mu
  :=
  \frac{1}{N_{\mathrm{train}}N_k N_b}
  \sum_{i=1}^{N_{\mathrm{train}}}
  \sum_{\ell=1}^{N_k}
  \sum_{n=1}^{N_b}
  \widetilde\omega_{h,n}(\mathbf k_\ell;\boldsymbol{\rho}^{(i)}),
  \qquad
  \sigma^2
  :=
  \frac{1}{N_{\mathrm{train}}N_k N_b}
  \sum_{i,\ell,n}
  \bigl(\widetilde\omega_{h,n}(\mathbf k_\ell;\boldsymbol{\rho}^{(i)})-\mu\bigr)^2.
\]
The standardized snapshots are
\begin{equation}\label{eq:standardized snapshots}
   \mathbf W_h^{(i),\mathrm{std}}
  :=
  \frac{\mathbf W_h^{(i)}-\mu\mathbf{1}}{\sigma}
  \in\mathbb{R}^{N_k\times N_b},
  \qquad i=1,\dots,N_{\mathrm{train}},   
\end{equation}
with subtraction and division understood componentwise.
These standardized targets $\mathbf W_h^{(i),\mathrm{std}}$ are used in the
loss~\eqref{eq:training-loss} below.

\subsubsection*{Supervised training objective}

For a design vector $\boldsymbol{\rho}\in\{0,1\}^{N_f}$, the
POD--DeepONet prediction of the standardized band data is introduced in \eqref{eq:POD-DO-pred}, i.e.,
\begin{equation*}\label{eq:normalized band}
  \mathbf W_h^{\rm POD\text{-}DO}(\boldsymbol{\rho};\theta)
  :=
  \Phi_{\mathrm{tr}}\,C_\theta(\boldsymbol{\rho})
  \in\mathbb{R}^{N_k\times N_b}.
\end{equation*}
We determine $\theta$ by minimizing the empirical mean-squared error
\begin{equation}
  L(\theta)
  :=
  \frac{1}{N_{\mathrm{train}}N_k N_b}
  \sum_{i=1}^{N_{\mathrm{train}}}
  \bigl\|
    \mathbf W_h^{\rm POD\text{-}DO}(\boldsymbol{\rho}^{(i)};\theta)
    - \mathbf W_h^{(i),\mathrm{std}}
  \bigr\|_F^2.
  \label{eq:training-loss}
\end{equation}
In practice, we minimize $L(\theta)$ by a
stochastic gradient method (Adam) until convergence and obtain the
trained parameter vector $\theta^\ast$.

\subsubsection*{Online evaluation}

Once trained, the POD--DeepONet surrogate provides a fast,
differentiable approximation of the discrete forward map restricted to
the pixel parametrization. For any design
$\boldsymbol{\rho}\in[0,1]^{N_f}$, we recover the band data
by inverting the standardization:
\begin{equation}\label{eq:recover lambda}
  \widetilde{\mathbf W}_h^{\rm POD\text{-}DO}(\boldsymbol{\rho};\theta^\ast)
  :=
  \sigma\,\mathbf W_h^{\rm POD\text{-}DO}(\boldsymbol{\rho};\theta^\ast)
  + \mu\mathbf{1}.
\end{equation}

The complete offline--online pipeline for forward evaluation with the POD--DeepONet surrogate is summarized in Algorithm~\ref{alg:POD-DO-full}.

\begin{algorithm}[htbp]
  \caption{Forward evaluation with the POD--DeepONet surrogate}
  \label{alg:POD-DO-full}
  \KwIn{Dataset $\{(\boldsymbol{\rho}^{(i)},\mathbf W_h^{(i)})\}_{i=1}^{N_{\mathrm{data}}}$;
        POD tolerance $\tau_{\mathrm{POD}}$; branch network $C_\theta$; size of training set $N_{\rm train}$; 
        query design $\boldsymbol{\rho}$.}
  \KwOut{Predicted band matrix $\widetilde{\mathbf W}_h^{\rm POD\text{-}DO}(\boldsymbol{\rho};\theta^\ast)$.}

  \BlankLine
  \tcp{Offline}
  Randomly select a training set $\{(\boldsymbol{\rho}^{(i)},\mathbf W_h^{(i)})\}_{i=1}^{N_{\mathrm{train}}}$ from $\{(\boldsymbol{\rho}^{(i)},\mathbf W_h^{(i)})\}_{i=1}^{N_{\mathrm{data}}}$\;
  Compute trunk matrix $\Phi_{\mathrm{tr}}$ by snapshot POD with tolerance
  $\tau_{\mathrm{POD}}$ \eqref{eq:trunk}\;
  Compute $\mu,\sigma$ and standardized data $\{\mathbf W_h^{(i),\mathrm{std}}\}_{i=1}^{N_{\mathrm{train}}}$ ~\eqref{eq:standardized snapshots}\;

  \BlankLine
  \tcp{Training}
  Find $\theta^\ast := \arg\min_\theta L(\theta)$
  using loss~\eqref{eq:training-loss} on
  $\{(\boldsymbol{\rho}^{(i)},\mathbf W_h^{(i),\mathrm{std}})\}_{i=1}^{N_{\mathrm{train}}}$\;

  \BlankLine
  \tcp{Online prediction}
  Compute the POD--DeepONet output for standardized data $\mathbf W_h^{\mathrm{POD\text{-}DO}}(\boldsymbol{\rho};\theta^\ast)
      =F_h^{\mathrm{POD\text{-}DO}}(\boldsymbol{\rho};\theta^\ast):= \Phi_{\mathrm{tr}} C_{\theta^\ast}(\boldsymbol{\rho})$\;
      Undo the standardization to obtain the predicted band matrix    
      \begin{equation*}
          \widetilde{\mathbf W}_h^{\rm POD\text{-}DO}(\boldsymbol{\boldsymbol{\rho}};\theta^\ast)
      := \sigma\,\mathbf W_h^{\mathrm{POD\text{-}DO}}(\boldsymbol{\rho};\theta^\ast)
         + \mu\,\mathbf 1.
      \end{equation*}
      
\end{algorithm}

\begin{remark}[fixed $\mathbf k$-grid in the POD trunk]The present POD--DeepONet surrogate learns a discrete band map defined on a prescribed high-symmetry path with a fixed $\mathbf k$-grid. The POD trunk is extracted from band snapshots sampled on this grid and therefore provides the most reliable reduced representation within the same path and sampling resolution. If a denser $\mathbf k$-sampling is desired, one may enrich the snapshot set on a refined grid and rebuild the POD basis, while keeping the branch architecture and the overall training pipeline unchanged. Likewise, alternative symmetry paths can be accommodated by regenerating the corresponding snapshot ensembles and constructing the associated POD trunk. Developing a continuous-in-$\mathbf k$ trunk, e.g., by parameterizing the path coordinate, constitutes a natural extension toward discretization-invariant evaluation.
\end{remark}

\subsection{Approximation properties of the POD--DeepONet surrogate}
\label{sec:convergence}

In this subsection, we analyze the approximation properties of the
POD--DeepONet surrogate on the relaxed pixel-based design space. We first
establish continuity of the discrete band map. Building on this, we prove a
universal approximation theorem for POD--DeepONet and derive a decomposition
of its total error.

We first state the continuity result for the discrete band map on the
relaxed pixel design space.

\begin{proposition}[Continuity of the discrete band map]
\label{prop:Fh-cont}
Assume that the admissible permittivities satisfy
\[
  0 < \epsilon_{\min} \le \epsilon(\mathbf x;\boldsymbol{\rho})
    \le \epsilon_{\max} < \infty,
  \qquad \text{for all }\mathbf x\in\Omega
  \text{ and all }\boldsymbol{\rho}\in[0,1]^{N_f}.
\]
In the TE polarization, we set
\[
  \alpha(\mathbf x;\boldsymbol{\rho}) := \epsilon(\mathbf x;\boldsymbol{\rho})^{-1},
  \qquad
  \beta(\mathbf x;\boldsymbol{\rho}) := 1,
\]
whereas in the TM polarization, we set
\[
  \alpha(\mathbf x;\boldsymbol{\rho}) := 1,
  \qquad
  \beta(\mathbf x;\boldsymbol{\rho}) := \epsilon(\mathbf x;\boldsymbol{\rho}).
\]
Let $\mathbf A(\mathbf k;\boldsymbol{\rho})$ and $\mathbf B(\boldsymbol{\rho})$ be the
finite-element matrices in \eqref{eq:matrixEVP} assembled from these
coefficients. Then, for every wave vector
$\mathbf k\in\mathcal K_{\rm hs}$ and every band index $1\le n\le N_b$, the
discrete eigenvalue
\[
  \lambda_{h,n}(\mathbf k;\boldsymbol{\rho})
  := \lambda_{h,n}\bigl(\mathbf k;\epsilon(\cdot;\boldsymbol{\rho})\bigr)
\]
depends continuously on $\boldsymbol{\rho}\in[0,1]^{N_f}$. Consequently, the
discrete band map
\[
  F_h^{\rm pix} : [0,1]^{N_f} \to \mathbb R^{N_k\times N_b}, \qquad
  \boldsymbol{\rho} \mapsto \mathbf W_h(\boldsymbol{\rho})
\]
is a continuous mapping.
\end{proposition}

\begin{proof}
By the definition of $\epsilon(\mathbf x;\boldsymbol{\rho})$ in
\eqref{eq:eps-rho} and the lifting operator \eqref{eq:lift-op}, we have
\[
  \epsilon(\mathbf x;\boldsymbol{\rho})
  = \epsilon_{\rm air}
    + \tilde\rho_\ell\,(\epsilon_{\rm alum}-\epsilon_{\rm air}),
  \qquad \mathbf x\in P_\ell,\ \ell=1,\dots,N_{\rm pix}.
\]
Thus, on each pixel $P_\ell$, the value of $\epsilon$ depends continuously on
the scalar $\tilde\rho_\ell$. The same property holds for
$\alpha(\mathbf x;\boldsymbol{\rho})$ and
$\beta(\mathbf x;\boldsymbol{\rho})$, so $\alpha$ and $\beta$ are continuous
functions of $\boldsymbol{\rho}\in[0,1]^{N_f}$.

Next, in the finite-element discretization, the matrices
$\mathbf A(\mathbf k;\boldsymbol{\rho})$ and $\mathbf B(\boldsymbol{\rho})$ are assembled
from element integrals whose integrands are linear in $\alpha$ and $\beta$
(cf.\ \eqref{eq:matrixEVP}). Therefore, each entry of
$\mathbf A(\mathbf k;\boldsymbol{\rho})$ and $\mathbf B(\boldsymbol{\rho})$ is a continuous
function of $\boldsymbol{\rho}$, and the mappings
\[
  \boldsymbol{\rho} \mapsto \mathbf A(\mathbf k;\boldsymbol{\rho}),\qquad
  \boldsymbol{\rho} \mapsto \mathbf B(\boldsymbol{\rho})
\]
are continuous with respect to the Frobenius norm.

Moreover, the uniform bounds on $\epsilon(\mathbf x;\boldsymbol{\rho})$
imply that $\beta(\mathbf x;\boldsymbol{\rho})$ is uniformly bounded below by
a positive constant. Hence, the mass matrix $\mathbf B(\boldsymbol{\rho})$ is
Hermitian positive definite for every $\boldsymbol{\rho}\in[0,1]^{N_f}$ and
belongs to the positive-definite cone. We can therefore define the Hermitian
matrix
\[
  \mathbf H(\mathbf k;\boldsymbol{\rho})
  := \mathbf B(\boldsymbol{\rho})^{-1/2}
     \mathbf A(\mathbf k;\boldsymbol{\rho})
     \mathbf B(\boldsymbol{\rho})^{-1/2},
\]
where $\mathbf B(\boldsymbol{\rho})^{-1/2}$ denotes the principal matrix square root.
The map $\mathbf B\mapsto \mathbf B^{-1/2}$ is analytic, and hence continuous, on the
positive-definite cone \cite{bhatia2013matrix}. Combined with the continuity
of $\mathbf A$ and $\mathbf B$, this shows that
$\boldsymbol{\rho}\mapsto \mathbf H(\mathbf k;\boldsymbol{\rho})$ is continuous in
the Frobenius norm.

The generalized eigenproblem \eqref{eq:matrixEVP} coincides with the
ordinary eigenproblem for $\mathbf H(\mathbf k;\boldsymbol{\rho})$, and
$\lambda_{h,n}(\mathbf k;\boldsymbol{\rho})$ is the $n$th eigenvalue of
$\mathbf H(\mathbf k;\boldsymbol{\rho})$, ordered non-decreasingly. Fix
$\mathbf k\in\mathcal K_{\rm hs}$ and
$\boldsymbol{\rho},\boldsymbol{\rho}'\in[0,1]^{N_f}$, and let
$\lambda_{h,n}(\mathbf k;\boldsymbol{\rho})$ and
$\lambda_{h,n}(\mathbf k;\boldsymbol{\rho}')$,
$n=1,\dots,N_h$, denote the ordered eigenvalues of
$\mathbf H(\mathbf k;\boldsymbol{\rho})$ and $\mathbf H(\mathbf k;\boldsymbol{\rho}')$,
respectively. The Hoffman--Wielandt inequality for Hermitian matrices
\cite{bhatia2013matrix} gives
\begin{equation}
  \bigl|
    \lambda_{h,n}(\mathbf k;\boldsymbol{\rho})
    - \lambda_{h,n}(\mathbf k;\boldsymbol{\rho}')
  \bigr|
  \le
  \bigl\|
    \mathbf H(\mathbf k;\boldsymbol{\rho})
    - \mathbf H(\mathbf k;\boldsymbol{\rho}')
  \bigr\|_F,
  \label{eq:eig-Lip}
\end{equation}
for every $1\le n\le N_b$. Hence, $\lambda_{h,n}(\mathbf k;\cdot)$ is also
continuous on $[0,1]^{N_f}$.

Finally, for each sampled wave vector $\mathbf k_\ell$ and $1\le n\le N_b$,
the normalized band value
\[
  \widetilde\omega_{h,n}(\mathbf k_\ell;\boldsymbol{\rho})
  = \frac{a}{2\pi}\sqrt{\lambda_{h,n}(\mathbf k_\ell;\boldsymbol{\rho})}
\]
is the composition of the continuous function
$\lambda_{h,n}(\mathbf k_\ell;\cdot)$ with the square-root on $[0,\infty)$
and is therefore continuous on $[0,1]^{N_f}$. Collecting these finitely many
scalar functions into $\mathbf W_h(\boldsymbol{\rho})$ yields the continuity of the discrete band map $F_h^{\rm pix}$, as claimed.
\end{proof}

We now turn to the approximation properties of the POD--DeepONet surrogate.

\begin{theorem}[Approximation properties of POD--DeepONet]
\label{thm:POD-DO-UAT}
Let $\mathbf W_h^{\rm POD}$ be given
by~\eqref{eq:POD-matrix-approx}.  Fix $1\le N_{\mathrm{POD}}\le N_k$ and a
trunk matrix $\Phi_{\rm tr}\in\mathbb R^{N_k\times N_{\mathrm{POD}}}$ with
orthonormal columns.  For any $\varepsilon>0$, there exists a branch
network $C_\theta:[0,1]^{N_f}\to\mathbb R^{N_{\mathrm{POD}}\times N_b}$
and an associated POD--DeepONet band map
\[
  \mathbf W_h^{\mathrm{POD\text{-}DO}}(\boldsymbol{\rho})
  := \Phi_{\rm tr} C_\theta(\boldsymbol{\rho}),
  \qquad \boldsymbol{\rho}\in[0,1]^{N_f},
\]
such that
\[
  \sup_{\boldsymbol{\rho}\in[0,1]^{N_f}}
  \bigl\|
    \mathbf W_h^{\rm POD}(\boldsymbol{\rho})
    - \mathbf W_h^{\mathrm{POD\text{-}DO}}(\boldsymbol{\rho})
  \bigr\|_F
  \le \varepsilon .
\]
\end{theorem}

\begin{proof}
By Proposition~\ref{prop:Fh-cont}, the discrete band
map $F_h^{\rm pix}$ defined
in~\eqref{eq:Fpix} is continuous on $[0,1]^{N_f}$, and hence the
corresponding band matrix $\mathbf W_h(\boldsymbol{\rho})
  := F_h^{\rm pix}(\boldsymbol{\rho})
  \in\mathbb R^{N_k\times N_b}$ depends continuously on $\boldsymbol{\rho}$.  Recall the POD coefficient
map 
\[
  C:[0,1]^{N_f}\to\mathbb R^{N_{\mathrm{POD}}\times N_b},\qquad
  C(\boldsymbol{\rho}) := \Phi_{\rm tr}^\top \mathbf W_h(\boldsymbol{\rho}),
\]
introduced in~\eqref{eq:POD-matrix-approx}.  Since $\Phi_{\rm tr}^\top$ is
a fixed linear operator, $C(\cdot)$ is also continuous on $[0,1]^{N_f}$.

By the universal approximation theorem for feed-forward neural
networks (see, e.g.,~\cite{cybenko1989approximation,hornik1991approximation}),
for any $\varepsilon>0$ there exists a branch network
$C_\theta:[0,1]^{N_f}\to\mathbb R^{N_{\mathrm{POD}}\times N_b}$ such that
\[
  \sup_{\boldsymbol{\rho}\in[0,1]^{N_f}}
  \bigl\|C(\boldsymbol{\rho})-C_\theta(\boldsymbol{\rho})\bigr\|_F
  \le \varepsilon .
\]
Using the definition of $\mathbf W_h^{\rm POD}$ in~\eqref{eq:POD-matrix-approx}
and that of $\mathbf W_h^{\mathrm{POD\text{-}DO}}$ above, we obtain
for all $\boldsymbol{\rho}\in[0,1]^{N_f}$,
\[
  \bigl\|
    \mathbf W_h^{\rm POD}(\boldsymbol{\rho})
    - \mathbf W_h^{\mathrm{POD\text{-}DO}}(\boldsymbol{\rho})
  \bigr\|_F
  = \bigl\|
      \Phi_{\rm tr}\bigl(C(\boldsymbol{\rho})-C_\theta(\boldsymbol{\rho})\bigr)
    \bigr\|_F
  \le \bigl\|C(\boldsymbol{\rho})-C_\theta(\boldsymbol{\rho})\bigr\|_F
  \le \varepsilon ,
\]
where we used that $\Phi_{\rm tr}$ has orthonormal columns and thus acts as
a contraction in the Frobenius norm.  Taking the supremum over
$\boldsymbol{\rho}\in[0,1]^{N_f}$ yields the desired estimate.
\end{proof}

Combining Theorem~\ref{thm:POD-DO-UAT} with Proposition~\ref{prop:POD-error}, we can decompose the total
approximation error of the surrogate into two contributions.
For any $\boldsymbol{\rho}\in[0,1]^{N_f}$,
\begin{equation*}
  \bigl\|
    \mathbf W_h(\boldsymbol{\rho})
    - \mathbf W_h^{\mathrm{POD\text{-}DO}}(\boldsymbol{\rho})
  \bigr\|_F
  \le
  \bigl\|
    \mathbf W_h(\boldsymbol{\rho})
    - \mathbf W_h^{\rm POD}(\boldsymbol{\rho})
  \bigr\|_F
  +
  \bigl\|
    \mathbf W_h^{\rm POD}(\boldsymbol{\rho})
    - \mathbf W_h^{\mathrm{POD\text{-}DO}}(\boldsymbol{\rho})
  \bigr\|_F .
\end{equation*}
The first term on the right-hand side is the \emph{POD truncation error}.
It is controlled by the tolerance $\tau_{\mathrm{POD}}$ used to select the
POD rank and can be reduced by increasing $N_{\mathrm{POD}}$ or enriching
the snapshot set; see Proposition~\ref{prop:POD-error}.  The second term is
the \emph{network approximation error}, which, for a fixed trunk
$\Phi_{\rm tr}$, can be made arbitrarily small in principle by increasing
the expressiveness of the branch network, as guaranteed by
Theorem~\ref{thm:POD-DO-UAT}.  In practice, we first choose
$N_{\mathrm{POD}}$ so that the POD truncation error falls below a prescribed
tolerance $\tau_{\mathrm{POD}}$, and then select the network architecture
and training procedure so that the remaining discrepancy is of the same
order or smaller, making the POD truncation the dominant source of error in
the surrogate.

These results show that, for a fixed POD rank, the POD--DeepONet surrogate
can approximate the discrete band map $F_h^{\rm pix}$ uniformly on the
relaxed design space $[0,1]^{N_f}$. Consequently, the surrogate-based
objective functionals used in the inverse problems below provide consistent
approximations of their discrete finite-element counterparts.

\begin{remark}
Throughout this subsection, we work on the relaxed design space
$[0,1]^{N_f}$ introduced in Section~\ref{sec:design-space}.  The discrete
band map and the POD--DeepONet surrogate are evaluated in practice on
binary designs $\boldsymbol{\rho}\in\{0,1\}^{N_f}$, which form a subset of
$[0,1]^{N_f}$, so the continuity and approximation results above apply in
particular to the physically relevant two-material unit cells.  The
relaxation is used to simplify the analysis and to enable gradient-based
inverse-design algorithms. In the numerical experiments, the relaxed
outputs are subsequently projected to binary pixels to obtain two-phase
unit cells, and the resulting performance confirms that this relaxed
formulation is adequate for the applications considered here.
\end{remark}

\section{POD–DeepONet-based inverse band–structure design}
\label{sec:Inverse design based on the POD–DeepONet forward map}

In this section, we develop POD--DeepONet–based algorithms for the two
inverse band--structure problems introduced in
Section~\ref{sec:forward-inverse}.  In both settings, we seek design
vectors $\widehat{\boldsymbol{\rho}}$ whose associated unit cells reproduce a
prescribed band structure or realize a prescribed band gap.  To enable
gradient-based optimization, we work with relaxed density space
$[0,1]^{N_f}$ as introduced in
Section~\ref{sec:design-space}.  We construct inverse neural networks
that map a target specification to a relaxed design $\widehat{\boldsymbol{\rho}}$ and evaluate these
networks through the pre-trained POD--DeepONet surrogate, which provides
fast band predictions and gradients with respect to $\widehat{\boldsymbol{\rho}}$.
We first describe the dispersion-to-structure inverse design, followed
by the band-gap inverse problem.

\subsection{Inverse dispersion-to-structure problem}\label{sec:Inverse dispersion-to-structure design}
The dispersion-to-structure problem seeks a design
$\widehat{\boldsymbol{\rho}}$ whose band structure is close to a prescribed target
band matrix $\mathbf W_h$ ~\eqref{eq:inv-disp}. Rather than solving
\eqref{eq:inv-disp} separately for each target, we introduce a neural network for the inverse map and train it in a supervised fashion.

We reuse the forward data set
\[
  \bigl\{(\boldsymbol{\rho}^{(i)},\mathbf W_h^{(i)})\bigr\}_{i=1}^{N_{\mathrm{data}}}
\]
introduced in Section~\ref{sec:Inverse design based on the POD–DeepONet forward map and evaluation of the POD--DeepONet forward map}
and select a training subset $ \bigl\{(\boldsymbol{\rho}^{(i)},\mathbf W_h^{(i)})\bigr\}_{i=1}^{N_{\mathrm{train}}}$ with $N_{\mathrm{train}}\leq N_{\mathrm{data}}$.
Following the same standardization procedure as in
Section~\ref{sec:Inverse design based on the POD–DeepONet forward map and evaluation of the POD--DeepONet forward map},
we obtain standardized targets
$\mathbf W_h^{(i),\mathrm{std}}\in\mathbb R^{N_k\times N_b}$, for $i=1,\cdots,N_{N_{\mathrm{data}}}$.
For each target, we form a feature vector
$\mathbf y^{(i)}\in\mathbb R^{N_\omega}$, $N_\omega:=N_k N_b$, by flattening
$\mathbf W_h^{(i),\mathrm{std}}$ into a single column.
The inverse network
\[
  G^{\mathrm{disp}}_\phi:\mathbb R^{N_\omega}\to\mathbb R^{N_f},\qquad
  \mathbf z = G^{\mathrm{disp}}_\phi(\mathbf y),
\]
is a fully connected multilayer perceptron with parameters $\phi$.
Given a target feature vector $\mathbf y$, the network produces an
unconstrained vector $\mathbf z\in\mathbb R^{N_f}$, which is then mapped to a
relaxed wedge and finally to a binary wedge as described next.

\subsubsection*{Relaxed parametrization and binary projection}

To convert the unconstrained output
$\mathbf z = G^{\mathrm{disp}}_\phi(\mathbf y)$ into a relaxed wedge in
$[0,1]^{N_f}$, we first apply the logistic map
\[
  \sigma:\mathbb R^{N_f}\to(0,1)^{N_f},\qquad
  \boldsymbol\rho^{\rm sig}
  = \sigma(\mathbf z)
  := \bigl(\sigma(z_j)\bigr)_{j=1}^{N_f},\quad
  \sigma(t):=\frac{1}{1+e^{-t}},
\]
which ensures $0<\rho^{\rm sig}_j<1$ for all $j=1,\cdots,N_f$.
To sharpen the relaxed design toward nearly binary values, we then apply
a smooth Heaviside projection componentwise.
For a steepness parameter $\beta>0$ and threshold $\eta\in(0,1)$, we set
\begin{equation}
  \boldsymbol\rho(\mathbf z)
  := H_\beta\bigl(\boldsymbol\rho^{\rm sig}\bigr)\in(0,1)^{N_f},\qquad
  \bigl(H_\beta(\boldsymbol\rho^{\rm sig})\bigr)_j
  :=
  \frac{\tanh(\beta\eta)+\tanh\bigl(\beta(\rho^{\rm sig}_j-\eta)\bigr)}
       {\tanh(\beta\eta)+\tanh\bigl(\beta(1-\eta)\bigr)},
  \quad j=1,\dots,N_f.
  \label{eq:rho-Hbeta}
\end{equation}
During training, we use a continuation strategy: $\beta$ is increased from
a small initial value to a large final value, so that
$\boldsymbol\rho(\mathbf z)$ is gradually pushed closer to $\{0,1\}^{N_f}$
while the map $\mathbf z\mapsto\boldsymbol\rho(\mathbf z)$ remains smooth
and differentiable.

Combining the inverse network with this parametrization, the relaxed
wedge associated with a target feature vector $\mathbf y$ is
\begin{equation}
  \boldsymbol{\rho}(\mathbf y;\phi,\beta)
  :=
  H_\beta\bigl(\sigma(G^{\mathrm{disp}}_\phi(\mathbf y))\bigr)
  \in[0,1]^{N_f}.
  \label{eq:rho-inverse}
\end{equation}
Thus, the network output $\mathbf z=G^{\mathrm{disp}}_\phi(\mathbf y)$ is converted, via
the sigmoid and Heaviside transforms, into a relaxed density
$\boldsymbol{\rho}(\mathbf y;\phi,\beta)$ in the continuous design
space $[0,1]^{N_f}$.

After training has converged, we obtain a binary wedge by hard
thresholding,
\begin{equation}\label{eq:final-rho-bin}
  \rho^{\rm bin}_j(\mathbf y;\phi,\beta)
  :=
  \mathbf 1_{\{\rho_j(\mathbf y;\phi,\beta)>1/2\}},
  \qquad j=1,\dots,N_f,
\end{equation}
which yields a discrete two-material unit cell.

\subsubsection*{Training objective}

The inverse network should produce unit cells whose band diagrams match
the targets, remain close to binary designs, and stay near the supervised
examples from the database.  Let
\[
  \bigl\{
    (\mathbf W_h^{(i),\mathrm{std}},
    \mathbf y^{(i)},
    \boldsymbol{\rho}^{(i)}
  \bigr)\}_{i=1}^{N_{\mathrm{train}}}
\]
denote the standardized target band matrices, their feature vectors, and the
corresponding wedge designs in $\{0,1\}^{N_f}$ on the training set.
For each $i$, the inverse network and the relaxed parametrization produce
\[
  \widehat{\boldsymbol{\rho}}^{(i)}(\phi,\beta)
  :=
  \boldsymbol{\rho}\bigl(\mathbf y^{(i)};\phi,\beta\bigr)
  \in[0,1]^{N_f},
\]
as in~\eqref{eq:rho-inverse}.  The POD--DeepONet surrogate with frozen
parameters $\theta^\ast$ then gives the standardized band prediction
$\mathbf W_h^{\mathrm{POD\text{-}DO}}(\widehat{\boldsymbol{\rho}}^{(i)}(\phi,\beta);\theta^\ast)$.

For each training sample, we define three contributions:
\begin{align}
  J_{\mathrm{MSE}}^{(i)}(\phi,\beta)
  &:=
  \frac{1}{N_k N_b}
  \bigl\|
    \mathbf W_h^{\mathrm{POD\text{-}DO}}
      (\widehat{\boldsymbol{\rho}}^{(i)}(\phi,\beta);\theta^\ast)
    - \mathbf W_h^{(i),\mathrm{std}}
  \bigr\|_F^2,
  \label{eq:JMSE-ours}\\[0.3em]
  R_{\mathrm{bin}}^{(i)}(\phi,\beta)
  &:=
  \frac{1}{N_f}
  \sum_{j=1}^{N_f}
    \widehat{\rho}^{(i)}_j(\phi,\beta)\bigl(1-\widehat{\rho}^{(i)}_j(\phi,\beta)\bigr),
  \label{eq:R-bin-ours}\\[0.3em]
  R_{\mathrm{sup}}^{(i)}(\phi,\beta)
  &:=
  \frac{1}{N_f}
  \bigl\|
    \widehat{\boldsymbol{\rho}}^{(i)}(\phi,\beta)
    - \boldsymbol{\rho}^{(i)}
  \bigr\|_2^2.
  \label{eq:R-sup-ours}
\end{align}
Here, $J_{\mathrm{MSE}}^{(i)}$ measures the mismatch between the predicted and
target band diagrams, $R_{\mathrm{bin}}^{(i)}$ penalizes grey pixels and
vanishes exactly for $\{0,1\}$ designs, and $R_{\mathrm{sup}}^{(i)}$ is a
weak proximity term that keeps the relaxed designs close to the supervised
examples.

The empirical training objective reads
\begin{equation}
  \mathcal J_{\mathrm{disp}}(\phi,\beta)
  :=
  \frac{1}{N_{\mathrm{train}}}
  \sum_{i=1}^{N_{\mathrm{train}}}
  \Bigl(
    J_{\mathrm{MSE}}^{(i)}(\phi,\beta)
    + \gamma_{\mathrm{bin}} R_{\mathrm{bin}}^{(i)}(\phi,\beta)
    + \gamma_{\mathrm{sup}} R_{\mathrm{sup}}^{(i)}(\phi,\beta)
  \Bigr),
  \label{eq:Jdisp-empirical}
\end{equation}
with regularization weights $\gamma_{\mathrm{bin}},\gamma_{\mathrm{sup}}>0$.
The parameters $\phi$ are optimized by gradient-based methods while keeping
$\theta^\ast$ fixed. The gradients of all three contributions are propagated
through the POD--DeepONet surrogate and the mapping
$\mathbf y\mapsto\boldsymbol{\rho}(\mathbf y;\phi,\beta)$ by automatic
differentiation.

At inference time, only the target band feature vector is required:
a single evaluation of $G^{\mathrm{disp}}_\phi$, followed by the transforms
\eqref{eq:rho-inverse} and~\eqref{eq:final-rho-bin}, yields a binary wedge
that approximately reproduces the prescribed band diagram.

Note that both the dispersion-to-structure and the band-gap maps are highly non-injective, since
different microstructures can give rise to essentially the same band
diagram. Optimizing only $J^{(i)}_{\rm MSE}$ and
$R^{(i)}_{\rm bin}$ would therefore leave a large equivalence class of
admissible designs.  The supervised term $R^{(i)}_{\rm sup}$ is introduced
as a data-driven regularization that selects among these competing
solutions by steering the inverse network toward the design manifold
represented in the database, thereby stabilizing the training process.

\begin{remark}[Relaxed inputs for the forward surrogate]
During inverse training, the combination of the sigmoid and the Heaviside functions keeps
$\boldsymbol{\rho}(\mathbf y;\phi,\beta)$ close to $\{0,1\}^{N_f}$, so the
use of relaxed inputs $\boldsymbol{\rho}\in[0,1]^{N_f}$ constitutes only a
mild extrapolation.  Proposition~\ref{prop:Fh-cont} gives continuity
of the discrete band map on $[0,1]^{N_f}$, and
Theorem~\ref{thm:POD-DO-UAT} ensures that POD--DeepONet can approximate this
continuous map on the relaxed domain.  While the forward POD--DeepONet
surrogate is trained on binary designs
$\boldsymbol{\rho}\in\{0,1\}^{N_f}$ and evaluated on relaxed inputs in the
inverse problems, the numerical results show that it remains accurate on
these near-binary inputs.
\end{remark}

\subsection{Inverse band-gap problem}\label{sec:Inverse band-gap design}

We proceed similarly for the band-gap inverse problem, where the target
is not a full band diagram but a prescribed gap interval $(a,b)$ between
bands $p$ and $p+1$.  We encode this specification by the
three-dimensional feature vector
\[
  \mathbf g := (a,b,p)^\top \in \mathbb R^{3},
\]
whose components are standardized before being fed to the network.
An inverse network
\[
  G^{\mathrm{gap}}_\phi:\mathbb R^{3}\to\mathbb R^{N_f},
  \qquad
  \mathbf z = G^{\mathrm{gap}}_\phi(\mathbf g),
\]
with the same multilayer-perceptron architecture as
$G^{\mathrm{disp}}_\phi$, maps $\mathbf g$ to an unconstrained parameter
vector $\mathbf z$.  The relaxed wedge is then
obtained via the sigmoid and Heaviside projections,
\[
  \boldsymbol{\rho}(\mathbf g;\phi,\beta)
  := H_\beta\bigl(\sigma(G^{\mathrm{gap}}_\phi(\mathbf g))\bigr)
  \in[0,1]^{N_f},
\]
and serves as input to the fixed POD--DeepONet surrogate
$\mathbf W_h^{\mathrm{POD\text{-}DO}}(\boldsymbol{\rho};\theta^\ast)$.

The gap inverse network should produce unit cells that open the desired
band gap, remain close to binary designs, and stay near representative
examples from the database.  Let
\[
  \bigl\{
    (\mathbf g^{(i)},\,
    \boldsymbol{\rho}^{(i)}
  \bigr)\}_{i=1}^{N_{\mathrm{train}}}
\]
denote the standardized gap targets
$\mathbf g^{(i)}=(a^{(i)},b^{(i)},p^{(i)})^\top$ and the corresponding
wedge designs $\boldsymbol{\rho}^{(i)}\in\{0,1\}^{N_f}$
used for training.  For each $i$, we define the relaxed wedge
\[
  \widehat{\boldsymbol{\rho}}^{(i)}(\phi,\beta)
  :=
  \boldsymbol{\rho}\bigl(\mathbf g^{(i)};\phi,\beta\bigr)
  \in[0,1]^{N_f},
\]
and evaluate the POD--DeepONet surrogate
$\mathbf W_h^{\mathrm{POD\text{-}DO}}(\widehat{\boldsymbol{\rho}}^{(i)}(\phi,\beta);\theta^\ast)$
with frozen parameters $\theta^\ast$.
Let
$\omega^{\mathrm{POD\text{-}DO}}_{h,n}(\mathbf k_\ell;
  \boldsymbol\rho,\theta^\ast)$
denote the $(\ell,n)$-entry of
$\mathbf W_h^{\mathrm{POD\text{-}DO}}(\boldsymbol\rho;\theta^\ast)$.
For a given target $(a^{(i)},b^{(i)},p^{(i)})$, we define the gap-enforcement
term
\begin{equation}
\begin{aligned}
  J_{\mathrm{gap}}^{(i)}(\phi,\beta)
  :=\;&
  \frac{1}{N_k N_b}
  \sum_{\ell=1}^{N_k}\sum_{n=1}^{N_b}
    \mathrm{ReLU}\bigl(
      \omega^{\mathrm{POD\text{-}DO}}_{h,n}(\mathbf k_\ell;
        \widehat{\boldsymbol{\rho}}^{(i)}(\phi,\beta),\theta^\ast)-a^{(i)}\bigr)\,
    \mathrm{ReLU}\bigl(
      b^{(i)}-\omega^{\mathrm{POD\text{-}DO}}_{h,n}(\mathbf k_\ell;
        \widehat{\boldsymbol{\rho}}^{(i)}(\phi,\beta),\theta^\ast)\bigr)\\
  &\;+\frac{1}{N_k}
  \sum_{\ell=1}^{N_k}
  \Bigl(
    \mathrm{ReLU}\bigl(
      \omega^{\mathrm{POD\text{-}DO}}_{h,p^{(i)}}(\mathbf k_\ell;
        \widehat{\boldsymbol{\rho}}^{(i)}(\phi,\beta),\theta^\ast)-a^{(i)}\bigr)^{2}
    +
    \mathrm{ReLU}\bigl(
      b^{(i)}-\omega^{\mathrm{POD\text{-}DO}}_{h,p^{(i)}+1}(\mathbf k_\ell;
        \widehat{\boldsymbol{\rho}}^{(i)}(\phi,\beta),\theta^\ast)\bigr)^{2}
  \Bigr),
\end{aligned}
\label{eq:Jgap-sample}
\end{equation}
which penalizes surrogate bands lying inside $(a^{(i)},b^{(i)})$ and pushes
the $p^{(i)}$-th band below $a^{(i)}$ while lifting the $(p^{(i)}+1)$-st
band above $b^{(i)}$.
In addition, as in the discussion above, we also use the binarity penalty and a weak proximity term, as defined in \eqref{eq:R-bin-ours}-\eqref{eq:R-sup-ours}.
Therefore, the empirical training objective is
\begin{equation}
  \mathcal J_{\mathrm{gap}}(\phi,\beta)
  :=
  \frac{1}{N_{\mathrm{train}}}
  \sum_{i=1}^{N_{\mathrm{train}}}
  \Bigl(
    J_{\mathrm{gap}}^{(i)}(\phi,\beta)
    + \gamma_{\mathrm{bin}} R_{\mathrm{bin}}^{(i)}(\phi,\beta)
    + \gamma_{\mathrm{sup}} R_{\mathrm{sup}}^{(i)}(\phi,\beta)
  \Bigr),
  \label{eq:Jgap-empirical}
\end{equation}
with regularization weights $\gamma_{\mathrm{bin}},\gamma_{\mathrm{sup}}>0$.
The parameters $\phi$ are optimized by gradient-based methods while
$\theta^\ast$ is kept fixed. Gradients are propagated through the
POD--DeepONet surrogate and the map
$\mathbf g\mapsto\boldsymbol{\rho}(\mathbf g;\phi,\beta)$ by automatic
differentiation.

At inference time, only the gap descriptor $\mathbf g$ is required:
a single evaluation of $G^{\mathrm{gap}}_\phi$, followed by the transforms
\eqref{eq:rho-inverse} and~\eqref{eq:final-rho-bin}, yields a binary wedge
that approximately opens the prescribed band gap.

\begin{remark}The gap-enforcement term $J^{(i)}_{\rm gap}$ penalizes surrogate band values that fall inside the target interval
$(a^{(i)},b^{(i)})$ and pushes the $p^{(i)}$-th and $(p^{(i)}\!+\!1)$-st bands
below $a^{(i)}$ and above $b^{(i)}$, respectively.  Thus, the loss is designed
to ensure that the prescribed interval is contained in a complete band gap,
rather than to enforce exact matching of the gap edges.
In practice, however, the supervised proximity term and the limited
band-gap enlargement attainable under the present two-phase, symmetry- and
resolution-constrained parametrization jointly act against excessive
overshooting.  As a result, the realized gaps in our numerical experiments
tend to bracket the target interval tightly, with gap edges commonly located
near $a^{(i)}$ and $b^{(i)}$.
\end{remark}

\begin{algorithm}[htbp]
  \caption{Gradient-based inverse design with POD--DeepONet}
  \label{alg:inverse-design}
  \KwIn{Training data $\{(\mathbf y^{(i)},\boldsymbol{\rho}^{(i)})\}_{i=1}^{N_{\mathrm{train}}}$ for Problem~\ref{sec:Inverse dispersion-to-structure design}
        or $\{(\mathbf g^{(i)},\boldsymbol{\rho}^{(i)})\}_{i=1}^{N_{\mathrm{train}}}$ for Problem~\ref{sec:Inverse band-gap design};
        fixed POD--DeepONet surrogate $\mathbf W_h^{\mathrm{POD\text{-}DO}}(\cdot;\theta^\ast)$;
        inverse network $G_\phi$; Heaviside threshold $\eta$;
        continuation schedule for $\beta$; regularization weights
        $\gamma_{\mathrm{bin}}$, $\gamma_{\mathrm{sup}}$; query target
        $\mathbf y$ or $\mathbf g$.}
  \KwOut{Binary unit cell
         $\epsilon^{\mathrm{bin}}(\mathbf x;\boldsymbol{\rho}^{\mathrm{bin}})$.}

  \BlankLine
  \tcp{Offline}
  For each $i=1,\dots,N_{\mathrm{train}}$, standardize the input
  $\mathbf y^{(i)}$ (dispersion targets) or $\mathbf g^{(i)}$ (gap targets)
  to obtain a feature vector $\mathbf s^{(i)}$, and form the standardized
  training set $\{(\mathbf s^{(i)},\boldsymbol{\rho}^{(i)})\}_{i=1}^{N_{\mathrm{train}}}$\;

  \BlankLine
  \tcp{Training}
  For each $i=1,\dots,N_{\mathrm{train}}$, define
  \[
    \mathbf z^{(i)} := G_\phi(\mathbf s^{(i)}),\qquad
    \widehat{\boldsymbol{\rho}}^{(i)}(\phi,\beta)
      := H_\beta\bigl(\sigma(\mathbf z^{(i)})\bigr).
  \]
  Use $\bigl\{(\mathbf s^{(i)},\boldsymbol{\rho}^{(i)},
      \widehat{\boldsymbol{\rho}}^{(i)}(\phi,\beta))\bigr\}_{i=1}^{N_{\mathrm{train}}}$
  to form the empirical objective
  $\mathcal J_{\mathrm{disp}}(\phi,\beta)$ in~\eqref{eq:Jdisp-empirical}
  or $\mathcal J_{\mathrm{gap}}(\phi,\beta)$ in~\eqref{eq:Jgap-empirical}.\\
  Find $\phi^\ast := \arg\min_{\phi}\mathcal J_{\mathrm{disp}}(\phi,\beta)$
  or $\phi^\ast := \arg\min_{\phi}\mathcal J_{\mathrm{gap}}(\phi,\beta)$
  on the training set\;

  \BlankLine
  \tcp{Online prediction}
  Construct and standardize the feature vector $\mathbf s$ for the query
  target $\mathbf y$ or $\mathbf g$\;
  $\mathbf z \leftarrow G_{\phi^\ast}(\mathbf s)$,\quad
  $\boldsymbol{\rho}^\ast \leftarrow H_\beta\bigl(\sigma(\mathbf z)\bigr)$\;
  $\rho^{\mathrm{bin}}_j \leftarrow \mathbf 1_{\{\rho^\ast_j>1/2\}}$,
  \ $j=1,\dots,N_f$\;
  Construct the binary unit cell
  $\epsilon^{\mathrm{bin}}(\mathbf x;\boldsymbol{\rho}^{\mathrm{bin}})$
  via the pixel parametrization~\eqref{eq:eps-rho}.
\end{algorithm}

Algorithm~\ref{alg:inverse-design} summarizes the gradient-based inverse-design procedure for both the dispersion-to-structure and band-gap problems, in which all standardized targets (band-diagram features $\mathbf y$ or gap descriptors $\mathbf g$) are represented by feature vectors $\mathbf s$ and processed within a unified pipeline.
Figure~\ref{fig:poddo-framework} provides a compact overview of the
overall workflow, from snapshot POD and surrogate training to forward
queries and the two inverse design problems.

\section{Numerical experiments}
\label{sec:numerics}

In this section, we assess the performance of the POD--DeepONet forward surrogate in Algorithm \ref{alg:POD-DO-full} and the inverse design schemes based on it in Algorithm \ref{alg:inverse-design}.

\begin{figure}[htbp]
  \centering
  \includegraphics[width=\textwidth,
                  trim={0cm 0cm 0cm 0cm},clip]{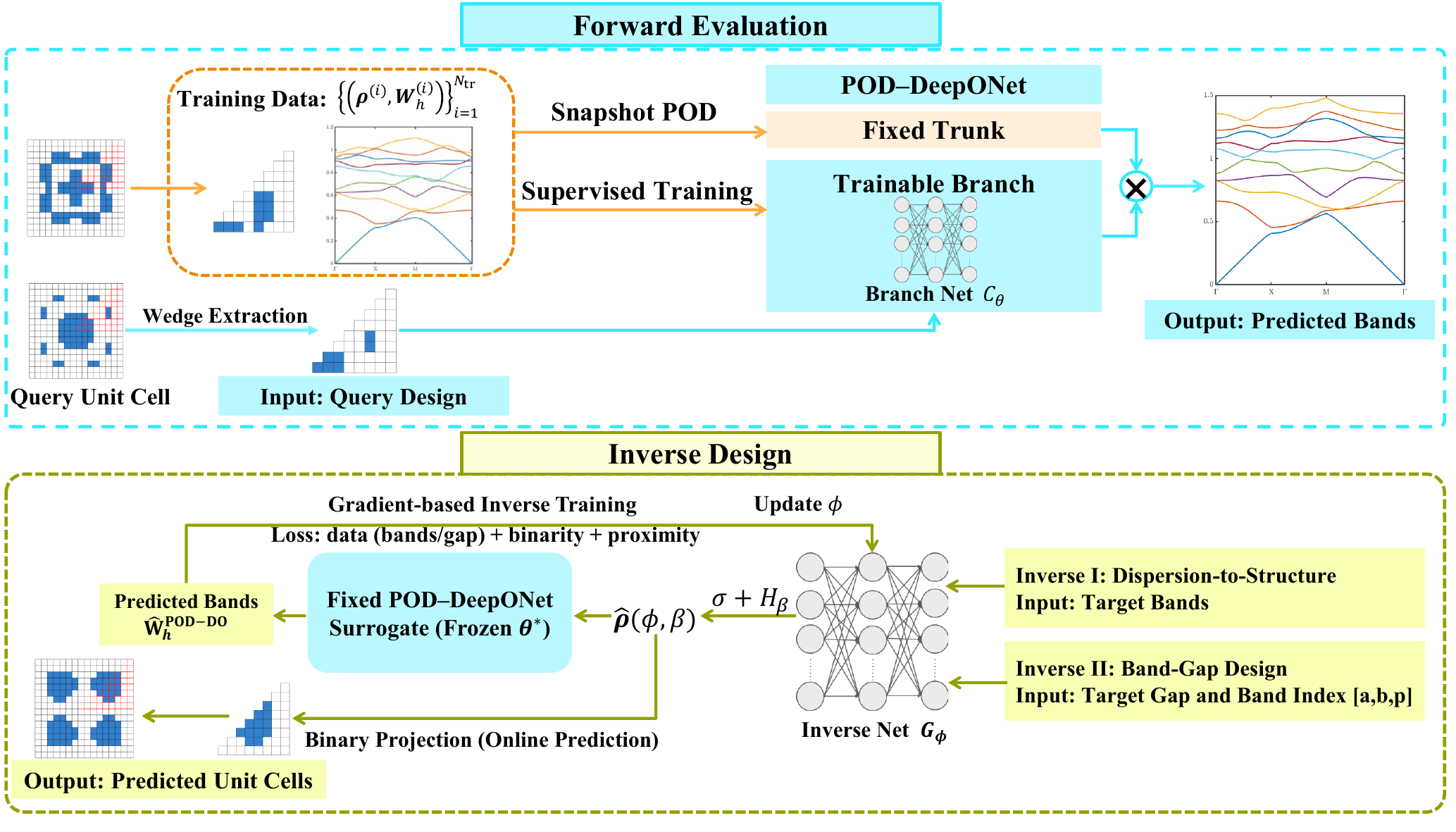}
  \caption{Schematic workflow of the POD--DeepONet framework.
The top panel summarizes forward evaluation: snapshot POD constructs a fixed
trunk, the branch network is trained using a band MSE loss on standardized
data (Algorithm~\ref{alg:POD-DO-full}), and the resulting surrogate predicts
band structures for query designs.
The bottom panel shows the two POD--DeepONet-based inverse design procedures,
where dispersion-to-structure and band-gap targets are handled by
gradient-based inverse training with data, binarity, and proximity terms
(Algorithm~\ref{alg:inverse-design}).}
  \label{fig:poddo-framework}
\end{figure}

We consider two-dimensional photonic crystals with square unit cells in TE
polarization, governed by the scalar Helmholtz eigenproblem~\eqref{TE2}.
The TM case can be treated in the same manner.
All examples use binary, $p4m$-symmetric unit cells of $16\times16$ pixel grid with
$\epsilon(\mathbf x)\in\{\epsilon_{\rm air},\epsilon_{\rm alum}\}$, where
$\epsilon_{\rm air}=1$ and $\epsilon_{\rm alum}=8.9$. 
Note that each $p4m$-symmetric unit cell can be represented by
a wedge design vector $\boldsymbol{\rho}\in\{0,1\}^{N_f}$ with $N_f=36$
independent pixels (see Figure~\ref{fig:unitcell-examples}). 
Motivated by the open data set of
Jiang \emph{et al.}~\cite{jiang2022dispersion}, we randomly generate wedge
design vectors and then lift them to full $16\times16$ unit cells.
After removing
highly fragmented or trivial patterns, we obtain a database of
$N_{\mathrm{data}}=87{,}474$ distinct unit cells.
For each design $\boldsymbol{\rho}$, we solve the TE eigenproblem~\eqref{TE2}
by a conforming $P_1$ finite-element discretization on a shape-regular
triangulation of the unit cell, leading to the matrix eigenproblem
\eqref{eq:matrixEVP}. We then compute the first $N_b=10$ normalized band
frequencies at $N_k=31$ wave vectors that are uniformly distributed
along the high-symmetry path $\mathcal K_{\rm hs}$
and collect them in
\[
  \mathbf W_h(\boldsymbol{\rho})
  :=
  \bigl(
    \widetilde\omega_{h,n}(\mathbf k_\ell;\boldsymbol{\rho})
  \bigr)_{\ell=1,\dots,N_k}^{n=1,\dots,N_b}
  \in\mathbb R^{N_k\times N_b}.
\]
These matrices form the high-fidelity band-structure data set. We use this data to generate POD snapshots, train the POD--DeepONet forward surrogate, and the two inverse-design models.

The full data set of $N_{\mathrm{data}}=87{,}474$ unit cells is randomly split
into $N_{\mathrm{train}}=81{,}474$ training samples,
$N_{\mathrm{val}}=5{,}000$ validation samples, and
$N_{\mathrm{test}}=1{,}000$ test samples.
Before training, all band data $\mathbf W_h(\boldsymbol{\rho})$ are
standardized as introduced in
Section~\ref{sec:Inverse design based on the POD–DeepONet forward map and evaluation of the POD--DeepONet forward map}, and all errors reported below are measured after
inverting this standardization.
Following Section~\ref{sec:POD}, we assemble the snapshot matrix
$X\in\mathbb R^{N_s\times N_k}$ as defined in \eqref{eq:snapshot-matrix} and
perform a singular value decomposition of $X^\top$.
The POD rank $N_{\rm POD}$ is chosen by the tolerance $\tau_{\rm POD}=10^{-7}$ as in \eqref{eq:POD-criterion}, which yields
$N_{\rm POD}=25$, and the corresponding POD modes form the trunk matrix
$\Phi_{\rm tr}\in\mathbb R^{N_k\times25}$, kept fixed in all subsequent
experiments.

\begin{figure}[htbp]
  \centering
  \includegraphics[width=.9\textwidth,
                  trim={2cm 6.5cm 2cm 6cm},clip]{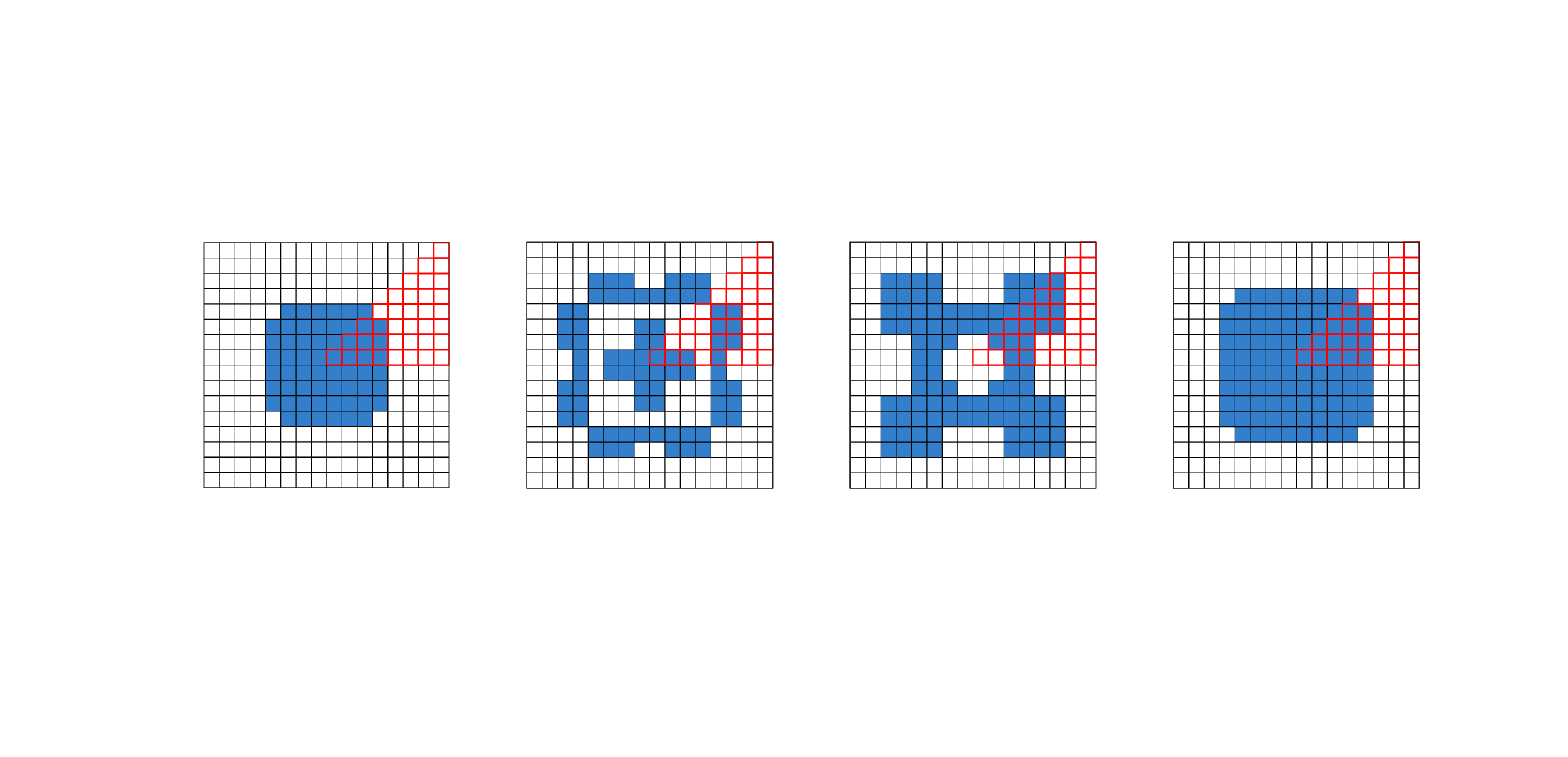}
  \caption{Representative $16\times16$ pixel-based unit cells
  from the data set. Blue pixels denote the high-permittivity material
  $\epsilon_{\rm alum}$ and white pixels the background
  $\epsilon_{\rm air}$. The red staircase region in each panel marks the design wedge $\boldsymbol{\rho}$ with $N_f=36$ independent pixels used to
  parametrize the unit cells.}
  \label{fig:unitcell-examples}
\end{figure}

\subsection{Accuracy of the POD--DeepONet forward surrogate}
\label{sec:forward-numerics}

We first assess the accuracy of the POD--DeepONet surrogate for the forward
map $F_h^{\rm pix}$ on the pixel-based design space.

In the numerical implementation of the branch network $C_\theta$, we use a
fully connected architecture with two hidden layers of width $128$ and ReLU
activations.
Training is carried out in \textsc{PyTorch} using the Adam optimizer with
learning rate $10^{-3}$, mini–batches of size $256$, and 1000 epochs.
The loss function is the mean squared error as in~\eqref{eq:training-loss}.
Table~\ref{tab:data-network-summary} summarizes the main data-set sizes and
network dimensions used in this section.

To highlight the benefit of embedding band information in the fixed trunk basis,
we additionally consider a direct MLP baseline that learns a map from the same wedge
features to the full band outputs without any reduced-order structure.
Specifically, the baseline is a fully connected network with one hidden layer of width
$153$ and ReLU activation, taking the 36-dimensional wedge vector as input and outputting
the $31\times 10$ band values directly.
With this choice, the baseline has approximately $5.34\times 10^{4}$ trainable parameters,
essentially matching the Branch network, thereby ensuring a fair comparison.
We train this baseline using the same settings as the POD--DeepONet.

\begin{table}[htbp]
  \centering
  \caption{Summary of dataset sizes and network dimensions for the POD--DeepONet surrogate.}
  \begin{tabular}{lcc}
    \toprule
    Quantity & Symbol & Value \\
    \midrule
    Number of $\mathbf k$-points on path & $N_k$ & $31$ \\
    Number of bands & $N_b$ & $10$ \\
    Trunk POD rank & $N_{\mathrm{POD}}$ & $25$ \\
    Input feature dimension (branch) & $N_f$ & $36$ \\
    Training samples & $N_{\mathrm{train}}$ & $81{,}474$ \\
    Validation samples & $N_{\mathrm{val}}$ & $5{,}000$ \\
    Test samples & $N_{\mathrm{test}}$ & $1{,}000$ \\
    \bottomrule
  \end{tabular}
  \label{tab:data-network-summary}
\end{table}

Let $\mathcal I_{\mathrm{test}}$ be the index set of the $N_{\mathrm{test}}$
test samples.
On the test set, we quantify prediction accuracy
using band-wise, sample-wise, and global error measures.
For each test index $i\in\mathcal I_{\mathrm{test}}$ and each wave vector $\mathbf k_\ell$, for $\ell=1,\cdots,N_k$, we denote by
$\omega^{\rm POD\text{-}DO}_{h,n}(\mathbf k_\ell;\boldsymbol{\rho}^{(i)},\theta^\ast)$
the predicted band value from the POD--DeepONet or MLP surrogate and by $\widetilde\omega_{h,n}(\mathbf k_\ell;\boldsymbol{\rho}^{(i)})$ the corresponding
high-fidelity band value from FEM solver.
For each band index $1\le n\le N_b$, we define the band-wise
root-mean-square error
\begin{equation}
  \mathrm{RMSE}_n
  :=
  \biggl(
    \frac{1}{N_{\mathrm{test}}\,N_k}
    \sum_{i\in\mathcal I_{\mathrm{test}}}
    \sum_{\ell=1}^{N_k}
    \bigl(
      \omega^{\rm POD\text{-}DO}_{h,n}(\mathbf k_\ell;\boldsymbol{\rho}^{(i)},\theta^\ast)
      -
      \widetilde\omega_{h,n}(\mathbf k_\ell;\boldsymbol{\rho}^{(i)})
    \bigr)^2
  \biggr)^{1/2},
  \label{eq:err-band-rmse}
\end{equation}
the corresponding relative error
\begin{equation}
  \mathrm{relRMSE}_n
  :=
  \frac{\mathrm{RMSE}_n}{%
    \biggl(
      \frac{1}{N_{\mathrm{test}}\,N_k}
      \sum_{i\in\mathcal I_{\mathrm{test}}}
      \sum_{\ell=1}^{N_k}
      \widetilde\omega_{h,n}(\mathbf k_\ell;\boldsymbol{\rho}^{(i)})^2
    \biggr)^{1/2}
  },
  \label{eq:err-band-relrmse}
\end{equation}
and the band-wise mean relative error
\begin{equation}
  \mathrm{MRE}_n
  :=
  \frac{100}{N_{\mathrm{test}}\,N_k}
  \sum_{i\in\mathcal I_{\mathrm{test}}}
  \sum_{\ell=1}^{N_k}
  \frac{
    \bigl|
      \omega^{\rm POD\text{-}DO}_{h,n}(\mathbf k_\ell;\boldsymbol{\rho}^{(i)},\theta^\ast)
      -
      \widetilde\omega_{h,n}(\mathbf k_\ell;\boldsymbol{\rho}^{(i)})
    \bigr|
  }{
    \max\bigl(\bigl|\widetilde\omega_{h,n}(\mathbf k_\ell;\boldsymbol{\rho}^{(i)})\bigr|,\tau\bigr)
  }.
  \label{eq:err-band-mre}
\end{equation}
To measure the spread over different test designs, we also define,
for each $i\in\mathcal I_{\mathrm{test}}$, the sample-wise errors
\begin{align}
  \mathrm{RMSE}_i
  &:=
  \biggl(
    \frac{1}{N_k N_b}
    \sum_{\ell=1}^{N_k}
    \sum_{n=1}^{N_b}
    \bigl(
      \omega^{\rm POD\text{-}DO}_{h,n}(\mathbf k_\ell;\boldsymbol{\rho}^{(i)},\theta^\ast)
      -
      \widetilde\omega_{h,n}(\mathbf k_\ell;\boldsymbol{\rho}^{(i)})
    \bigr)^2
  \biggr)^{1/2},
  \label{eq:err-sample-rmse}
  \\[0.25em]
  \mathrm{relRMSE}_i
  &:=
  \frac{\mathrm{RMSE}_i}{%
    \biggl(
      \frac{1}{N_k N_b}
      \sum_{\ell=1}^{N_k}
      \sum_{n=1}^{N_b}
      \widetilde\omega_{h,n}(\mathbf k_\ell;\boldsymbol{\rho}^{(i)})^2
    \biggr)^{1/2}
  },
  \label{eq:err-sample-relrmse}
\end{align}
and the sample-wise mean relative error
\begin{equation}
  \mathrm{MRE}_i
  :=
  \frac{100}{N_k N_b}
  \sum_{\ell=1}^{N_k}
  \sum_{n=1}^{N_b}
  \frac{
    \bigl|
      \omega^{\rm POD\text{-}DO}_{h,n}(\mathbf k_\ell;\boldsymbol{\rho}^{(i)},\theta^\ast)
      -
      \widetilde\omega_{h,n}(\mathbf k_\ell;\boldsymbol{\rho}^{(i)})
    \bigr|
  }{
    \max\bigl(\bigl|\widetilde\omega_{h,n}(\mathbf k_\ell;\boldsymbol{\rho}^{(i)})\bigr|,\tau\bigr)
  },
  \label{eq:err-sample-mre}
\end{equation}
Here, $\tau>0$ in \eqref{eq:err-band-mre} and \eqref{eq:err-sample-mre} is a small threshold (we take $\tau=10^{-4}$) to avoid
division by very small band frequencies.

The band-wise prediction errors on the test set are summarized in
Table~\ref{tab:band-errors}.
The POD--DeepONet consistently outperforms the direct MLP across all ten bands.
While both surrogates achieve errors at the order of $10^{-3}$,
the advantage of the POD--DeepONet is systematic:
the average band-wise mean relative error decreases from $0.914\%$ for the MLP
to $0.455\%$ for the POD--DeepONet, with uniformly smaller
${\rm RMSE}_n$ and ${\rm relRMSE}_n$ across the spectrum.
This suggests that the POD-informed trunk introduces an effective low-rank inductive bias,
constraining predictions to physically meaningful subspaces rather than fitting the full
high-dimensional outputs directly.
To quantify variability across designs,
Table~\ref{tab:sample-errors} reports the sample-wise summary statistics and representative
test cells ranked by ${\rm MRE}_i$.
The POD--DeepONet improves the mean accuracy and yields a tighter error tail.
In particular, the maximum ${\rm MRE}_i$ is reduced from $3.324\%$ for the MLP
to $1.378\%$, with correspondingly smaller worst-case ${\rm RMSE}_i$
and ${\rm relRMSE}_i$.
This improved tail behavior is important for inverse design, where optimization can push a
black-box surrogate into poorly constrained regimes.
Figure~\ref{fig:forward-bands-example} supports this conclusion: for the best, median, and worst
test cells, the POD--DeepONet tracks the FEM reference well, with the most visible discrepancies
appearing in higher bands near crossings and in segments of rapid variation.
Overall, the POD--DeepONet offers higher average fidelity and improved reliability for the first
ten TE bands, providing a stronger basis for the inverse-design studies that follow.

\begin{table}[htbp]
  \centering
  \caption{Band-wise prediction errors of the POD--DeepONet forward surrogate and the MLP baseline on the test set.}
  \label{tab:forward-bandwise-branch-mlp}
  \setlength{\tabcolsep}{6pt}
  {\small\begin{tabular}{c c c c c c c}
    \toprule
    \multirow{2}{*}{Band $n$}
    & \multicolumn{3}{c}{POD--DeepONet}
    & \multicolumn{3}{c}{MLP} \\
    \cmidrule(lr){2-4}\cmidrule(lr){5-7}
    & $\mathrm{RMSE}_n$
    & $\mathrm{relRMSE}_n$
    & $\mathrm{MRE}_n$ (\%)
    & $\mathrm{RMSE}_n$
    & $\mathrm{relRMSE}_n$
    & $\mathrm{MRE}_n$ (\%) \\
    \midrule
    1  & $2.266\times10^{-3}$ & $6.000\times10^{-3}$ & $0.660$
       & $4.457\times10^{-3}$ & $9.837\times10^{-3}$ & $1.153$ \\
    2  & $3.513\times10^{-3}$ & $6.613\times10^{-3}$ & $0.447$
       & $5.951\times10^{-3}$ & $1.088\times10^{-2}$ & $0.958$ \\
    3  & $4.062\times10^{-3}$ & $6.190\times10^{-3}$ & $0.407$
       & $5.974\times10^{-3}$ & $9.141\times10^{-3}$ & $0.923$ \\
    4  & $3.735\times10^{-3}$ & $4.866\times10^{-3}$ & $0.337$
       & $7.827\times10^{-3}$ & $8.542\times10^{-3}$ & $0.699$ \\
    5  & $5.388\times10^{-3}$ & $5.611\times10^{-3}$ & $0.369$
       & $8.314\times10^{-3}$ & $9.267\times10^{-3}$ & $0.824$ \\
    6  & $5.861\times10^{-3}$ & $6.069\times10^{-3}$ & $0.431$
       & $9.688\times10^{-3}$ & $1.002\times10^{-2}$ & $0.894$ \\
    7  & $6.587\times10^{-3}$ & $6.340\times10^{-3}$ & $0.434$
       & $1.212\times10^{-2}$ & $9.953\times10^{-3}$ & $0.876$ \\
    8  & $7.426\times10^{-3}$ & $6.696\times10^{-3}$ & $0.458$
       & $1.314\times10^{-2}$ & $1.205\times10^{-2}$ & $0.914$ \\
    9  & $8.249\times10^{-3}$ & $7.002\times10^{-3}$ & $0.487$
       & $1.439\times10^{-2}$ & $1.253\times10^{-2}$ & $0.938$ \\
    10 & $9.571\times10^{-3}$ & $7.596\times10^{-3}$ & $0.524$
       & $1.622\times10^{-2}$ & $1.328\times10^{-2}$ & $0.964$ \\
    \midrule
    Average
       & $5.666\times10^{-3}$ & $6.298\times10^{-3}$ & $0.455$
       & $9.808\times10^{-3}$ & $1.055\times10^{-2}$ & $0.914$ \\
    \bottomrule
  \end{tabular}}\label{tab:band-errors}
\end{table}

\begin{table}[htbp]
  \centering
  \caption{Sample-wise summary statistics and representative test cells for the POD--DeepONet forward surrogate and the MLP. Representative samples are ranked by $\mathrm{MRE}_i$.}
  \label{tab:sample-errors}
  \setlength{\tabcolsep}{6pt}
  {\small\begin{tabular}{l c c c c c c}
    \toprule
    \multirow{2}{*}{Sample/Statistics}
    & \multicolumn{3}{c}{POD--DeepONet}
    & \multicolumn{3}{c}{MLP} \\
    \cmidrule(lr){2-4}\cmidrule(lr){5-7}
    & $\mathrm{RMSE}_i$
    & $\mathrm{relRMSE}_i$
    & $\mathrm{MRE}_i$ (\%)
    & $\mathrm{RMSE}_i$
    & $\mathrm{relRMSE}_i$
    & $\mathrm{MRE}_i$ (\%) \\
    \midrule
    Mean
      & $5.471\times10^{-3}$ & $5.938\times10^{-3}$ & $0.471$
      & $8.831\times10^{-3}$ & $9.554\times10^{-3}$ & $0.693$ \\
    Max
      & $1.998\times10^{-2}$ & $2.164\times10^{-2}$ & $1.378$
      & $4.544\times10^{-2}$ & $4.943\times10^{-2}$ & $3.324$ \\
    \midrule
    Best
      & $1.629\times10^{-3}$ & $2.001\times10^{-3}$ & $0.152$
      & $4.110\times10^{-3}$ & $4.287\times10^{-3}$ & $0.307$ \\
    Median
      & $4.772\times10^{-3}$ & $4.592\times10^{-3}$ & $0.365$
      & $1.802\times10^{-2}$ & $8.832\times10^{-3}$ & $0.623$ \\
    Worst
      & $1.736\times10^{-2}$ & $1.751\times10^{-2}$ & $1.378$
      & $2.037\times10^{-2}$ & $2.047\times10^{-2}$ & $3.324$ \\

    \bottomrule
  \end{tabular}}
\end{table}

\begin{figure}[htbp]
  \centering
  \subfigure[Best test sample]{%
    \includegraphics[width=0.31\textwidth]{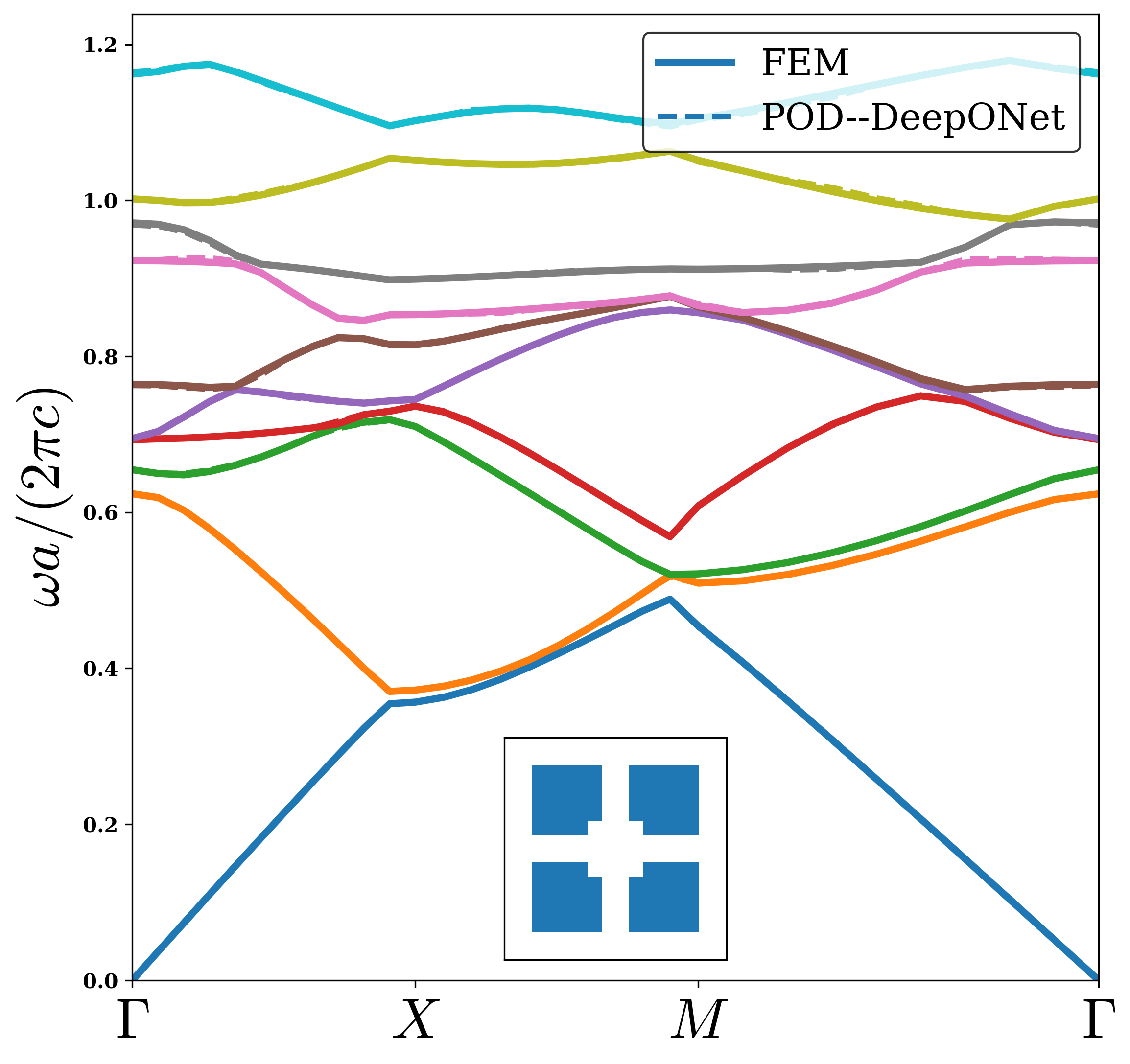}%
  }\hfill
  \subfigure[Median test sample]{%
    \includegraphics[width=0.31\textwidth]{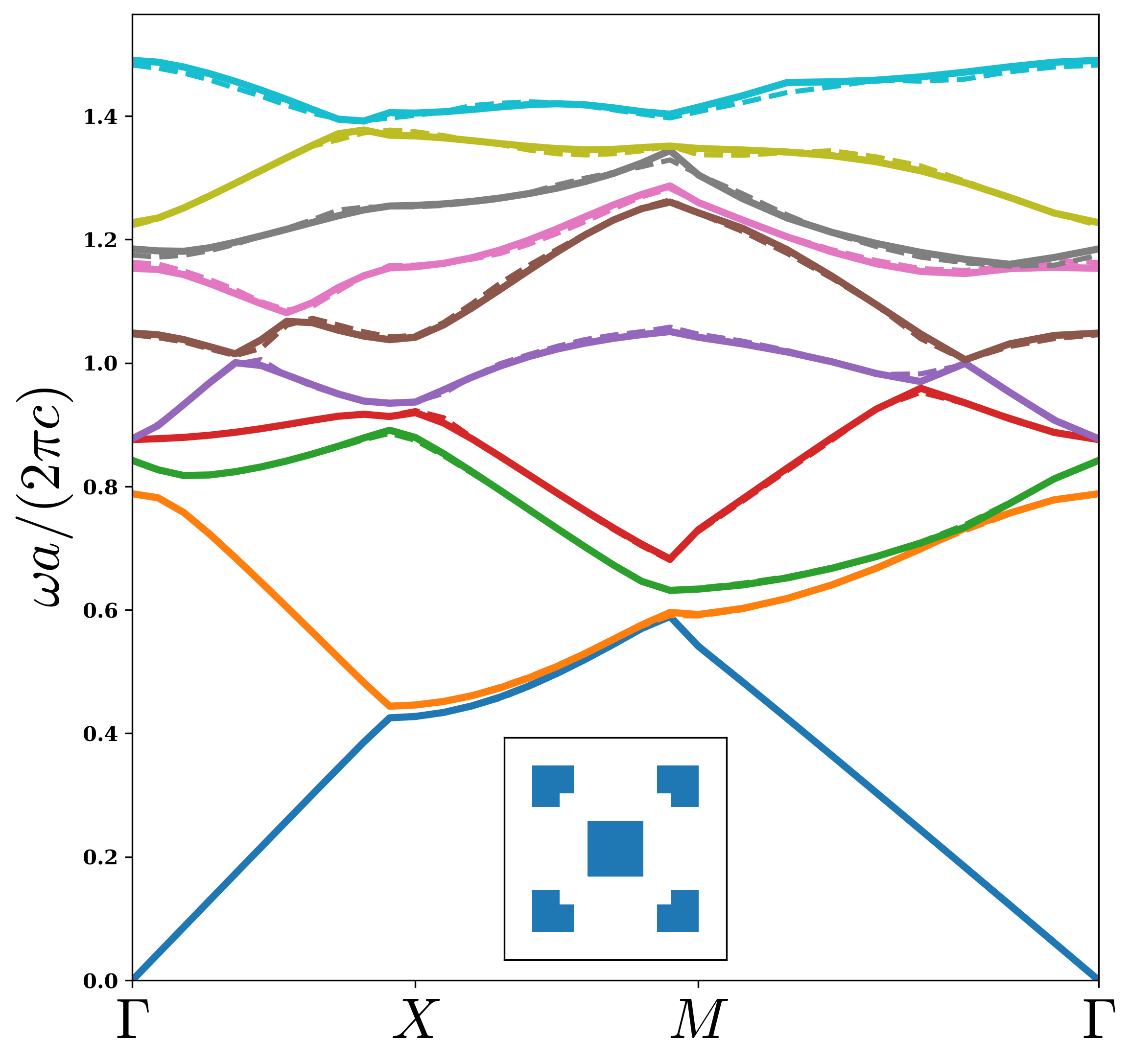}%
  }\hfill
  \subfigure[Worst test sample]{%
    \includegraphics[width=0.31\textwidth]{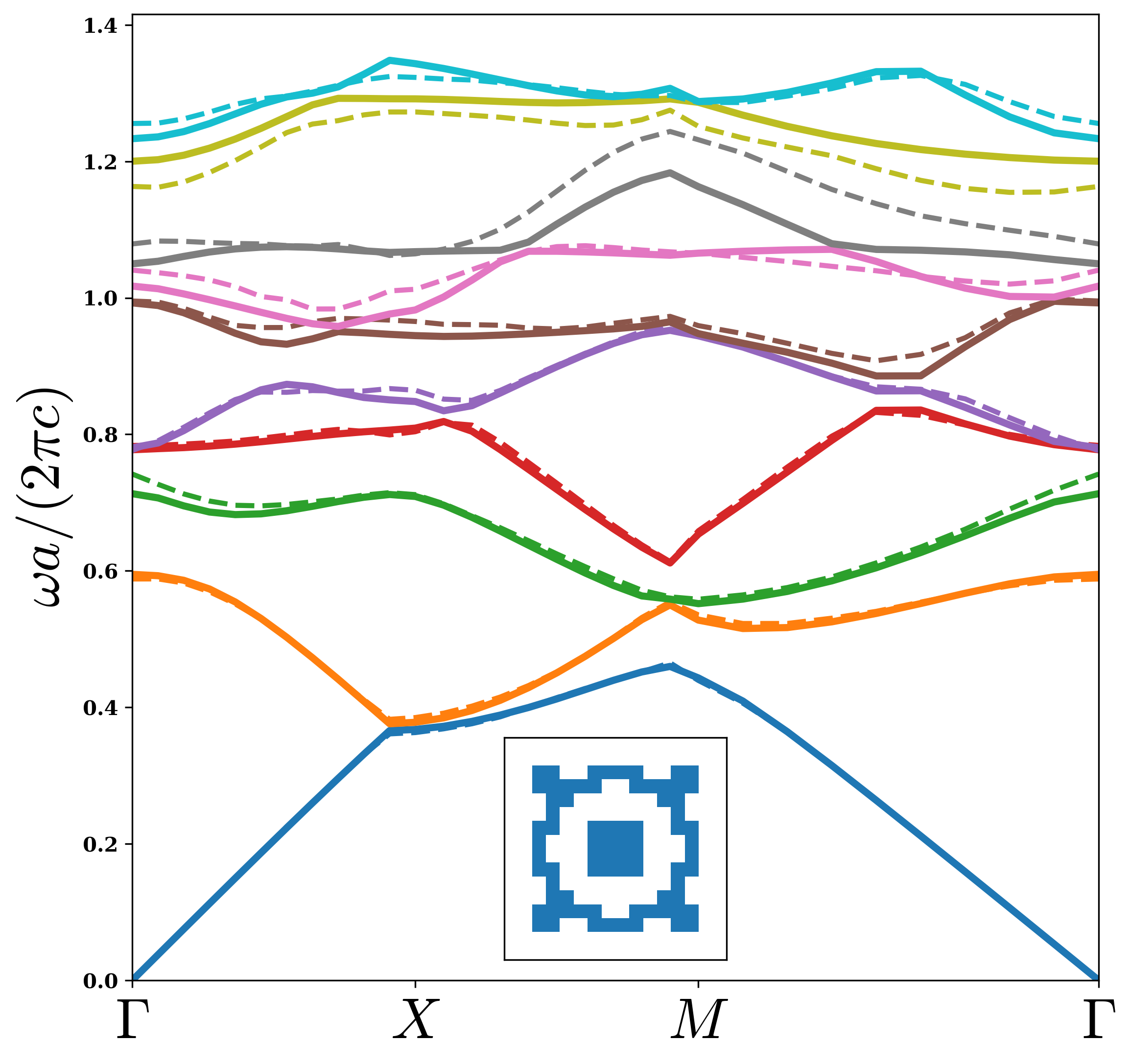}%
  }
  \caption{Best, median, and worst test samples ranked by ${\rm MSE}_i$.
  Each panel shows the normalized TE band structure along the high–symmetry path, with solid lines denoting the high–fidelity
  FEM bands and dashed lines the POD--DeepONet predictions. The corresponding
  $16\times16$ unit cell is displayed as an inset.}
  \label{fig:forward-bands-example}
\end{figure}

\subsection{Dispersion-to-structure inverse design}
\label{subsec:inverse-accuracy}

We now assess the performance of the inverse network for the
dispersion-to-structure problem.
We reuse the same database as in Section~\ref{sec:forward-numerics}.
The inverse network $G^{\rm disp}_\phi$ is implemented as a fully connected
MLP with three hidden layers of width $256$ and ReLU activations.
During training, we fix the POD--DeepONet surrogate and update only $\phi$
by minimizing the empirical loss in~\eqref{eq:Jdisp-empirical} with weights $\gamma_{\mathrm{bin}}=10^{-3}$, $\gamma_{\mathrm{sup}}=10^{-2}$. We use the Adam optimizer with
learning rate $10^{-3}$, a batch size of $512$, and train for $T=150$
epochs. At epoch $e=1,\dots,T$, the steepness parameter in $H_\beta$ is set to
\[
  \beta(e)
  =
  \beta_0 + (\beta_{\max}-\beta_0)\Bigl(\frac{e}{T}\Bigr)^2,
  \qquad
  \beta_0=1,\ \beta_{\max}=16,
\]
so that $\beta$ increases quadratically from $1$ to $16$ over the course
of training.

To further highlight the advantage of embedding band information in the POD-style trunk,
we also consider an MLP-based inverse-design baseline whose forward surrogate adopts
the same MLP architecture as in the forward comparison, so that the inverse
performance difference primarily reflects the effect of the trunk-informed representation.

After training, we evaluate this mapping on the test index set
$\mathcal I_{\mathrm{test}}$.
For each $i\in\mathcal I_{\mathrm{test}}$, the database provides a reference
wedge $\boldsymbol{\rho}^{(i)}$ and the corresponding FEM
band matrix $\mathbf W_h^{(i)}$.
Using $\mathbf W_h^{(i)}$, we construct the standardized
feature vector $\mathbf y^{(i)}$ and compute the relaxed inverse design
\[
  \widehat{\boldsymbol{\rho}}^{(i)}
  :=
  \boldsymbol{\rho}(\mathbf y^{(i)};\phi^\ast,\beta)
  \in[0,1]^{N_f}.
\]
We then enforce strict binarity by thresholding at $1/2$,
\begin{equation}
  \boldsymbol{\rho}^{\mathrm{bin},(i)}
  :=
  \mathbf 1_{\{\widehat{\boldsymbol{\rho}}^{(i)}>1/2\}}
  \in\{0,1\}^{N_f}.
  \label{eq:rho-bin-threshold}
\end{equation}
Given $\boldsymbol{\rho}^{\mathrm{bin},(i)}$, we reconstruct the associated
$16\times16$ pixel-based unit cell, recompute its TE band structure with the
high-fidelity FEM solver, and obtain the normalized inverse band matrix
$\mathbf W_h^{\mathrm{inv},(i)}$.
Comparing $\mathbf W_h^{\mathrm{inv},(i)}$ with
$\mathbf W_h^{(i)}$ using the band-wise and sample-wise error
measures in Section~\ref{sec:forward-numerics} quantifies the accuracy of
the inverse design.

The band-wise inverse-design errors are reported in
Table~\ref{tab:inverse-band-errors}.
A clear and consistent advantage of the POD--DeepONet-based pipeline is observed over
all ten bands.
In particular, the average ${\rm RMSE}_n$ decreases from
$3.246\times10^{-2}$ for the MLP-based approach to
$1.360\times10^{-2}$ for the POD--DeepONet-based approach, and the average
${\rm relRMSE}_n$ is reduced from $3.690\times10^{-2}$ to $1.625\times10^{-2}$.
The average band-wise mean relative error also drops from
$1.624\%$ to $0.980\%$.
These systematic reductions suggest that the POD-informed trunk provides a
structurally aligned low-rank inductive bias for dispersion targets, improving fidelity
across the spectrum.

To examine variability over individual targets,
Table~\ref{tab:inverse-best} reports the sample-wise summary statistics and representative
inverse designs ranked by ${\rm MRE}_i$.
The POD--DeepONet-based inverse improves the mean accuracy
(${\rm MRE}_i$ decreases from $1.615\%$ to $0.989\%$)
and tightens the error tail:
the maximum ${\rm MRE}_i$ is reduced from $16.232\%$ to $10.578\%$,
with corresponding decreases in the worst-case ${\rm RMSE}_i$ and ${\rm relRMSE}_i$.
This tail improvement is practically important because inverse optimization can steer
a surrogate into poorly constrained regions; the POD-constrained representation appears
to stabilize the search.
The zero-error best case in Table~\ref{tab:inverse-best} is consistent with
an exact recovery of the ground-truth binary unit cell under the adopted symmetry and feature encoding.

Figure~\ref{fig:inverse-four-samples} illustrates those representative
inverse-design results.
In each subfigure, the left column shows the target unit cell (top) and the POD--DeepONet-based optimized unit cell (bottom),
and the right panel displays the corresponding TE band structures along the
high-symmetry path (solid: target FEM; dashed: inverse-designed).
For the median and worst examples
(Figures~\ref{fig:inv-b}–\ref{fig:inv-c}), the optimized designs differ from the targets only in a few pixels,
yet the FEM bands of the inverse-designed cells remain close to the targets.
A randomly selected example (Figure~\ref{fig:inv-d}) shows similar behavior.
This is consistent with the continuity of the discrete band map established in
Proposition~\ref{prop:Fh-cont}: small pixelwise perturbations of the
permittivity field induce only small changes in the band functions.

Overall, these results indicate that the POD--DeepONet--based inverse
network reliably recovers TE dispersion diagrams with discrepancies of only a few percent
despite the highly nonconvex binary design space.

\subsection{Band-gap inverse design}
\label{sec:inverse-gap}

The gap-based inverse network is trained and evaluated in an analogous
fashion to the dispersion-based case.
Here, the input is the low-dimensional feature vector
\(\mathbf g := (a,b,p)^\top \in \mathbb R^{3}\) describing a target gap
interval \((a,b)\) between bands \(p\) and \(p+1\).
An MLP \(G^{\rm gap}_\phi\) with the same architecture as
\(G^{\rm disp}_\phi\) maps \(\mathbf g\) to a relaxed wedge
\(\boldsymbol{\rho}\), and we minimize the empirical loss~\eqref{eq:Jgap-empirical} with the same
Heaviside schedule and binarity weight as in
Section~\ref{subsec:inverse-accuracy}.
To isolate the contribution of the POD-informed trunk representation in the
forward constraint, we repeat the same gap-inversion procedure using the
direct MLP forward surrogate as a baseline.
The resulting binary wedges are lifted to full unit cells and their band
structures are recomputed by the high-fidelity FEM solver.

The band-gap targets are drawn from the high-fidelity FEM database used
in Section~\ref{sec:forward-numerics}.
Among the $87{,}474$ unit cells, $56{,}849$ exhibit at least one gap
between the first ten bands, yielding $105{,}229$ gaps in total (some
cells contribute multiple gaps).
We discard gaps with FEM width $\le 0.01$, and for each remaining gap
we record the band index $p$, the lower and upper FEM edges
$g_{L,\mathrm{FEM}}$, $g_{U,\mathrm{FEM}}$, and the width
\[
  w_{\mathrm{FEM}} := g_{U,\mathrm{FEM}} - g_{L,\mathrm{FEM}} > 0.01.
\]
Using the same train/validation/test split of the unit-cell database as in
Section~\ref{sec:forward-numerics}, the training subset contains $35{,}858$
gaps with $w_{\mathrm{FEM}}>0.01$ and the validation subset contains $2{,}205$
such gaps.
These gaps are encoded as descriptors
$\mathbf g^{(i)} = \bigl(g_{\mathrm{mid}}^{(i)}, w_{\mathrm{FEM}}^{(i)}, p^{(i)}\bigr)$
and form the training and validation sets for the gap-to-structure inverse
network.
For quantitative evaluation, we select at random
$N_{\mathrm{test}}^{\mathrm{gap}} = 500$ gaps with $w_{\mathrm{FEM}}>0.01$
from the test set and use these targets in the error analysis below.

\begin{table}[htbp]
  \centering
  \caption{Sample-wise inverse-design errors on the testing set. Representative samples are ranked by ${\rm MRE}_i$.}
  \label{tab:inverse-sample-errors-pod-mlp}
  \setlength{\tabcolsep}{6pt}
  {\small\begin{tabular}{l c c c c c c}
    \toprule
    \multirow{2}{*}{Sample/Statistic}
    & \multicolumn{3}{c}{POD--DeepONet-based}
    & \multicolumn{3}{c}{MLP-based} \\
    \cmidrule(lr){2-4}\cmidrule(lr){5-7}
    & ${\rm RMSE}_i$ & ${\rm relRMSE}_i$ & ${\rm MRE}_i$ (\%)
    & ${\rm RMSE}_i$ & ${\rm relRMSE}_i$ & ${\rm MRE}_i$ (\%) \\
    \midrule
    Mean
      & $1.067\times10^{-2}$ & $1.112\times10^{-2}$ & $0.989$
      & $2.125\times10^{-2}$ & $2.473\times10^{-2}$ & $1.615$ \\
    Max
      & $1.214\times10^{-1}$ & $1.314\times10^{-1}$ & $10.578$
      & $2.073\times10^{-1}$ & $2.202\times10^{-1}$ & $16.232$ \\
    \midrule
    Best
      & $0$ & $0$ & $0$
      & $0$ & $0$ & $0$ \\
    Median
      & $5.501\times10^{-3}$ & $6.172\times10^{-3}$ & $1.041$
      & $2.374\times10^{-2}$ & $3.614\times10^{-2}$ & $2.58$ \\
    Worst
      & $2.292\times10^{-2}$ & $2.736\times10^{-2}$ & $10.578$
      & $1.673\times10^{-1}$ & $1.802\times10^{-1}$ & $16.232$ \\
    \bottomrule
  \end{tabular}}\label{tab:inverse-best}
\end{table}

\begin{table}[htbp]
  \centering
  \caption{Band-wise inverse-design errors on the testing set.}
  \label{tab:inverse-band-errors-pod-mlp}
  \setlength{\tabcolsep}{6pt}
  {\small\begin{tabular}{c c c c c c c}
    \toprule
    \multirow{2}{*}{Band $n$}
    & \multicolumn{3}{c}{POD--DeepONet-based}
    & \multicolumn{3}{c}{MLP-based} \\
    \cmidrule(lr){2-4}\cmidrule(lr){5-7}
    & ${\rm RMSE}_n$ & ${\rm relRMSE}_n$ & ${\rm MRE}_n$ (\%)
    & ${\rm RMSE}_n$ & ${\rm relRMSE}_n$ & ${\rm MRE}_n$ (\%) \\
    \midrule
     1  & $6.442\times10^{-3}$ & $1.914\times10^{-2}$ & $0.933$
        & $1.306\times10^{-2}$ & $3.474\times10^{-2}$ & $1.774$ \\
     2  & $1.067\times10^{-2}$ & $2.006\times10^{-2}$ & $0.988$
        & $2.112\times10^{-2}$ & $3.999\times10^{-2}$ & $1.866$ \\
     3  & $1.151\times10^{-2}$ & $1.717\times10^{-2}$ & $0.941$
        & $2.438\times10^{-2}$ & $3.850\times10^{-2}$ & $1.791$ \\
     4  & $1.010\times10^{-2}$ & $1.320\times10^{-2}$ & $0.855$
        & $2.401\times10^{-2}$ & $3.072\times10^{-2}$ & $1.654$ \\
     5  & $1.337\times10^{-2}$ & $1.486\times10^{-2}$ & $0.979$
        & $3.322\times10^{-2}$ & $3.577\times10^{-2}$ & $1.188$ \\
     6  & $1.483\times10^{-2}$ & $1.531\times10^{-2}$ & $1.034$
        & $3.487\times10^{-2}$ & $3.598\times10^{-2}$ & $1.134$ \\
     7  & $1.487\times10^{-2}$ & $1.430\times10^{-2}$ & $1.004$
        & $3.772\times10^{-2}$ & $3.645\times10^{-2}$ & $1.179$ \\
     8  & $1.682\times10^{-2}$ & $1.510\times10^{-2}$ & $0.983$
        & $4.019\times10^{-2}$ & $3.755\times10^{-2}$ & $1.836$ \\
     9  & $1.861\times10^{-2}$ & $1.582\times10^{-2}$ & $1.030$
        & $4.519\times10^{-2}$ & $3.901\times10^{-2}$ & $1.901$ \\
    10  & $1.874\times10^{-2}$ & $1.755\times10^{-2}$ & $1.051$
        & $5.080\times10^{-2}$ & $4.033\times10^{-2}$ & $1.914$ \\
    \midrule
    Average
        & $1.360\times10^{-2}$ & $1.625\times10^{-2}$ & $0.980$
        & $3.246\times10^{-2}$ & $3.690\times10^{-2}$ & $1.624$ \\
    \bottomrule
  \end{tabular}}\label{tab:inverse-band-errors}
\end{table}

\begin{figure}[!t]
  \centering

  \subfigure[Best sample]{\label{fig:inv-a}
    \begin{minipage}[b]{0.08\textwidth}
      \centering
      \includegraphics[width=\linewidth,trim={3cm 0cm 2.5cm 0cm},clip]{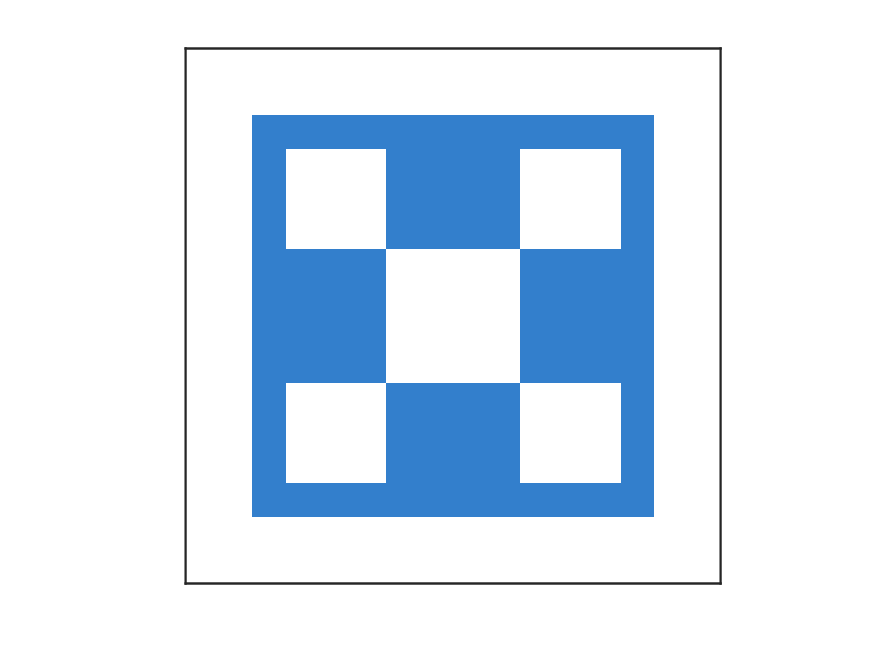}\\[2pt]
      \includegraphics[width=\linewidth,trim={3cm 0cm 2.5cm 0cm},clip]{deep_learning_imag/inverse1_cell_best.eps}
    \end{minipage}\hfill
    \begin{minipage}[b]{0.3\textwidth}
      \centering
      \includegraphics[width=\linewidth,trim={1cm 0cm 1cm 0cm},clip]{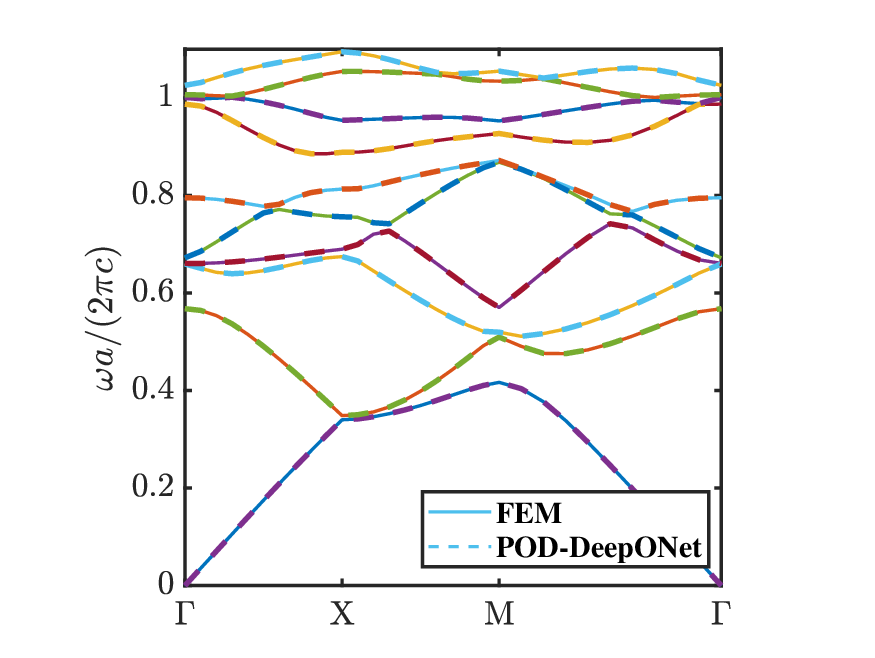}
    \end{minipage}
  }
  \subfigure[Median sample]{\label{fig:inv-b}
    \begin{minipage}[b]{0.08\textwidth}
      \centering
      \includegraphics[width=\linewidth,trim={3cm 0cm 2.5cm 0cm},clip]{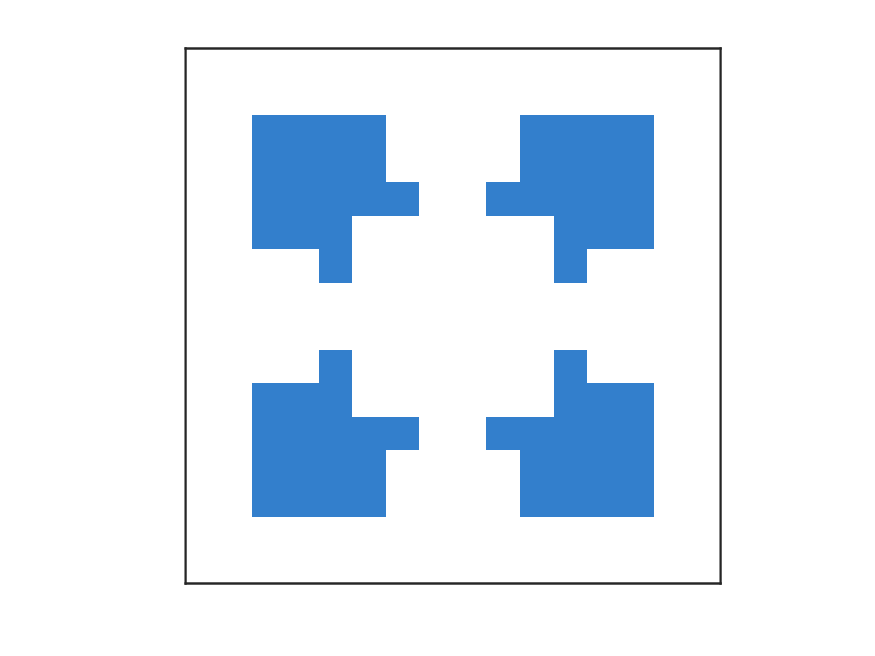}\\[2pt]
      \includegraphics[width=\linewidth,trim={3cm 0cm 2.5cm 0cm},clip]{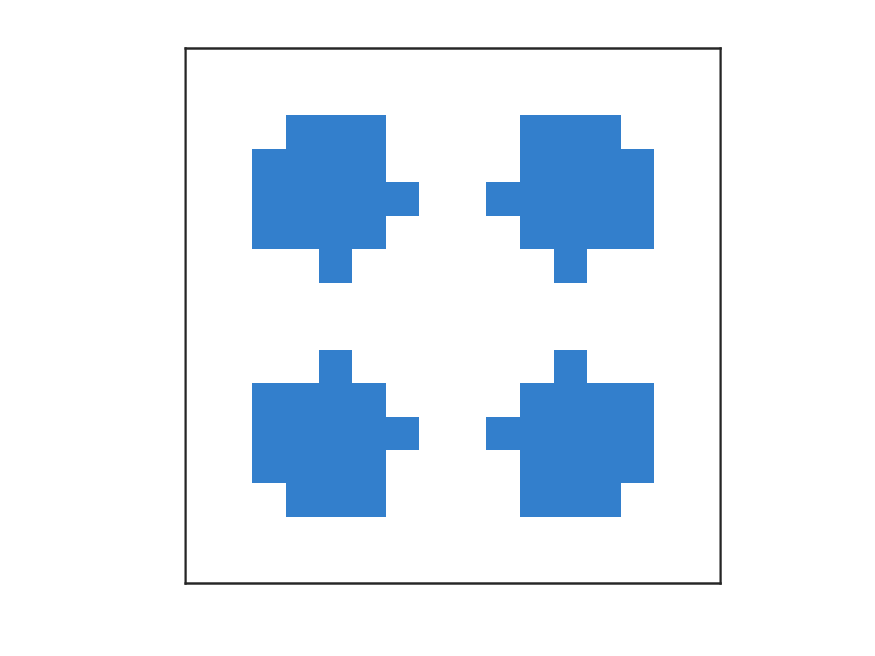}
    \end{minipage}
    \begin{minipage}[b]{0.3\textwidth}
      \centering
      \includegraphics[width=\linewidth,trim={1cm 0cm 1cm 0cm},clip]{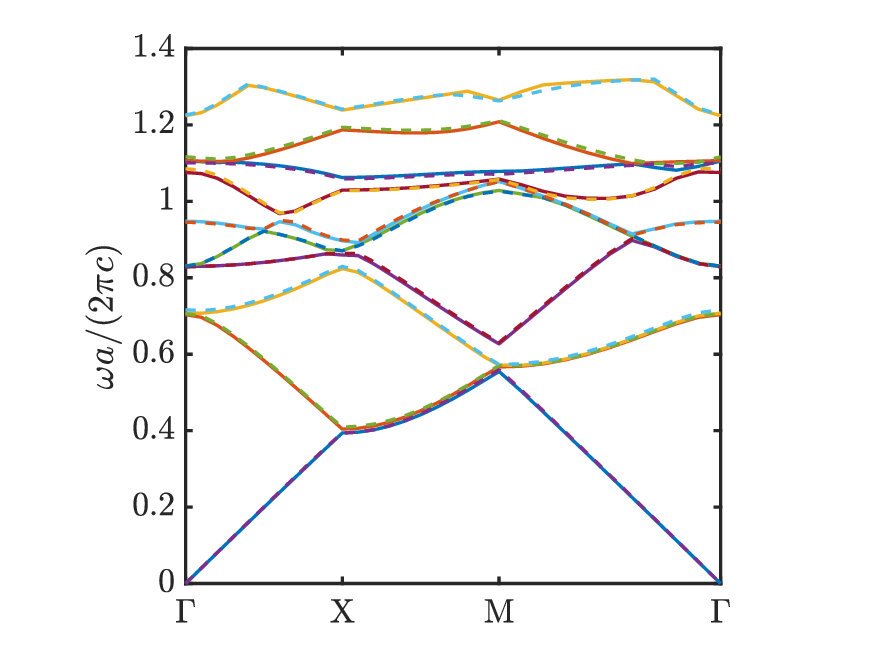}
    \end{minipage}
  }\\
  \subfigure[Worst Sample]{\label{fig:inv-c}
    \begin{minipage}[b]{0.08\textwidth}
      \centering
      \includegraphics[width=\linewidth,trim={3cm 0cm 2.5cm 0cm},clip]{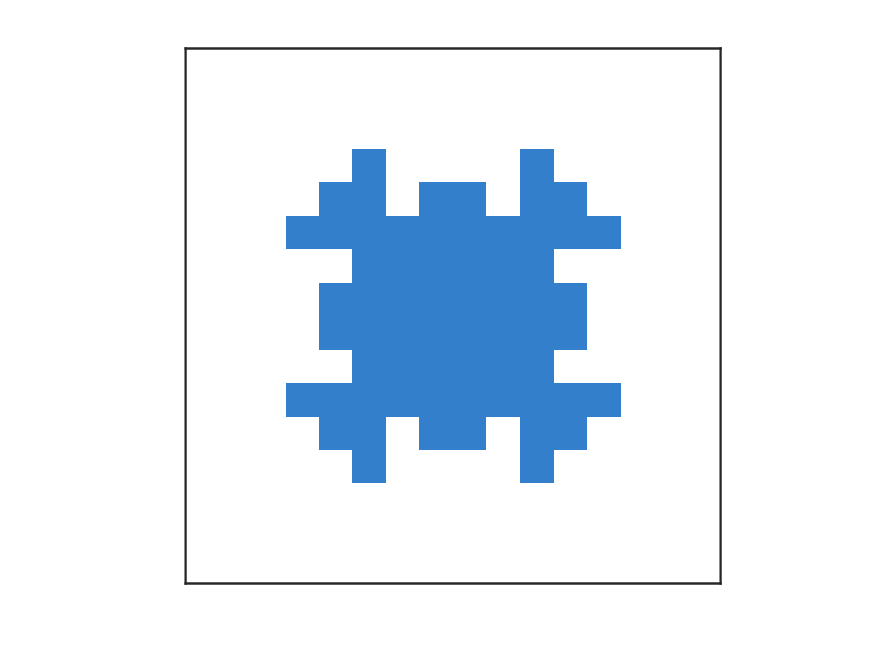}\\[2pt]
      \includegraphics[width=\linewidth,trim={3cm 0cm 2.5cm 0cm},clip]{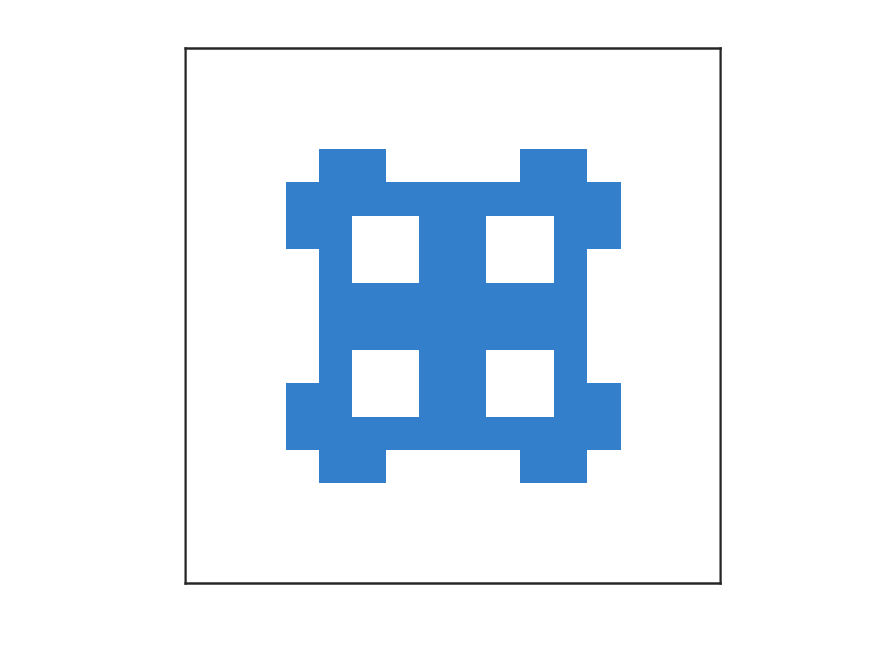}
    \end{minipage}
    \begin{minipage}[b]{0.3\textwidth}
      \centering
      \includegraphics[width=\linewidth,trim={1cm 0cm 1cm 0cm},clip]{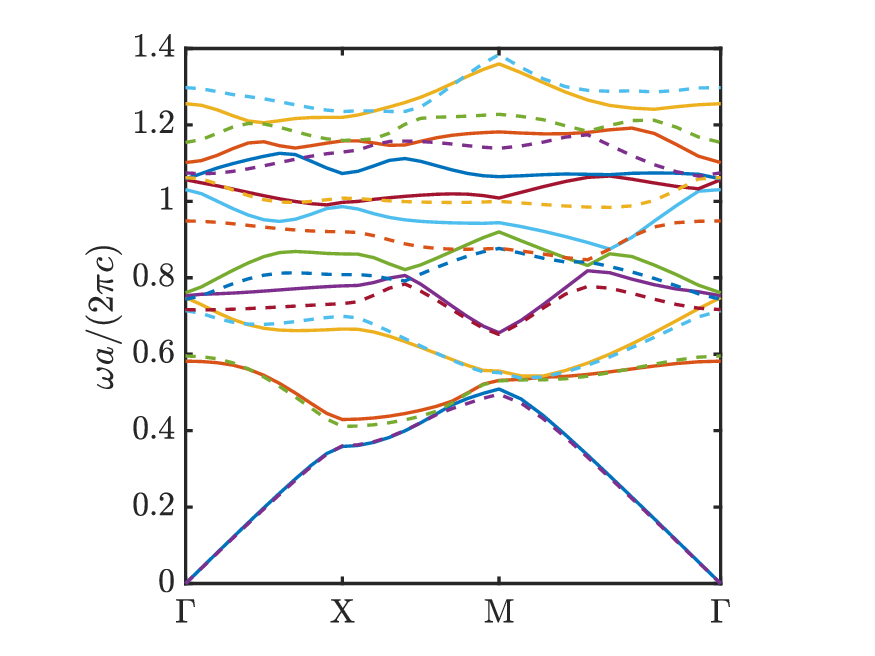}
    \end{minipage}
  }
  \subfigure[Random sample]{\label{fig:inv-d}
    \begin{minipage}[b]{0.08\textwidth}
      \centering
      \includegraphics[width=\linewidth,trim={3cm 0cm 2.5cm 0cm},clip]{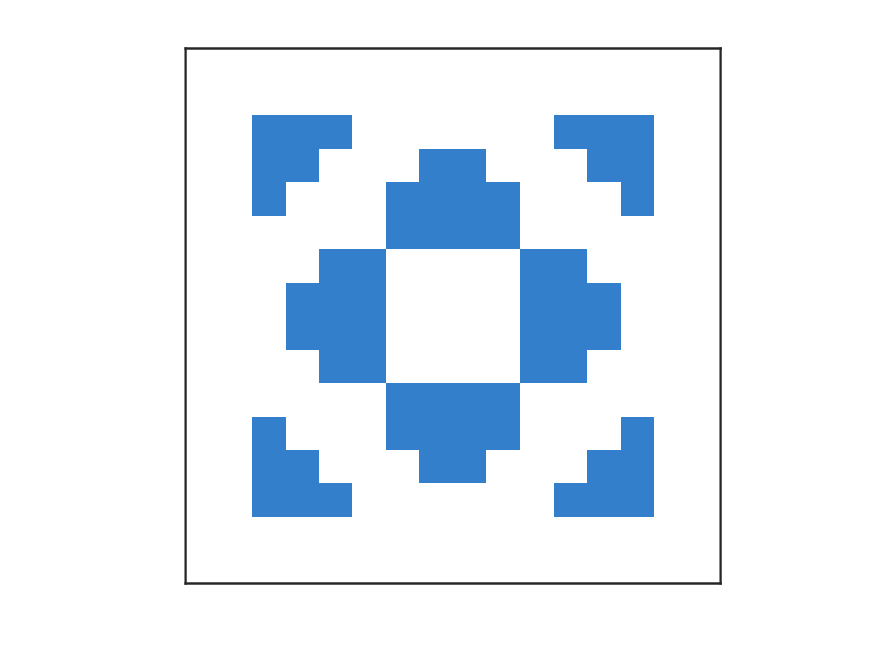}\\[2pt]
      \includegraphics[width=\linewidth,trim={3cm 0cm 2.5cm 0cm},clip]{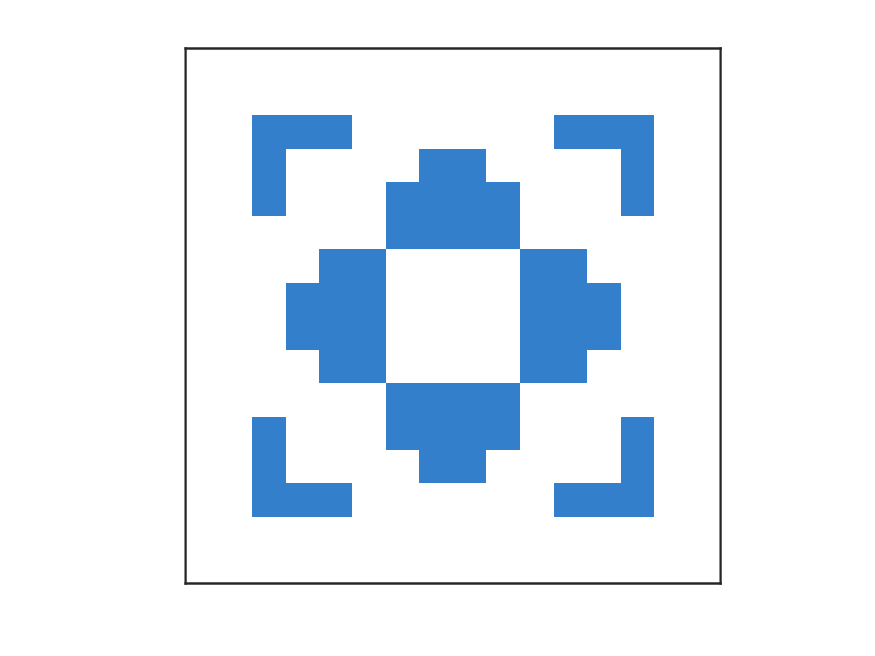}
    \end{minipage}
    \begin{minipage}[b]{0.3\textwidth}
      \centering
      \includegraphics[width=\linewidth,trim={1cm 0cm 1cm 0cm},clip]{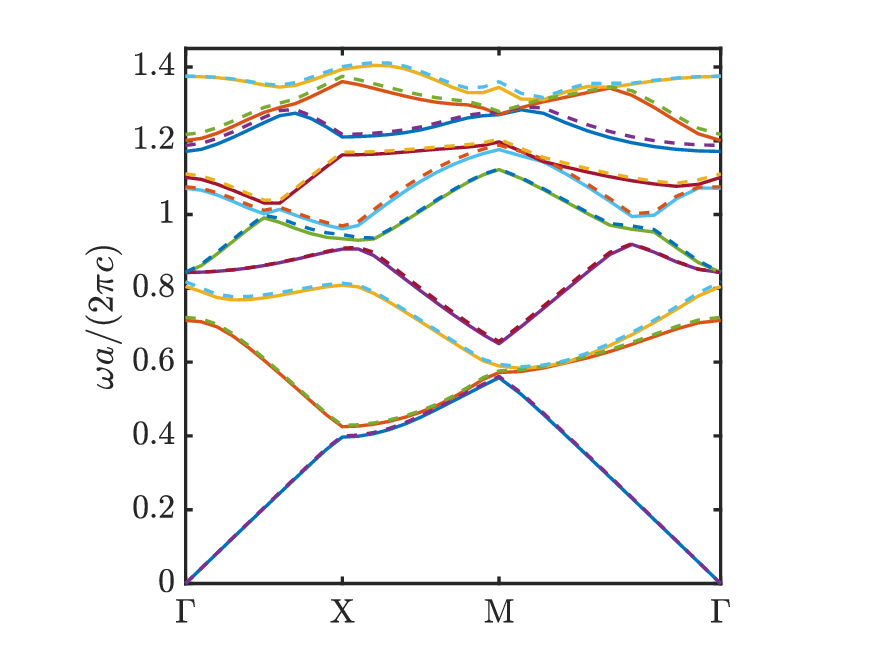}
    \end{minipage}
  }
  \caption{Representative inverse-design test samples.  In each subfigure,
the left column shows the target (top) and inverse-designed (bottom)
$16\times16$ unit cells, and the right panel shows the corresponding TE
band structures along high--symmetry path (solid:
target FEM bands; dashed: bands of the inverse-designed cell).}
  \label{fig:inverse-four-samples}
\end{figure}

After training, for each target $\mathbf g^{(p_i)}$ in the test set, we obtain an
optimized wedge vector
$\boldsymbol{\rho}_{\mathrm{inv}}^{(p_j)}$. Then, we reconstruct the
corresponding unit cell and recompute its band structure by the
high-fidelity FEM solver.
From this band diagram, we extract the band gap
\[
  \bigl[g_{L,\mathrm{pred}}^{(p_j)},g_{U,\mathrm{pred}}^{(p_j)}\bigr],
  \qquad
  w_{\mathrm{pred}}^{(p_j)}
  := g_{U,\mathrm{pred}}^{(p_j)} - g_{L,\mathrm{pred}}^{(p_j)},
\]
between bands $p_j$ and $p_j+1$.
We quantify the relative errors in the lower and upper edges and in the
width by
\begin{equation}
  e_{\mathrm{low}}^{(p_j)} 
  = \left|\frac{g_{L,\mathrm{pred}}^{(p_j)} - g_{L,\mathrm{FEM}}^{(p_j)}}
               {g_{L,\mathrm{FEM}}^{(p_j)}}\right|,
  \quad
  e_{\mathrm{up}}^{(p_j)} 
  = \left|\frac{g_{U,\mathrm{pred}}^{(p_j)} - g_{U,\mathrm{FEM}}^{(p_j)}}
               {g_{U,\mathrm{FEM}}^{(p_j)}}\right|,
  \quad
  e_{\mathrm{w}}^{(p_j)} 
  = \left|\frac{w_{\mathrm{pred}}^{(p_j)} - w_{\mathrm{FEM}}^{(p_j)}}
               {w_{\mathrm{FEM}}^{(p_j)}}\right|,
  \label{eq:gap-errors-sample}
\end{equation}
and summarize them by the arithmetic means
\begin{equation}
  \mathrm{rMAE}_{\mathrm{low}}
  = \frac{1}{N_{\mathrm{test}}^{\mathrm{gap}}}\sum_{j=1}^{N_{\mathrm{test}}^{\mathrm{gap}}}e_{\mathrm{low}}^{(p_j)},
  \quad
  \mathrm{rMAE}_{\mathrm{up}}
  = \frac{1}{N_{\mathrm{test}}^{\mathrm{gap}}}\sum_{j=1}^{N_{\mathrm{test}}^{\mathrm{gap}}}e_{\mathrm{up}}^{(p_j)},
  \quad
  \mathrm{rMAE}_{\mathrm{w}}
  = \frac{1}{N_{\mathrm{test}}^{\mathrm{gap}}}\sum_{j=1}^{N_{\mathrm{test}}^{\mathrm{gap}}}e_{\mathrm{w}}^{(p_j)}.
  \label{eq:gap-rmae}
\end{equation}
We also examine how well the predicted FEM gap overlaps the reference
FEM gap.
For each experiment, we set
\[
  I_{\mathrm{true}}^{(p_j)} 
  = \bigl[g_{L,\mathrm{FEM}}^{(p_j)},g_{U,\mathrm{FEM}}^{(p_j)}\bigr],
  \qquad
  I_{\mathrm{pred}}^{(p_j)} 
  = \bigl[g_{L,\mathrm{pred}}^{(p_j)},g_{U,\mathrm{pred}}^{(p_j)}\bigr],
\]
and define the lengths of their intersection and union as
\begin{align}
  \ell_{\mathrm{overlap}}^{(p_j)} 
  &= \max\Bigl\{0,\;
      \min\bigl(g_{U,\mathrm{FEM}}^{(p_j)},g_{U,\mathrm{pred}}^{(p_j)}\bigr)
      - \max\bigl(g_{L,\mathrm{FEM}}^{(p_j)},g_{L,\mathrm{pred}}^{(p_j)}\bigr)\Bigr\},
  \label{eq:gap-overlap-len}
  \\
  \ell_{\mathrm{union}}^{(p_j)} 
  &= \max\bigl(g_{U,\mathrm{FEM}}^{(p_j)},g_{U,\mathrm{pred}}^{(p_j)}\bigr)
     - \min\bigl(g_{L,\mathrm{FEM}}^{(p_j)},g_{L,\mathrm{pred}}^{(p_j)}\bigr).
  \label{eq:gap-union-len}
\end{align}
From these we form a Jaccard-type overlap ratio and two coverage ratios,
\begin{equation}
  R_{\mathrm{overlap}}^{(p_j)} 
  = \frac{\ell_{\mathrm{overlap}}^{(p_j)}}{\ell_{\mathrm{union}}^{(p_j)}},
  \quad
  R_{\mathrm{cov,true}}^{(p_j)} 
  = \frac{\ell_{\mathrm{overlap}}^{(p_j)}}{w_{\mathrm{FEM}}^{(p_j)}},
  \quad
  R_{\mathrm{cov,pred}}^{(p_j)} 
  = \frac{\ell_{\mathrm{overlap}}^{(p_j)}}{w_{\mathrm{pred}}^{(p_j)}}.
  \label{eq:gap-overlap-ratios}
\end{equation}
We report these statistics through their arithmetic means
$\overline{R}_{\mathrm{overlap}}$, 
$\overline{R}_{\mathrm{cov,true}}$, and 
$\overline{R}_{\mathrm{cov,pred}}$.
Finally, we declare an inverse design \emph{successful} if the predicted
band gap contains the reference FEM gap up to a small relative tolerance.
For a prescribed $\tau_{\mathrm{gap}}>0$ we define
\begin{equation}
  \chi_{\mathrm{succ}}^{(p_j)}
  = \mathbf{1}\Bigl\{
      g_{L,\mathrm{pred}}^{(p_j)} 
      \le g_{L,\mathrm{FEM}}^{(p_j)} + \tau_{\mathrm{gap}}\,
          \bigl|g_{L,\mathrm{FEM}}^{(p_j)}\bigr|,
      \quad
      g_{U,\mathrm{pred}}^{(p_j)} 
      \ge g_{U,\mathrm{FEM}}^{(p_j)} - \tau_{\mathrm{gap}}\,
          \bigl|g_{U,\mathrm{FEM}}^{(p_j)}\bigr|
    \Bigr\},
  \label{eq:gap-success}
\end{equation}
and take $\tau_{\mathrm{gap}} = 5\times10^{-3}$ in the experiments.
The overall FEM-level success rate is
$N_{\mathrm{succ}}/N_{\mathrm{test}}^{\mathrm{gap}}$, where
$N_{\mathrm{succ}} := \sum_{j=1}^{N_{\mathrm{test}}^{\mathrm{gap}}}\chi_{\mathrm{succ}}^{(p_j)}$.

\begin{table}[htbp]
  \centering
  \caption{Aggregate statistics for the inverse band-gap design problem over all
$N_{\mathrm{test}}^{\mathrm{gap}}=500$ targets.}
  \label{tab:inverse-gap-stats}
  \setlength{\tabcolsep}{8pt}
  \begin{tabular}{l c c}
    \toprule
    & POD--DeepONet-based & MLP-based \\
    \midrule
    $\mathrm{rMAE}_{\mathrm{low}}$
      & $9.01\times10^{-3}$
      & $1.047\times10^{-2}$ \\
    $\mathrm{rMAE}_{\mathrm{up}}$
      & $8.17\times10^{-3}$
      & $9.288\times10^{-3}$ \\
    $\mathrm{rMAE}_{\mathrm{w}}$
      & $5.78\times10^{-2}$
      & $2.704\times10^{-1}$ \\
    \midrule
    $\overline{R}_{\mathrm{overlap}}$
      & $0.828$
      & $0.612$ \\
    $\overline{R}_{\mathrm{cov,true}}$
      & $0.874$
      & $0.675$ \\
    $\overline{R}_{\mathrm{cov,pred}}$
      & $0.902$
      & $0.674$ \\
    \midrule
    $N_{\mathrm{succ}}/N_{\mathrm{test}}^{\mathrm{gap}}$
      & $0.929$
      & $0.651$ \\
    \bottomrule
  \end{tabular}
\end{table}

Table~\ref{tab:inverse-gap-stats} summarizes the performance of the inverse band-gap design problem over all
$N_{\mathrm{test}}^{\mathrm{gap}}=500$ targets.
For the POD--DeepONet-based pipeline, the mean relative errors in the lower and upper gap edges,
$\mathrm{rMAE}_{\mathrm{low}}=9.01\times10^{-3}$ and
$\mathrm{rMAE}_{\mathrm{up}}=8.17\times10^{-3}$, are both below
one percent, indicating that the inverse designs place the gap edges very
close to their FEM references.
Compared with the MLP-based counterpart, these edge errors are consistently smaller,
but the more decisive advantage appears in the gap-width recovery:
$\mathrm{rMAE}_{\mathrm{w}}$ is reduced from $2.704\times10^{-1}$ to
$5.78\times10^{-2}$, demonstrating that the POD-constrained inverse map
controls compounded edge misalignments far more effectively.
This suggests that the POD-style trunk provides a structurally compatible
representation of band information, which yields a better-conditioned
inverse constraint than treating the full band data as a generic output.
On average, the FEM gap produced by the optimized unit cell overlaps
$82.8\%$ of the union of the target and predicted intervals
($\overline{R}_{\mathrm{overlap}}=0.828$), covers $87.4\%$ of the target
width ($\overline{R}_{\mathrm{cov,true}}=0.874$), and retains $90.2\%$ of
its own width inside the target interval
($\overline{R}_{\mathrm{cov,pred}}=0.902$).
All three coverage measures are markedly higher than those of the MLP-based pipeline, indicating that the POD--DeepONet
constraint better preserves both the location and the extent of the target spectral window.
Consequently, most inverse designs satisfy the success criterion in
\eqref{eq:gap-success}, yielding an overall success rate of
$N_{\mathrm{succ}}/N_{\mathrm{test}}^{\mathrm{gap}}=0.929$,
substantially exceeding the MLP-based success rate of $0.651$.

Figure~\ref{fig:inverse-gap-intervals} provides a complementary visualization of the $500$
POD--DeepONet-based inverse-design experiments. Each horizontal pair of segments corresponds
to one target band gap: the red segment indicates the prescribed FEM gap interval, while the
yellow segment reports the FEM gap obtained from the unit cell produced by the inverse-design
pipeline. For readability, the $500$ samples are split into two panels: the left panel shows
samples $1$--$250$, and the right panel shows samples $251$--$500$. Across most targets, the
yellow segments closely overlap the red ones, visually confirming the small endpoint and width
errors as well as the high success rate reported in Table~\ref{tab:inverse-gap-stats}. Noticeable deviations are
relatively rare and occur primarily in the narrow-gap, high-frequency
regime.
Figure~\ref{fig:inv-gap-cases} shows two randomly selected examples of the inverse band-gap design problem; although the predicted unit cells differ markedly from the corresponding true cells, their band diagrams still almost satisfy the target $\mathbf g=(a,b,p)$ constraints, illustrating the non-uniqueness of the inverse mapping.

\begin{figure}[htbp]
  \centering
  \includegraphics[width=\linewidth,trim={4cm 0cm 3cm 0cm},clip]{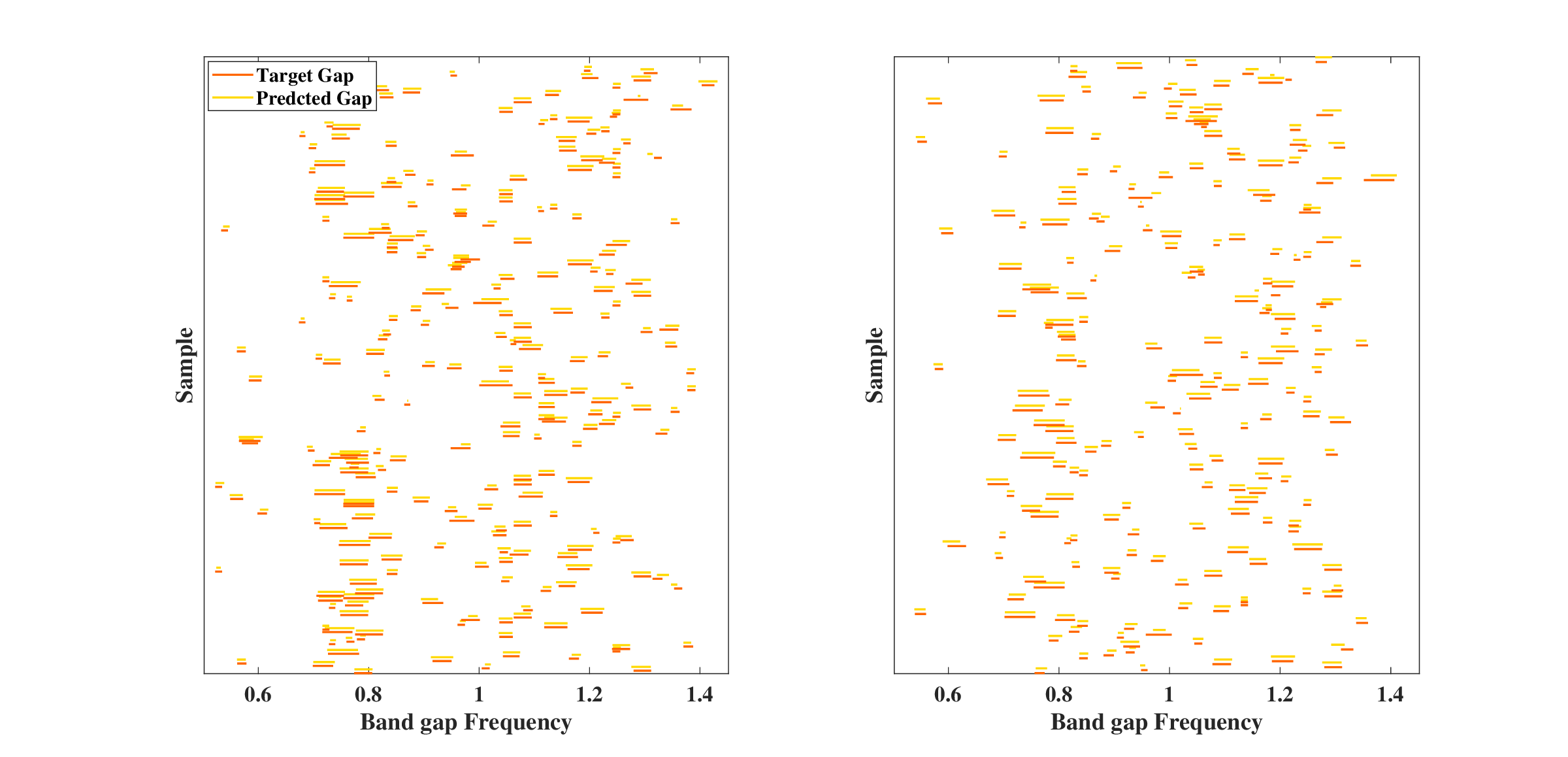}
  \caption{FEM band-gap intervals for $500$ inverse-design targets. For each sample, the red
segment denotes the prescribed target gap, and the yellow segment shows the FEM gap achieved by the inverse-designed unit cell. The samples are split into two panels for clarity: left,
samples $1$--$250$; right, samples $251$--$500$.}
  \label{fig:inverse-gap-intervals}
\end{figure}

\begin{figure}[t]
  \centering
  \setlength{\tabcolsep}{2pt}
  \renewcommand{\arraystretch}{1}

  \begin{tabular}{ccc}
    \multicolumn{3}{l}{\small \textbf{Example A (target $\mathbf g=(1.2409,1.2547,9)$)}}\\[-1mm]
    \includegraphics[width=0.2\textwidth]{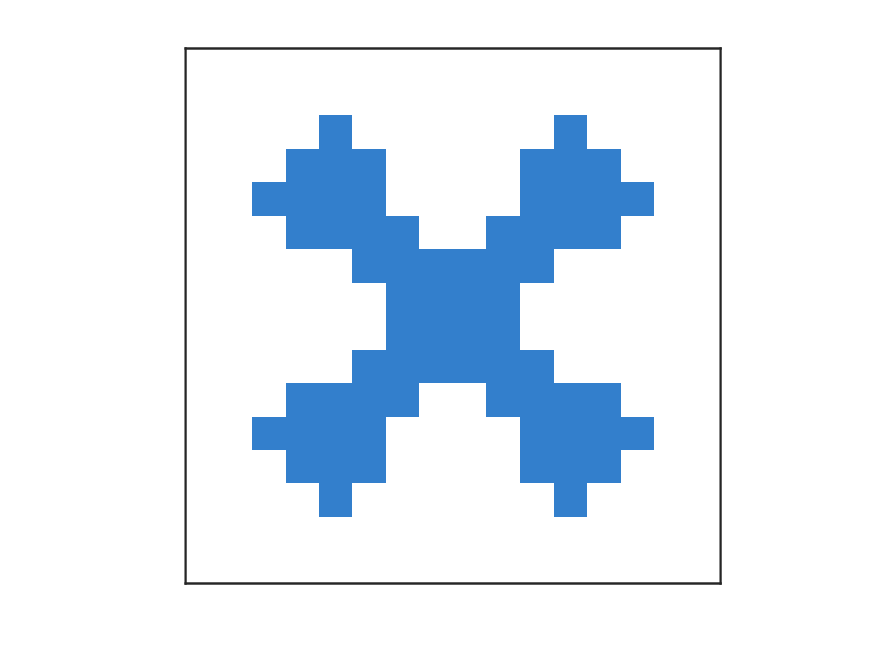} &
    \includegraphics[width=0.3\textwidth]{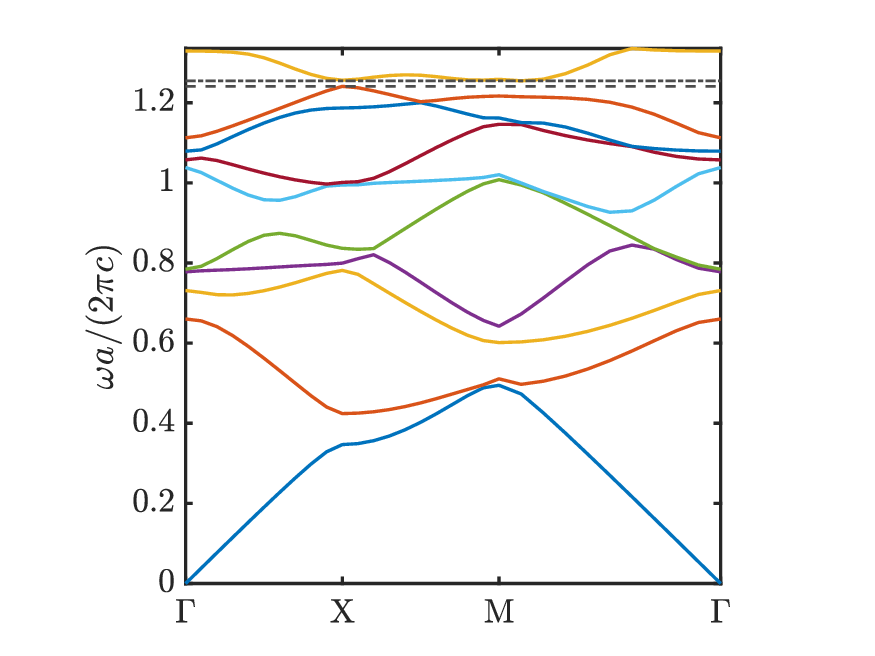} &
    \includegraphics[width=0.3\textwidth]{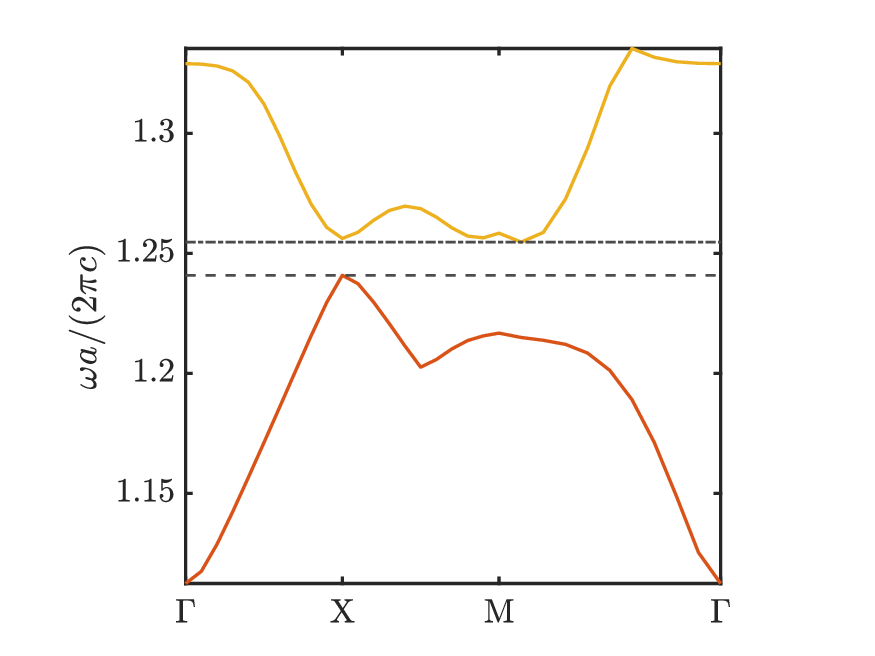} \\
    \scriptsize (a) True cell &
    \scriptsize (b) True bands &
    \scriptsize (c) True zoom \\[1mm]
    \includegraphics[width=0.2\textwidth]{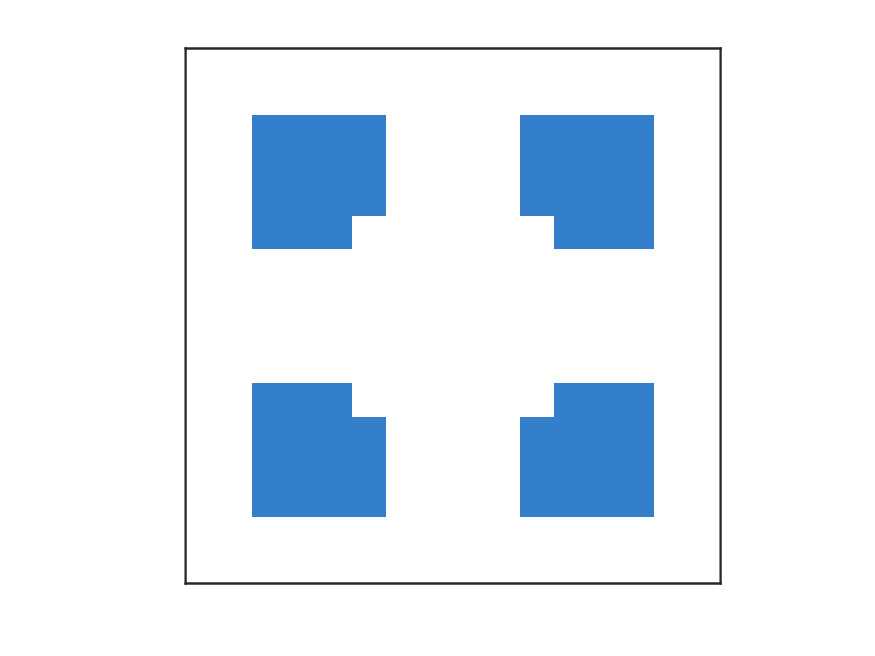} &
    \includegraphics[width=0.3\textwidth]{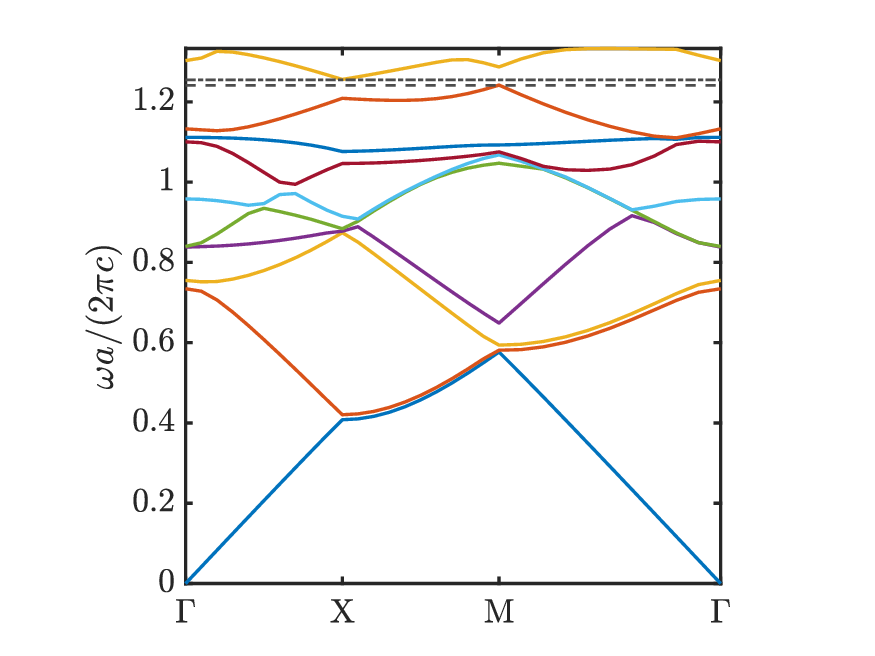} &
    \includegraphics[width=0.3\textwidth]{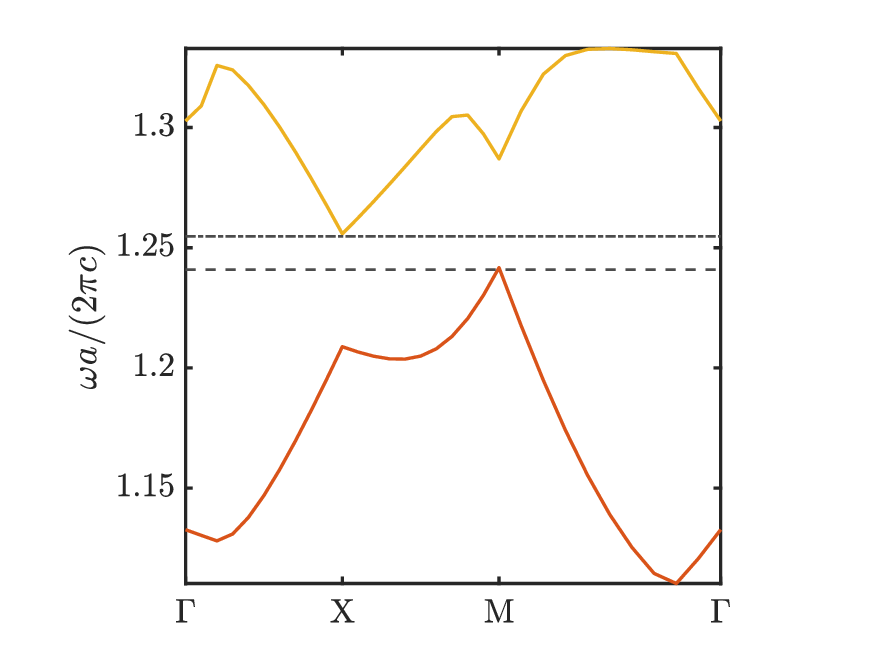} \\
    \scriptsize (d) Pred. cell &
    \scriptsize (e) Pred. bands &
    \scriptsize (f) Pred. zoom \\
  \end{tabular}

  \vspace{2mm}

  \begin{tabular}{ccc}
    \multicolumn{3}{l}{\small \textbf{Example B (target $\mathbf g=(1.1978,1.2234
,9)$)}}\\[-1mm]
    \includegraphics[width=0.2\textwidth]{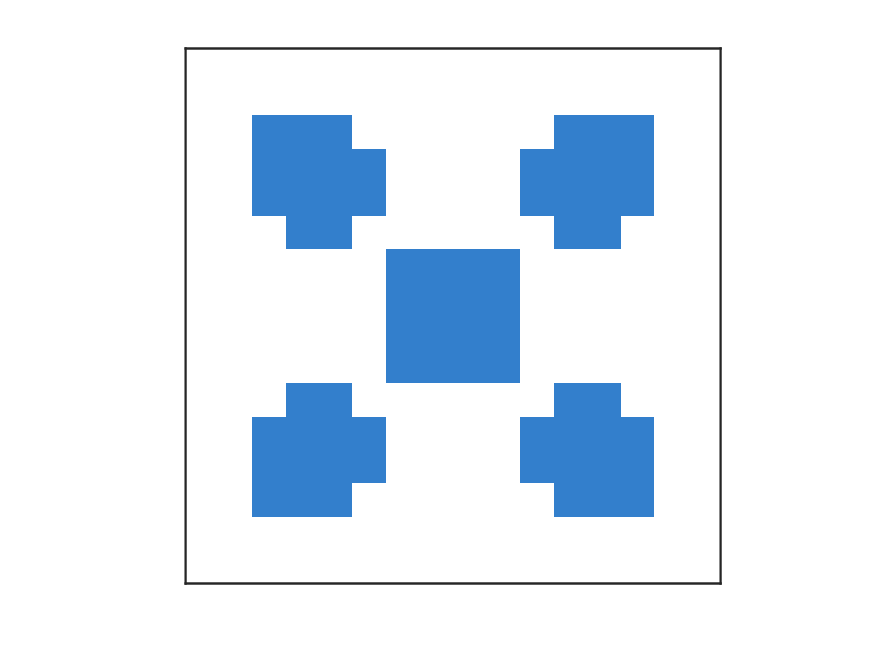} &
    \includegraphics[width=0.3\textwidth]{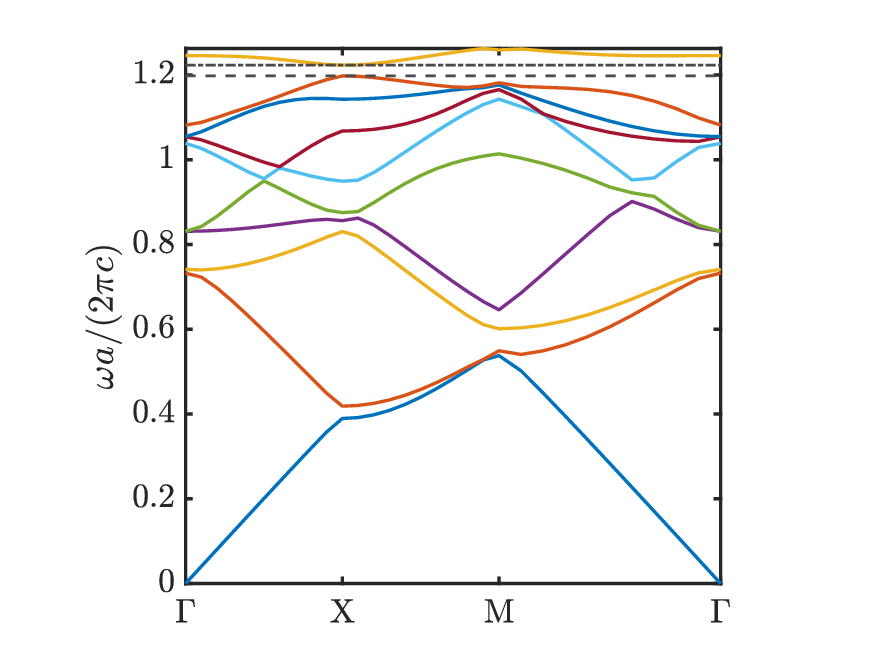} &
    \includegraphics[width=0.3\textwidth]{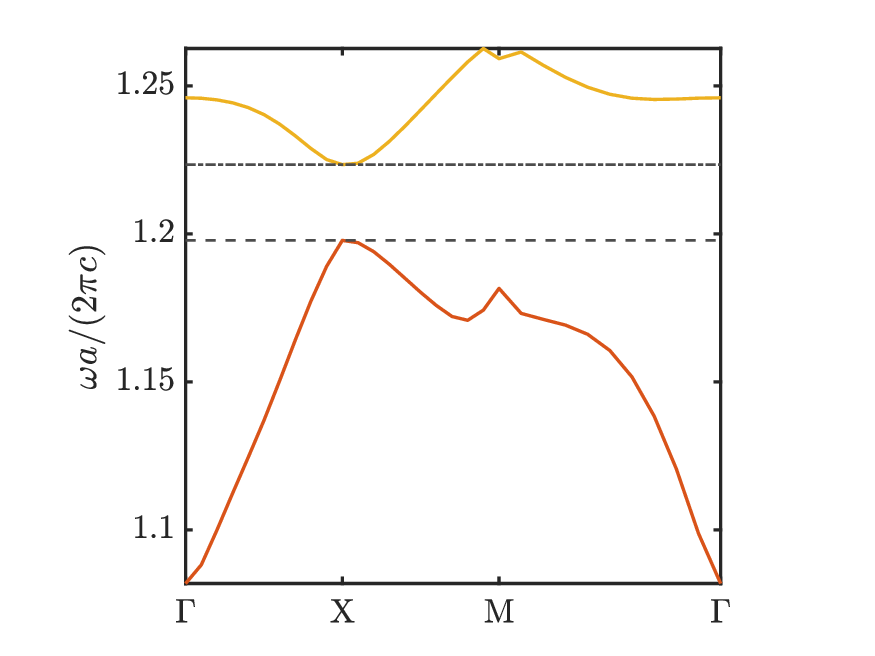} \\
    \scriptsize (g) True cell &
    \scriptsize (h) True bands &
    \scriptsize (i) True zoom \\[1mm]
    \includegraphics[width=0.2\textwidth]{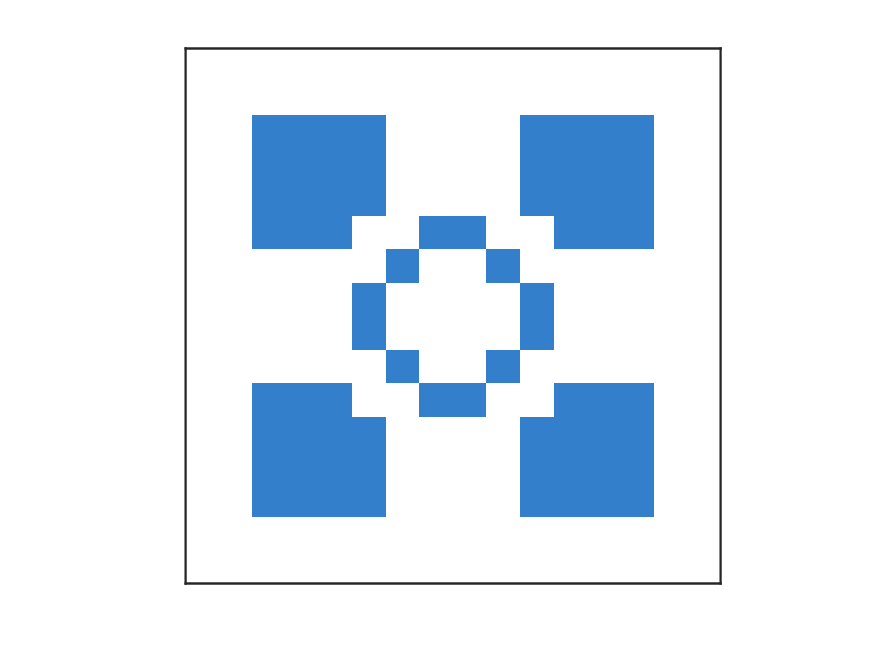} &
    \includegraphics[width=0.3\textwidth]{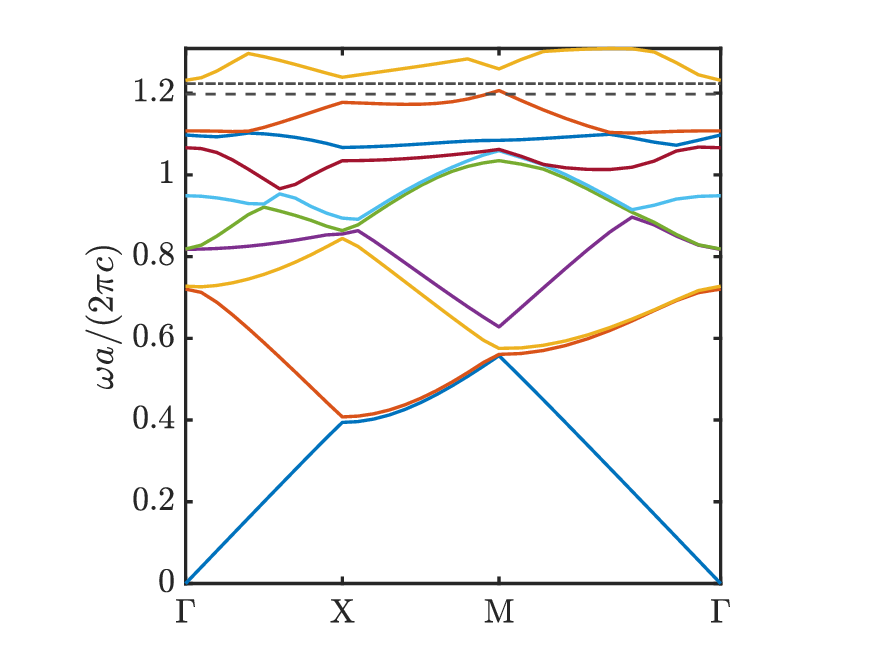} &
    \includegraphics[width=0.3\textwidth]{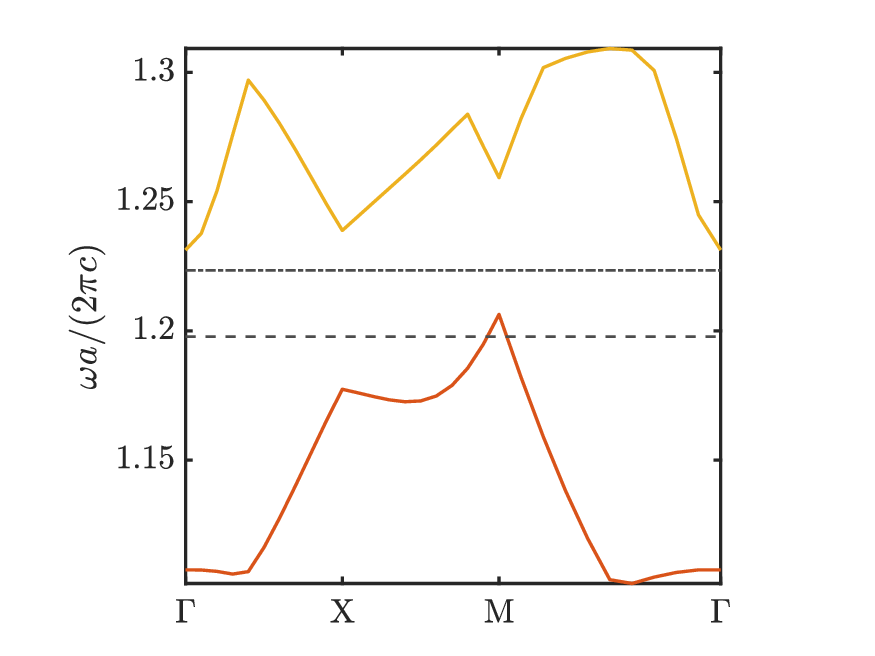} \\
    \scriptsize (j) Pred. cell &
    \scriptsize (k) Pred. bands &
    \scriptsize (l) Pred. zoom \\
  \end{tabular}

  \caption{Representative results for the inverse band-gap design problem.
  Each case compares the true and predicted unit cells, the corresponding band
  diagrams along $\mathcal K_{\rm hs}$, and a zoomed view near the target gap.}
  \label{fig:inv-gap-cases}
\end{figure}

\section{Conclusion}\label{sec:conclusion}

In this work, we developed a POD--DeepONet framework for forward and inverse
band-structure analysis of 2D photonic crystals with binary, pixel-based
$p4m$-symmetric unit cells. On the forward problem, the POD trunk is extracted
from high-fidelity band snapshots, and, combined with the branch network, yields
a compact and differentiable surrogate that achieves high accuracy for the
band function prediction, with an average error of about
$0.46\%$. Using this surrogate, we proposed two inverse design strategies.
For the dispersion-to-structure task, the inverse network attains
approximately $1\%$ average band-wise accuracy, indicating that the proposed
end-to-end differentiable formulation can reliably recover unit cells that
match target dispersion data. Moreover, even for the more challenging
band-gap inverse problem, the numerical results remain encouraging and show
that the learned designs can open gaps that tightly bracket prescribed target
intervals.
The present study focuses on a discrete band-map setting along a prescribed
high-symmetry path, where the POD trunk is constructed from snapshots on a
fixed $\mathbf k$-grid. While this choice enables a compact and accurate reduced
representation for efficient forward evaluation and end-to-end inverse
optimization, it also implies that gap identification and band-function
prediction are tied to the adopted path and sampling resolution. In addition,
although our formulation is compatible with the discrete ordering of the
lowest $N_b$ bands used in training, we do not explicitly address band-tracking
stabilization near crossings or degeneracies in the current implementation.
Future work will therefore pursue methodological extensions that relax these
constraints. First, we will develop continuous-in-$\mathbf k$ or multi-resolution
trunk constructions, such as path-coordinate parameterizations and adaptive
snapshot enrichment, to support discretization-invariant evaluation and more
reliable gap certification under refined or alternative $\mathbf k$-samplings. Second,
we will investigate band-tracking and crossing-aware objectives that improve
label consistency and gradient stability in inverse optimization. Third, we
will extend the present binary, symmetry-restricted parametrization to
multi-material and higher-resolution design spaces.
Finally, we will assess uncertainty-aware inverse
design by incorporating surrogate error indicators and fabrication-inspired
perturbations, with targeted validation against high-fidelity solvers.

\section*{Acknowledgment}
Guang Lin would like to thank the support of National Science Foundation (DMS-2533878, DMS-2053746, DMS-2134209, ECCS-2328241, CBET-2347401 and OAC-2311848), and U.S.~Department of Energy (DOE) Office of Science Advanced Scientific Computing Research program DE-SC0023161, the SciDAC LEADS Institute, and DOE–Fusion Energy Science, under grant number: DE-SC0024583. Guangliang Li would like to thank the support of Hong Kong RGC General Research Fund (Project number: 17309925).

\bibliographystyle{model1-num-names}
\bibliography{refer}

\begin{thebibliography}{63}
\expandafter\ifx\csname natexlab\endcsname\relax\def\natexlab#1{#1}\fi
\providecommand{\url}[1]{\texttt{#1}}
\providecommand{\href}[2]{#2}
\providecommand{\path}[1]{#1}
\providecommand{\DOIprefix}{doi:}
\providecommand{\ArXivprefix}{arXiv:}
\providecommand{\URLprefix}{URL: }
\providecommand{\Pubmedprefix}{pmid:}
\providecommand{\doi}[1]{\href{http://dx.doi.org/#1}{\path{#1}}}
\providecommand{\Pubmed}[1]{\href{pmid:#1}{\path{#1}}}
\providecommand{\bibinfo}[2]{#2}
\ifx\xfnm\relax \def\xfnm[#1]{\unskip,\space#1}\fi
\bibitem[{Yablonovitch(1987)}]{yablonovitch1987inhibited}
\bibinfo{author}{E.~Yablonovitch},
\newblock \bibinfo{title}{Inhibited spontaneous emission in solid-state physics
  and electronics},
\newblock \bibinfo{journal}{Physical review letters} \bibinfo{volume}{58}
  (\bibinfo{year}{1987}) \bibinfo{pages}{2059}.
\bibitem[{John(1987)}]{john1987strong}
\bibinfo{author}{S.~John},
\newblock \bibinfo{title}{Strong localization of photons in certain disordered
  dielectric superlattices},
\newblock \bibinfo{journal}{Physical review letters} \bibinfo{volume}{58}
  (\bibinfo{year}{1987}) \bibinfo{pages}{2486}.
\bibitem[{Johnson et~al.(2001)Johnson, Povinelli, and
  Joannopoulos}]{johnson2001new}
\bibinfo{author}{S.~G. Johnson}, \bibinfo{author}{M.~L. Povinelli},
  \bibinfo{author}{J.~D. Joannopoulos},
\newblock \bibinfo{title}{New photonic crystal system for integrated optics},
\newblock in: \bibinfo{booktitle}{Active and Passive Optical Components for WDM
  Communication}, volume \bibinfo{volume}{4532}, \bibinfo{organization}{SPIE},
  \bibinfo{year}{2001}, pp. \bibinfo{pages}{167--179}.
\bibitem[{Joannopoulos et~al.(2008)Joannopoulos, Johnson, Winn, and
  Meade}]{joannopoulos2008molding}
\bibinfo{author}{J.~D. Joannopoulos}, \bibinfo{author}{S.~G. Johnson},
  \bibinfo{author}{J.~N. Winn}, \bibinfo{author}{R.~D. Meade},
\newblock \bibinfo{title}{Molding the flow of light},
\newblock \bibinfo{journal}{Princeton Univ. Press, Princeton, NJ [ua]}
  (\bibinfo{year}{2008}).
\bibitem[{Notomi(2010)}]{notomi2010manipulating}
\bibinfo{author}{M.~Notomi},
\newblock \bibinfo{title}{Manipulating light with strongly modulated photonic
  crystals},
\newblock \bibinfo{journal}{Reports on Progress in Physics}
  \bibinfo{volume}{73} (\bibinfo{year}{2010}) \bibinfo{pages}{096501}.
\bibitem[{Painter et~al.(1999)Painter, Lee, Scherer, Yariv, O'brien, Dapkus,
  and Kim}]{painter1999two}
\bibinfo{author}{O.~Painter}, \bibinfo{author}{R.~Lee},
  \bibinfo{author}{A.~Scherer}, \bibinfo{author}{A.~Yariv},
  \bibinfo{author}{J.~O'brien}, \bibinfo{author}{P.~Dapkus},
  \bibinfo{author}{I.~Kim},
\newblock \bibinfo{title}{Two-dimensional photonic band-gap defect mode laser},
\newblock \bibinfo{journal}{Science} \bibinfo{volume}{284}
  (\bibinfo{year}{1999}) \bibinfo{pages}{1819--1821}.
\bibitem[{Kuchment(1993)}]{kuchment1993floquet}
\bibinfo{author}{P.~A. Kuchment}, \bibinfo{title}{Floquet theory for partial
  differential equations}, volume~\bibinfo{volume}{60},
  \bibinfo{publisher}{Springer Science \& Business Media},
  \bibinfo{year}{1993}.
\bibitem[{Ho et~al.(1990)Ho, Chan, and Soukoulis}]{ho1990existence}
\bibinfo{author}{K.~Ho}, \bibinfo{author}{C.~T. Chan}, \bibinfo{author}{C.~M.
  Soukoulis},
\newblock \bibinfo{title}{Existence of a photonic gap in periodic dielectric
  structures},
\newblock \bibinfo{journal}{Physical Review Letters} \bibinfo{volume}{65}
  (\bibinfo{year}{1990}) \bibinfo{pages}{3152}.
\bibitem[{Taflove et~al.(2005)Taflove, Hagness, and
  Piket-May}]{taflove2005computational}
\bibinfo{author}{A.~Taflove}, \bibinfo{author}{S.~C. Hagness},
  \bibinfo{author}{M.~Piket-May},
\newblock \bibinfo{title}{Computational electromagnetics: the finite-difference
  time-domain method},
\newblock \bibinfo{journal}{The Electrical Engineering Handbook}
  \bibinfo{volume}{3} (\bibinfo{year}{2005}) \bibinfo{pages}{15}.
\bibitem[{Qiu and He(2000)}]{qiu2000numerical}
\bibinfo{author}{M.~Qiu}, \bibinfo{author}{S.~He},
\newblock \bibinfo{title}{Numerical method for computing defect modes in
  two-dimensional photonic crystals with dielectric or metallic inclusions},
\newblock \bibinfo{journal}{Physical Review B} \bibinfo{volume}{61}
  (\bibinfo{year}{2000}) \bibinfo{pages}{12871}.
\bibitem[{Axmann and Kuchment(1999)}]{axmann1999efficient}
\bibinfo{author}{W.~Axmann}, \bibinfo{author}{P.~Kuchment},
\newblock \bibinfo{title}{An efficient finite element method for computing
  spectra of photonic and acoustic band-gap materials: I. scalar case},
\newblock \bibinfo{journal}{Journal of Computational Physics}
  \bibinfo{volume}{150} (\bibinfo{year}{1999}) \bibinfo{pages}{468--481}.
\bibitem[{Andonegui and Garcia-Adeva(2013)}]{andonegui2013finite}
\bibinfo{author}{I.~Andonegui}, \bibinfo{author}{A.~J. Garcia-Adeva},
\newblock \bibinfo{title}{The finite element method applied to the study of
  two-dimensional photonic crystals and resonant cavities},
\newblock \bibinfo{journal}{Optics Express} \bibinfo{volume}{21}
  (\bibinfo{year}{2013}) \bibinfo{pages}{4072--4092}.
\bibitem[{Setyawan and Curtarolo(2010)}]{setyawan2010high}
\bibinfo{author}{W.~Setyawan}, \bibinfo{author}{S.~Curtarolo},
\newblock \bibinfo{title}{High-throughput electronic band structure
  calculations: Challenges and tools},
\newblock \bibinfo{journal}{Computational materials science}
  \bibinfo{volume}{49} (\bibinfo{year}{2010}) \bibinfo{pages}{299--312}.
\bibitem[{Degirmenci and Landais(2013)}]{degirmenci2013finite}
\bibinfo{author}{E.~Degirmenci}, \bibinfo{author}{P.~Landais},
\newblock \bibinfo{title}{Finite element method analysis of band gap and
  transmission of two-dimensional metallic photonic crystals at terahertz
  frequencies},
\newblock \bibinfo{journal}{Applied optics} \bibinfo{volume}{52}
  (\bibinfo{year}{2013}) \bibinfo{pages}{7367--7375}.
\bibitem[{Li et~al.(2021)Li, Liu, Baronett, Liu, Wang, Li, Chen, Li, Zhu, and
  Chen}]{li2021computation}
\bibinfo{author}{J.~Li}, \bibinfo{author}{J.~Liu}, \bibinfo{author}{S.~A.
  Baronett}, \bibinfo{author}{M.~Liu}, \bibinfo{author}{L.~Wang},
  \bibinfo{author}{R.~Li}, \bibinfo{author}{Y.~Chen}, \bibinfo{author}{D.~Li},
  \bibinfo{author}{Q.~Zhu}, \bibinfo{author}{X.-Q. Chen},
\newblock \bibinfo{title}{Computation and data driven discovery of topological
  phononic materials},
\newblock \bibinfo{journal}{Nature communications} \bibinfo{volume}{12}
  (\bibinfo{year}{2021}) \bibinfo{pages}{1204}.
\bibitem[{Cersonsky et~al.(2021)Cersonsky, Antonaglia, Dice, and
  Glotzer}]{cersonsky2021diversity}
\bibinfo{author}{R.~K. Cersonsky}, \bibinfo{author}{J.~Antonaglia},
  \bibinfo{author}{B.~D. Dice}, \bibinfo{author}{S.~C. Glotzer},
\newblock \bibinfo{title}{The diversity of three-dimensional photonic
  crystals},
\newblock \bibinfo{journal}{Nature communications} \bibinfo{volume}{12}
  (\bibinfo{year}{2021}) \bibinfo{pages}{2543}.
\bibitem[{Wang et~al.(2025)Wang, Craster, and Li}]{wang2025hp}
\bibinfo{author}{Y.~Wang}, \bibinfo{author}{R.~Craster},
  \bibinfo{author}{G.~Li},
\newblock \bibinfo{title}{An hp-adaptive sampling algorithm for dispersion
  relation reconstruction of 3d photonic crystals},
\newblock \bibinfo{journal}{Journal of Computational Physics}
  \bibinfo{volume}{521} (\bibinfo{year}{2025}) \bibinfo{pages}{113572}.
\bibitem[{Sigmund and S{\o}ndergaard~Jensen(2003)}]{sigmund2003systematic}
\bibinfo{author}{O.~Sigmund}, \bibinfo{author}{J.~S{\o}ndergaard~Jensen},
\newblock \bibinfo{title}{Systematic design of phononic band--gap materials and
  structures by topology optimization},
\newblock \bibinfo{journal}{Philosophical Transactions of the Royal Society of
  London. Series A: Mathematical, Physical and Engineering Sciences}
  \bibinfo{volume}{361} (\bibinfo{year}{2003}) \bibinfo{pages}{1001--1019}.
\bibitem[{Men et~al.(2014)Men, Lee, Freund, Peraire, and
  Johnson}]{men2014robust}
\bibinfo{author}{H.~Men}, \bibinfo{author}{K.~Y. Lee}, \bibinfo{author}{R.~M.
  Freund}, \bibinfo{author}{J.~Peraire}, \bibinfo{author}{S.~G. Johnson},
\newblock \bibinfo{title}{Robust topology optimization of three-dimensional
  photonic-crystal band-gap structures},
\newblock \bibinfo{journal}{Optics express} \bibinfo{volume}{22}
  (\bibinfo{year}{2014}) \bibinfo{pages}{22632--22648}.
\bibitem[{Dalklint et~al.(2022)Dalklint, Wallin, Bertoldi, and
  Tortorelli}]{dalklint2022tunable}
\bibinfo{author}{A.~Dalklint}, \bibinfo{author}{M.~Wallin},
  \bibinfo{author}{K.~Bertoldi}, \bibinfo{author}{D.~Tortorelli},
\newblock \bibinfo{title}{Tunable phononic bandgap materials designed via
  topology optimization},
\newblock \bibinfo{journal}{Journal of the Mechanics and Physics of Solids}
  \bibinfo{volume}{163} (\bibinfo{year}{2022}) \bibinfo{pages}{104849}.
\bibitem[{Zhang et~al.(2021)Zhang, Xing, Liu, Luo, and
  Kang}]{zhang2021realization}
\bibinfo{author}{X.~Zhang}, \bibinfo{author}{J.~Xing},
  \bibinfo{author}{P.~Liu}, \bibinfo{author}{Y.~Luo},
  \bibinfo{author}{Z.~Kang},
\newblock \bibinfo{title}{Realization of full and directional band gap design
  by non-gradient topology optimization in acoustic metamaterials},
\newblock \bibinfo{journal}{Extreme Mechanics Letters} \bibinfo{volume}{42}
  (\bibinfo{year}{2021}) \bibinfo{pages}{101126}.
\bibitem[{Jia et~al.(2024)Jia, Bao, Luo, Wang, Zhang, and
  Kang}]{jia2024maximizing}
\bibinfo{author}{Z.~Jia}, \bibinfo{author}{Y.~Bao}, \bibinfo{author}{Y.~Luo},
  \bibinfo{author}{D.~Wang}, \bibinfo{author}{X.~Zhang},
  \bibinfo{author}{Z.~Kang},
\newblock \bibinfo{title}{Maximizing acoustic band gap in phononic crystals via
  topology optimization},
\newblock \bibinfo{journal}{International Journal of Mechanical Sciences}
  \bibinfo{volume}{270} (\bibinfo{year}{2024}) \bibinfo{pages}{109107}.
\bibitem[{Kao et~al.(2005)Kao, Osher, and Yablonovitch}]{kao2005maximizing}
\bibinfo{author}{C.~Y. Kao}, \bibinfo{author}{S.~Osher},
  \bibinfo{author}{E.~Yablonovitch},
\newblock \bibinfo{title}{Maximizing band gaps in two-dimensional photonic
  crystals by using level set methods},
\newblock \bibinfo{journal}{Applied Physics B} \bibinfo{volume}{81}
  (\bibinfo{year}{2005}) \bibinfo{pages}{235--244}.
\bibitem[{Cheng and Yang(2013)}]{cheng2013maximizing}
\bibinfo{author}{X.-l. Cheng}, \bibinfo{author}{J.~Yang},
\newblock \bibinfo{title}{Maximizing band gaps in two-dimensional photonic
  crystals in square lattices},
\newblock \bibinfo{journal}{Journal of the Optical Society of America A}
  \bibinfo{volume}{30} (\bibinfo{year}{2013}) \bibinfo{pages}{2314--2319}.
\bibitem[{Peurifoy et~al.(2018)Peurifoy, Shen, Jing, Yang, Cano-Renteria,
  DeLacy, Joannopoulos, Tegmark, and
  Solja{\v{c}}i{\'c}}]{peurifoy2018nanophotonic}
\bibinfo{author}{J.~Peurifoy}, \bibinfo{author}{Y.~Shen},
  \bibinfo{author}{L.~Jing}, \bibinfo{author}{Y.~Yang},
  \bibinfo{author}{F.~Cano-Renteria}, \bibinfo{author}{B.~G. DeLacy},
  \bibinfo{author}{J.~D. Joannopoulos}, \bibinfo{author}{M.~Tegmark},
  \bibinfo{author}{M.~Solja{\v{c}}i{\'c}},
\newblock \bibinfo{title}{Nanophotonic particle simulation and inverse design
  using artificial neural networks},
\newblock \bibinfo{journal}{Science advances} \bibinfo{volume}{4}
  (\bibinfo{year}{2018}) \bibinfo{pages}{eaar4206}.
\bibitem[{Tahersima et~al.(2019)Tahersima, Kojima, Koike-Akino, Jha, Wang, Lin,
  and Parsons}]{tahersima2019deep}
\bibinfo{author}{M.~H. Tahersima}, \bibinfo{author}{K.~Kojima},
  \bibinfo{author}{T.~Koike-Akino}, \bibinfo{author}{D.~Jha},
  \bibinfo{author}{B.~Wang}, \bibinfo{author}{C.~Lin},
  \bibinfo{author}{K.~Parsons},
\newblock \bibinfo{title}{Deep neural network inverse design of integrated
  photonic power splitters},
\newblock \bibinfo{journal}{Scientific reports} \bibinfo{volume}{9}
  (\bibinfo{year}{2019}) \bibinfo{pages}{1368}.
\bibitem[{Qiu et~al.(2021)Qiu, Wu, Luo, Yang, He, and
  Huang}]{qiu2021nanophotonic}
\bibinfo{author}{C.~Qiu}, \bibinfo{author}{X.~Wu}, \bibinfo{author}{Z.~Luo},
  \bibinfo{author}{H.~Yang}, \bibinfo{author}{G.~He},
  \bibinfo{author}{B.~Huang},
\newblock \bibinfo{title}{Nanophotonic inverse design with deep neural networks
  based on knowledge transfer using imbalanced datasets},
\newblock \bibinfo{journal}{Optics Express} \bibinfo{volume}{29}
  (\bibinfo{year}{2021}) \bibinfo{pages}{28406--28415}.
\bibitem[{Jiang et~al.(2022)Jiang, Zhu, Yin, Lu, Xie, and
  Yin}]{jiang2022dispersion}
\bibinfo{author}{W.~Jiang}, \bibinfo{author}{Y.~Zhu}, \bibinfo{author}{G.~Yin},
  \bibinfo{author}{H.~Lu}, \bibinfo{author}{L.~Xie}, \bibinfo{author}{M.~Yin},
\newblock \bibinfo{title}{Dispersion relation prediction and structure inverse
  design of elastic metamaterials via deep learning},
\newblock \bibinfo{journal}{Materials Today Physics} \bibinfo{volume}{22}
  (\bibinfo{year}{2022}) \bibinfo{pages}{100616}.
\bibitem[{Li et~al.(2020)Li, Ning, Liu, Yan, Luo, and Zhuang}]{li2020designing}
\bibinfo{author}{X.~Li}, \bibinfo{author}{S.~Ning}, \bibinfo{author}{Z.~Liu},
  \bibinfo{author}{Z.~Yan}, \bibinfo{author}{C.~Luo},
  \bibinfo{author}{Z.~Zhuang},
\newblock \bibinfo{title}{Designing phononic crystal with anticipated band gap
  through a deep learning based data-driven method},
\newblock \bibinfo{journal}{Computer Methods in Applied Mechanics and
  Engineering} \bibinfo{volume}{361} (\bibinfo{year}{2020})
  \bibinfo{pages}{112737}.
\bibitem[{Han et~al.(2022)Han, Han, and Li}]{han2022deep}
\bibinfo{author}{S.~Han}, \bibinfo{author}{Q.~Han}, \bibinfo{author}{C.~Li},
\newblock \bibinfo{title}{Deep-learning-based inverse design of phononic
  crystals for anticipated wave attenuation},
\newblock \bibinfo{journal}{Journal of Applied Physics} \bibinfo{volume}{132}
  (\bibinfo{year}{2022}).
\bibitem[{Wang et~al.(2024)Wang, Craster, and Li}]{wang2024predicting}
\bibinfo{author}{Y.~Wang}, \bibinfo{author}{R.~Craster},
  \bibinfo{author}{G.~Li},
\newblock \bibinfo{title}{Predicting band structures for 2d photonic crystals
  via deep learning},
\newblock \bibinfo{journal}{arXiv preprint arXiv:2411.06063}
  (\bibinfo{year}{2024}).
\bibitem[{Ma et~al.(2023)Ma, Hao, Yan, Jiang, Chen, and Tang}]{ma2023deep}
\bibinfo{author}{B.~Ma}, \bibinfo{author}{R.~Hao}, \bibinfo{author}{H.~Yan},
  \bibinfo{author}{H.~Jiang}, \bibinfo{author}{J.~Chen},
  \bibinfo{author}{K.~Tang},
\newblock \bibinfo{title}{Deep learning-based inverse design of the complete
  photonic band gap in two-dimensional photonic crystals},
\newblock \bibinfo{journal}{Current Nanoscience} \bibinfo{volume}{19}
  (\bibinfo{year}{2023}) \bibinfo{pages}{423--431}.
\bibitem[{Wan et~al.(2025)Wan, Zhang, Guo, and Zheng}]{wan2025deep}
\bibinfo{author}{X.-H. Wan}, \bibinfo{author}{Y.~Zhang}, \bibinfo{author}{Q.-H.
  Guo}, \bibinfo{author}{L.-Y. Zheng},
\newblock \bibinfo{title}{Deep learning-based inverse design of irregular
  phononic crystals},
\newblock \bibinfo{journal}{International Journal of Mechanical Sciences}
  (\bibinfo{year}{2025}) \bibinfo{pages}{110335}.
\bibitem[{Tran et~al.(2025)Tran, Nanthakumar, and Zhuang}]{tran2025deep}
\bibinfo{author}{T.~V. Tran}, \bibinfo{author}{S.~Nanthakumar},
  \bibinfo{author}{X.~Zhuang},
\newblock \bibinfo{title}{Deep learning-based framework for the on-demand
  inverse design of metamaterials with arbitrary target band gap},
\newblock \bibinfo{journal}{npj Artificial Intelligence} \bibinfo{volume}{1}
  (\bibinfo{year}{2025}) \bibinfo{pages}{2}.
\bibitem[{Song et~al.(2024)Song, Lee, Kim, and Min}]{song2024artificial}
\bibinfo{author}{J.~Song}, \bibinfo{author}{J.~Lee}, \bibinfo{author}{N.~Kim},
  \bibinfo{author}{K.~Min},
\newblock \bibinfo{title}{Artificial intelligence in the design of innovative
  metamaterials: A comprehensive review},
\newblock \bibinfo{journal}{International Journal of Precision Engineering and
  Manufacturing} \bibinfo{volume}{25} (\bibinfo{year}{2024})
  \bibinfo{pages}{225--244}.
\bibitem[{Deng et~al.(2024)Deng, Liu, and Shi}]{deng2024inverse}
\bibinfo{author}{R.~Deng}, \bibinfo{author}{W.~Liu}, \bibinfo{author}{L.~Shi},
\newblock \bibinfo{title}{Inverse design in photonic crystals},
\newblock \bibinfo{journal}{Nanophotonics} \bibinfo{volume}{13}
  (\bibinfo{year}{2024}) \bibinfo{pages}{1219--1237}.
\bibitem[{Chen et~al.(2022)Chen, Ogren, Daraio, Brinson, and
  Rudin}]{chen2022see}
\bibinfo{author}{Z.~Chen}, \bibinfo{author}{A.~Ogren},
  \bibinfo{author}{C.~Daraio}, \bibinfo{author}{L.~C. Brinson},
  \bibinfo{author}{C.~Rudin},
\newblock \bibinfo{title}{How to see hidden patterns in metamaterials with
  interpretable machine learning},
\newblock \bibinfo{journal}{Extreme Mechanics Letters} \bibinfo{volume}{57}
  (\bibinfo{year}{2022}) \bibinfo{pages}{101895}.
\bibitem[{Lu et~al.(2021)Lu, Jin, Pang, Zhang, and
  Karniadakis}]{lu2021learning}
\bibinfo{author}{L.~Lu}, \bibinfo{author}{P.~Jin}, \bibinfo{author}{G.~Pang},
  \bibinfo{author}{Z.~Zhang}, \bibinfo{author}{G.~E. Karniadakis},
\newblock \bibinfo{title}{Learning nonlinear operators via deeponet based on
  the universal approximation theorem of operators},
\newblock \bibinfo{journal}{Nature machine intelligence} \bibinfo{volume}{3}
  (\bibinfo{year}{2021}) \bibinfo{pages}{218--229}.
\bibitem[{Kovachki et~al.(2023)Kovachki, Li, Liu, Azizzadenesheli,
  Bhattacharya, Stuart, and Anandkumar}]{kovachki2023neural}
\bibinfo{author}{N.~Kovachki}, \bibinfo{author}{Z.~Li},
  \bibinfo{author}{B.~Liu}, \bibinfo{author}{K.~Azizzadenesheli},
  \bibinfo{author}{K.~Bhattacharya}, \bibinfo{author}{A.~Stuart},
  \bibinfo{author}{A.~Anandkumar},
\newblock \bibinfo{title}{Neural operator: Learning maps between function
  spaces with applications to pdes},
\newblock \bibinfo{journal}{Journal of Machine Learning Research}
  \bibinfo{volume}{24} (\bibinfo{year}{2023}) \bibinfo{pages}{1--97}.
\bibitem[{Eivazi et~al.(2024)Eivazi, Wittek, and Rausch}]{eivazi2024nonlinear}
\bibinfo{author}{H.~Eivazi}, \bibinfo{author}{S.~Wittek},
  \bibinfo{author}{A.~Rausch},
\newblock \bibinfo{title}{Nonlinear model reduction for operator learning},
\newblock \bibinfo{journal}{arXiv preprint arXiv:2403.18735}
  (\bibinfo{year}{2024}).
\bibitem[{Cheng et~al.(2025)Cheng, Sahadath, Yang, Pan, and
  Ji}]{cheng2025surrogate}
\bibinfo{author}{Q.~Cheng}, \bibinfo{author}{M.~H. Sahadath},
  \bibinfo{author}{H.~Yang}, \bibinfo{author}{S.~Pan}, \bibinfo{author}{W.~Ji},
\newblock \bibinfo{title}{Surrogate modeling of heat transfer under flow
  fluctuation conditions using fourier basis-deep operator network with
  uncertainty quantification},
\newblock \bibinfo{journal}{Progress in Nuclear Energy} \bibinfo{volume}{188}
  (\bibinfo{year}{2025}) \bibinfo{pages}{105895}.
\bibitem[{Wang and Lin(2025)}]{wang2025reduced}
\bibinfo{author}{Y.~Wang}, \bibinfo{author}{G.~Lin},
\newblock \bibinfo{title}{Reduced-basis deep operator learning for parametric
  pdes with independently varying boundary and source data},
\newblock \bibinfo{journal}{arXiv preprint arXiv:2511.18260}
  (\bibinfo{year}{2025}).
\bibitem[{Demo et~al.(2023)Demo, Tezzele, and Rozza}]{demo2023deeponet}
\bibinfo{author}{N.~Demo}, \bibinfo{author}{M.~Tezzele},
  \bibinfo{author}{G.~Rozza},
\newblock \bibinfo{title}{A deeponet multi-fidelity approach for residual
  learning in reduced order modeling},
\newblock \bibinfo{journal}{Advanced Modeling and Simulation in Engineering
  Sciences} \bibinfo{volume}{10} (\bibinfo{year}{2023}) \bibinfo{pages}{12}.
\bibitem[{Lu et~al.(2022)Lu, Pestourie, Johnson, and
  Romano}]{lu2022multifidelity}
\bibinfo{author}{L.~Lu}, \bibinfo{author}{R.~Pestourie}, \bibinfo{author}{S.~G.
  Johnson}, \bibinfo{author}{G.~Romano},
\newblock \bibinfo{title}{Multifidelity deep neural operators for efficient
  learning of partial differential equations with application to fast inverse
  design of nanoscale heat transport},
\newblock \bibinfo{journal}{Physical Review Research} \bibinfo{volume}{4}
  (\bibinfo{year}{2022}) \bibinfo{pages}{023210}.
\bibitem[{Wang and Li(2024)}]{wang2023dispersion}
\bibinfo{author}{Y.~Wang}, \bibinfo{author}{G.~Li},
\newblock \bibinfo{title}{Dispersion relation reconstruction for 2d photonic
  crystals based on polynomial interpolation},
\newblock \bibinfo{journal}{Journal of Computational Physics}
  \bibinfo{volume}{498} (\bibinfo{year}{2024}) \bibinfo{pages}{112659}.
\bibitem[{Wang and Li(2023)}]{wang2023hp}
\bibinfo{author}{Y.~Wang}, \bibinfo{author}{G.~Li},
\newblock \bibinfo{title}{An hp-adaptive sampling algorithm on dispersion
  relation reconstruction for 2d photonic crystals},
\newblock \bibinfo{journal}{arXiv preprint arXiv:2311.16454}
  (\bibinfo{year}{2023}).
\bibitem[{Bao et~al.(2001)Bao, Cowsar, and Masters}]{bao2001mathematical}
\bibinfo{author}{G.~Bao}, \bibinfo{author}{L.~Cowsar},
  \bibinfo{author}{W.~Masters}, \bibinfo{title}{Mathematical modeling in
  optical science}, \bibinfo{publisher}{SIAM}, \bibinfo{year}{2001}.
\bibitem[{Jackson(1999)}]{jackson1999classical}
\bibinfo{author}{J.~D. Jackson}, \bibinfo{title}{Classical electrodynamics},
  \bibinfo{year}{1999}.
\bibitem[{Kittel and McEuen(2018)}]{kittel2018introduction}
\bibinfo{author}{C.~Kittel}, \bibinfo{author}{P.~McEuen},
  \bibinfo{title}{Introduction to solid state physics},
  \bibinfo{publisher}{John Wiley \& Sons}, \bibinfo{year}{2018}.
\bibitem[{Glazman(1965)}]{glazman1965direct}
\bibinfo{author}{I.~M. Glazman}, \bibinfo{title}{Direct methods of qualitative
  spectral analysis of singular differential operators}, volume
  \bibinfo{volume}{2146}, \bibinfo{publisher}{Israel Program for Scientific
  Translations}, \bibinfo{year}{1965}.
\bibitem[{Joannopoulos et~al.(1997)Joannopoulos, Villeneuve, and
  Fan}]{joannopoulos1997photonic}
\bibinfo{author}{J.~D. Joannopoulos}, \bibinfo{author}{P.~R. Villeneuve},
  \bibinfo{author}{S.~Fan},
\newblock \bibinfo{title}{Photonic crystals},
\newblock \bibinfo{journal}{Solid State Communications} \bibinfo{volume}{102}
  (\bibinfo{year}{1997}) \bibinfo{pages}{165--173}.
\bibitem[{Johnson and Joannopoulos(2001)}]{johnson2001photonic}
\bibinfo{author}{S.~G. Johnson}, \bibinfo{author}{J.~D. Joannopoulos},
  \bibinfo{title}{Photonic crystals: the road from theory to practice},
  \bibinfo{publisher}{Springer Science \& Business Media},
  \bibinfo{year}{2001}.
\bibitem[{Johnson et~al.(1999)Johnson, Fan, Villeneuve, Joannopoulos, and
  Kolodziejski}]{johnson1999guided}
\bibinfo{author}{S.~G. Johnson}, \bibinfo{author}{S.~Fan},
  \bibinfo{author}{P.~R. Villeneuve}, \bibinfo{author}{J.~D. Joannopoulos},
  \bibinfo{author}{L.~Kolodziejski},
\newblock \bibinfo{title}{Guided modes in photonic crystal slabs},
\newblock \bibinfo{journal}{Physical Review B} \bibinfo{volume}{60}
  (\bibinfo{year}{1999}) \bibinfo{pages}{5751}.
\bibitem[{Akahane et~al.(2003)Akahane, Asano, Song, and Noda}]{akahane2003high}
\bibinfo{author}{Y.~Akahane}, \bibinfo{author}{T.~Asano},
  \bibinfo{author}{B.-S. Song}, \bibinfo{author}{S.~Noda},
\newblock \bibinfo{title}{High-q photonic nanocavity in a two-dimensional
  photonic crystal},
\newblock \bibinfo{journal}{nature} \bibinfo{volume}{425}
  (\bibinfo{year}{2003}) \bibinfo{pages}{944--947}.
\bibitem[{Bouckaert et~al.(1936)Bouckaert, Smoluchowski, and
  Wigner}]{bouckaert1936theory}
\bibinfo{author}{L.~P. Bouckaert}, \bibinfo{author}{R.~Smoluchowski},
  \bibinfo{author}{E.~Wigner},
\newblock \bibinfo{title}{Theory of brillouin zones and symmetry properties of
  wave functions in crystals},
\newblock \bibinfo{journal}{Physical Review} \bibinfo{volume}{50}
  (\bibinfo{year}{1936}) \bibinfo{pages}{58}.
\bibitem[{Schattschneider(1978)}]{schattschneider1978plane}
\bibinfo{author}{D.~Schattschneider},
\newblock \bibinfo{title}{The plane symmetry groups: their recognition and
  notation},
\newblock \bibinfo{journal}{The American Mathematical Monthly}
  \bibinfo{volume}{85} (\bibinfo{year}{1978}) \bibinfo{pages}{439--450}.
\bibitem[{Castell{\'o}-Lurbe et~al.(2014)Castell{\'o}-Lurbe, Torres-Company,
  and Silvestre}]{castello2014inverse}
\bibinfo{author}{D.~Castell{\'o}-Lurbe}, \bibinfo{author}{V.~Torres-Company},
  \bibinfo{author}{E.~Silvestre},
\newblock \bibinfo{title}{Inverse dispersion engineering in silicon
  waveguides},
\newblock \bibinfo{journal}{Journal of the Optical Society of America B}
  \bibinfo{volume}{31} (\bibinfo{year}{2014}) \bibinfo{pages}{1829--1835}.
\bibitem[{Dobson and Cox(1999)}]{dobson1999maximizing}
\bibinfo{author}{D.~C. Dobson}, \bibinfo{author}{S.~J. Cox},
\newblock \bibinfo{title}{Maximizing band gaps in two-dimensional photonic
  crystals},
\newblock \bibinfo{journal}{SIAM Journal on Applied Mathematics}
  \bibinfo{volume}{59} (\bibinfo{year}{1999}) \bibinfo{pages}{2108--2120}.
\bibitem[{Golub and Van~Loan(2013)}]{golub2013matrix}
\bibinfo{author}{G.~H. Golub}, \bibinfo{author}{C.~F. Van~Loan},
  \bibinfo{title}{Matrix computations}, \bibinfo{publisher}{JHU press},
  \bibinfo{year}{2013}.
\bibitem[{Eckart and Young(1936)}]{eckart1936approximation}
\bibinfo{author}{C.~Eckart}, \bibinfo{author}{G.~Young},
\newblock \bibinfo{title}{The approximation of one matrix by another of lower
  rank},
\newblock \bibinfo{journal}{Psychometrika} \bibinfo{volume}{1}
  (\bibinfo{year}{1936}) \bibinfo{pages}{211--218}.
\bibitem[{Bhatia(2013)}]{bhatia2013matrix}
\bibinfo{author}{R.~Bhatia}, \bibinfo{title}{Matrix analysis}, volume
  \bibinfo{volume}{169}, \bibinfo{publisher}{Springer Science \& Business
  Media}, \bibinfo{year}{2013}.
\bibitem[{Cybenko(1989)}]{cybenko1989approximation}
\bibinfo{author}{G.~Cybenko},
\newblock \bibinfo{title}{Approximation by superpositions of a sigmoidal
  function},
\newblock \bibinfo{journal}{Mathematics of control, signals and systems}
  \bibinfo{volume}{2} (\bibinfo{year}{1989}) \bibinfo{pages}{303--314}.
\bibitem[{Hornik(1991)}]{hornik1991approximation}
\bibinfo{author}{K.~Hornik},
\newblock \bibinfo{title}{Approximation capabilities of multilayer feedforward
  networks},
\newblock \bibinfo{journal}{Neural networks} \bibinfo{volume}{4}
  (\bibinfo{year}{1991}) \bibinfo{pages}{251--257}.

\end{thebibliography}

\end{document}